\documentclass[a4paper,UKenglish,cleveref, autoref, thm-restate]{lipics-v2021}
\usepackage{cite}

\usepackage[T1]{fontenc}
\usepackage{amsmath}
\usepackage[cmintegrals]{newtxmath}
\usepackage{bm}

\usepackage[utf8]{inputenc}
\usepackage{numprint}
\usepackage{graphicx}
\usepackage{stackrel}
\usepackage{mathrsfs}
\usepackage{wasysym}
\usepackage{tikz}
\usepackage[ruled, vlined, linesnumbered]{algorithm2e}
\usetikzlibrary{automata}
\usetikzlibrary{shapes}
\usepackage{hyperref}
\usepackage{comment}
\usepackage{colortbl}
\usepackage{array}
\usepackage{cases}
\usepackage{multirow}

\renewcommand{\Game}{\mathcal{G}}

\newcommand{\Mach}{\mathcal{M}}
\newcommand{\Kount}{\mathcal{K}}

\newcommand{\<}{\langle}
\renewcommand{\>}{\rangle}
\newcommand{\lv}{\lVert}
\newcommand{\rv}{\rVert}

\newcommand{\bsigma}{\bar{\sigma}}
\newcommand{\btau}{\bar{\tau}}

\renewcommand{\|}{\upharpoonright}

\newcommand{\WW}{\mathbb{W}}
\newcommand{\DD}{\mathbb{D}}
\newcommand{\NN}{\mathbb{N}}

\newcommand{\RR}{\mathbb{R}}
\newcommand{\QQ}{\mathbb{Q}}

\renewcommand{\SS}{\mathbb{S}}
\renewcommand{\AA}{\mathbb{A}}
\newcommand{\BB}{\mathbb{B}}

\newcommand{\LL}{\mathbb{L}}
\newcommand{\EE}{\mathbb{E}}

\newcommand{\playcircle}{{\scriptsize{\Circle}}}

\renewcommand{\l}{\ell}
\renewcommand{\epsilon}{\varepsilon}
\renewcommand{\phi}{\varphi}

\newcommand{\Sat}{\textsc{Sat}}
\newcommand{\coSat}{\textsc{coSat}}

\newcommand{\SubsetSum}{\textsc{SubsetSum}}

\newcommand{\subgameperfect}{\text{subgame-perfect}}
\newcommand{\Nash}{\text{Nash}}
\newcommand{\SPR}{\mathsf{SPR}}
\newcommand{\NR}{\mathsf{NR}}
\newcommand{\R}{\mathsf{R}}
\newcommand{\card}{\mathsf{card}}

\newcommand{\Plays}{\mathsf{Plays}}
\newcommand{\Hist}{\mathsf{Hist}}

\newcommand{\nego}{\mathsf{nego}}

\newcommand{\Occ}{\mathsf{Occ}}
\newcommand{\Inf}{\mathsf{Inf}}
\newcommand{\MP}{\mathsf{MP}}
\newcommand{\MPi}{\underline{\mathsf{MP}}}

\newcommand{\EL}{\mathsf{EL}}

\newcommand{\f}{\mathsf{f}}

\newcommand{\Comp}{\mathsf{Comp}}
\newcommand{\DS}{\mathsf{DS}}
\newcommand{\NDev}{\mathsf{NDev}}
\newcommand{\SPDev}{\mathsf{SPDev}}

\newcommand{\triangleone}{\blacktriangledown\!_1}

\newcommand{\Cl}{\mathcal{C}}

\newcommand{\Poly}{\mathsf{P}}
\newcommand{\NP}{\mathsf{NP}}
\newcommand{\ExpTime}{\mathsf{EXPTIME}}

\newcommand{\coNP}{\mathsf{coNP}}

\newcommand{\PSpace}{\mathsf{PSPACE}}

\newcommand{\bbx}{\bar{\bar{x}}}

\newcommand{\hC}{\hat{C}}

\newcommand{\stack}[2]{\stackrel{#1}{#2}}

\newtheorem{innercustomthm}{Theorem}
\newenvironment{customthm}[1]
{\renewcommand\theinnercustomthm{#1}\innercustomthm}
{\endinnercustomthm}
\newtheorem{innercustomlm}{Lemma}
\newenvironment{customlm}[1]
{\renewcommand\theinnercustomlm{#1}\innercustomlm}
{\endinnercustomlm}
\newtheorem{innercustomprop}{Proposition}
\newenvironment{customprop}[1]
{\renewcommand\theinnercustomprop{#1}\innercustomprop}
{\endinnercustomprop}

\newtheorem{pb}{Problem}

\title{Rational Verification and Checking for Nash and Subgame-perfect Equilibria in Graph Games}

\author{Léonard Brice}{Université Libre de Bruxelles, Belgium}{leonard.brice@ulb.be}{}{}

\author{Jean-François Raskin}{Université Libre de Bruxelles, Belgium}{jraskin@ulb.ac.be}{}{}

\author{Marie van den Bogaard}{Univ Gustave Eiffel, CNRS, LIGM, F-77454 Marne-la-Vallée, France}{marie.van-den-bogaard@univ-eiffel.fr}{}{}

\authorrunning{L. Brice, J.-F. Raskin, and M. van den Bogaard}

\Copyright{Léonard Brice, Jean-François Raskin, Marie van den Bogaard}



\nolinenumbers

\begin{document}

	\maketitle

\begin{abstract}
    We study two natural problems about rational behaviors in multiplayer non-zero-sum sequential infinite duration games played on graphs: checking problems, that consist in deciding whether a strategy profile, defined by a Mealy machine, is rational; and rational verification, that consists in deciding whether all the rational answers to a given strategy satisfy some specification.
    We give the complexities of those problems for two major concepts of rationality: Nash equilibria and subgame-perfect equilibria, and for five major classes of payoff functions: parity, quantitative reachability, energy, discounted-sum, and mean-payoff.
\end{abstract}

    \section{Introduction}

Formal methods are essential to guarantee the correctness of safety critical computer systems. Techniques like model-checking~\cite{DBLP:reference/mc/2018} or automated theorem proving~\cite{DBLP:journals/pacmpl/EbnerURAM17} are now routinely used to develop systematically hardware pieces as well as embedded control systems.
Nevertheless, there are contexts in which formal methods have not yet been applied successfully large-scale.
That is the case for instance of multi-agent systems, which still represent a challenge for formal verification techniques, because they are usually composed of {\em heterogeneous} components, ranging from traditional pieces of reactive code, to wholly autonomous robots or human users.
Producing operational model abstractions for this diversity of sub-systems can be challenging, or even impossible. 

While it may be inconvenient, to say the least, to produce an operational model of the behavior of a human or a complex autonomous robot, it may be easier to identify the high level objectives of those components.
And taking into account those objectives is often key for reasoning about the correctness of a system that interacts with those components.
Indeed, a system is usually not supposed to be correct in all circumstances, but only when agents in its environment behave in a way that concurs with their own objectives.
This is why we need to develop verification frameworks that allow us to reason on correctness in the presence of rational agents: agents whose behaviors are rational with regards to their high level objectives.
In \emph{rational verification}, a system needs to enforce a property $\phi$, not in all possible executions, but only in those executions in which agents of the environment behave rationally with regards to their own objectives. %

Rationality is the focus point of game theory and can be formalized in several ways. For instance, rational behavior for the agents can be modeled by the notion of \emph{Nash equilibrium} (NE)~\cite{Nas50} in a multiplayer non-zero sum game graph~\cite{DBLP:conf/fossacs/Ummels08}. NEs have been used in a few promising contributions, like in verification of non-repudiation and fair exchange protocols~\cite{DBLP:journals/jcs/KremerR03,DBLP:conf/concur/KremerR01,DBLP:conf/vmcai/ChatterjeeR12}, or planning of self-driving cars interacting with human drivers~\cite{DBLP:conf/rss/SadighSSD16}, etc. 
Nevertheless, those works do not propose a general framework for rational verification and their contributions are rather specific to the particular application domains that they consider. There is thus a need for more systematic study of formal frameworks for rational verification. Such a study has been started recently: for instance, the authors of \cite{GutierrezNPW20} study the automatic verification of an LTL specification in multi-agent systems that behave according to an NE, and in \cite{DBLP:conf/concur/BruyereRT22}, the authors study a setting in which the environment has multiple objectives and only produces behaviors that are Pareto-optimal with regards to those objectives. This work contributes to that line of research by considering a notion of rationality formalized by \emph{subgame-perfect equilibria} (SPEs), a refinement of NEs that is better suited to formalize rationality in sequential games, since NEs suffer from non-credible threats in such contexts (see e.g.~\cite{Osborne04}).

More precisely, we consider here two decision problems. 
First, in the {\em checking problems}, the inputs are: $(i)$ a multiplayer game graph, $(ii)$ a finite state description of a (potentially infinite) set of strategy profiles for the players in the game, and $(iii)$ a description of their objectives. The problem asks to check that all the strategy profiles in the set are NEs, or SPEs. This mathematical setting is well suited to formalize, for instance, that a high level description of a protocol, that contains nondeterminism, is such that all its implementations lead to rational behaviors of the entities participating to the protocol. This setting can be used to formalize elegantly the verification problems solved in~\cite{DBLP:journals/jcs/KremerR03,DBLP:conf/vmcai/ChatterjeeR12}, for instance. 
Second, the {\em rational verification problem} takes as inputs: $(i)$ a multiplayer game graph with a designated player called {\em Leader}, $(ii)$ a finite state description of a (potentially infinite) set of strategies for Leader, $(iii)$ a description of the objective for Leader, and $(iv)$ a description of the objectives of all the other players. It asks whether for all possible fixed strategies $\sigma_\LL$ of Leader (defined by the finite state description), for all possible rational responses of the other agents, the generated outcome satisfies Leader's objective. That problem is well-suited to formalize the verification of correctness of a controller interacting with an environment composed of rational agents, and intending to enforce a given property.

To solve those problems, we first provide two general constructions that reduce those problems to simpler ones. We show that they lead to algorithms that are computationally optimal for a large variety of classes of games with objectives ranging from Boolean $\omega$-regular objectives, like parity objectives, to quantitative ones, like mean-payoff objectives. Several lower complexity bounds require new constructions. We now detail our technical contributions before comparing our results with the existing results in the literature.

\subsection{Contributions}
To solve checking problems, we provide as a first preliminary result a general construction, called the \emph{deviation game} (Definition~\ref{def_deviation_games}): a game that simulates the parallel construction of a play compatible with the strategy profile, and of another play in which one player is deviating.
Thus, the checking problems reduce to the simpler problem of deciding, given a game, whether there exists a play in which some player gets a better payoff than some other one (Corollary~\ref{cor_privilege_problem}).
To the best of our knowledge, there is no general polynomial-time reduction in the other direction, hence the latter problem may be strictly harder than the former; but it turns out to be sufficient to prove that all the checking problems in our five classes of games can be solved with simple polynomial time graph algorithms (Theorems~\ref{thm_parity_checking}, \ref{thm_mp_checking}, \ref{thm_reach_checking}, \ref{thm_energy_det_ne_checking}, and \ref{thm_ds_checking}).
Interestingly, there is one exception: in energy games, the checking problems are $\coNP$-complete (Theorem~\ref{thm_energy_checking}) --- and those problems are closely related to the succinct one counter automaton reachability problem.

To solve the rational verification problems, we provide a general construction, called the \emph{product game} (Definition~\ref{defi_product_game}): we show that, given a game and a finite-state description of a set of Leader's strategies, one can incorporate the memory states of that finite-state description in the arena of the game in a way that Leader is implicitly forced to follow some strategy in the set.
Thus, we show that the rational verification problem reduces in polynomial time to the \emph{universal threshold problem}, a problem that is easier to study algorithmically: given a game, does every equilibrium satisfy a given specification? 
Also, some game classes we analyze have been addressed with slightly different definitions in previous literature. Interestingly, we provide a reduction in the opposite direction as well (Corollary 6).

We use that tool to prove the undecidability of rational verification in energy games (Theorems~\ref{thm_energy_nash_verif} and~\ref{thm_energy_sp_verif}); in the case of subgame-perfect rational verification, we show that undecidability holds even when Leader plays against only two players.
We show that Nash rational verification is co-recursively enumerable in those games, and leave that question open for subgame-perfect rational verification --- but contrary to the Nash setting, SPEs may require infinite memory to reach some payoffs (Proposition~\ref{pptn_energy_infinite_memory}).
In discounted-sum games, we show that the rational verification problems are at least as hard as the \emph{target discounted-sum problem} (Theorem~\ref{thm_ds_verif_hardness}), whose decidability is an open question.
However, we prove that those problems are recursively enumerable (Theorem~\ref{thm_ds_verif_easiness}).
In the case of mean-payoff games, Corollary~\ref{cor_reductions}, combined with older results, entails that the rational verification problems are $\coNP$-complete.
But that case highlights a subtlety in the definition of rational verification: if one wants to check that a strategy is such that \emph{every} rational response satisfies the specification, then when no such response exists, the strategy will be accepted.
In the case of mean-payoff games, that leads to results that can be considered as counter-intuitive.
We thus propose a stronger definition of the rational verification problem, called \emph{achaotic rational verification}, to avoid that weakness: it consists in deciding whether a strategy satisfies the specification against every response that is \emph{as rational as it can be}, using the notions of $\epsilon$-NE and $\epsilon$-SPE, that are quantitative relaxations of NE and SPE.
We show that such a problem is $\Poly^\NP$-complete in mean-payoff games (Theorem~\ref{thm_mp_ach_verif}), and that in every other setting (Nash or subgame-perfect rational verification in the two other game classes), it coincides with rational verification, since rational responses always exist (Proposition~\ref{prop_ach_verif}).
\begin{table*}
\centering
\scriptsize
	\begin{tabular}{c|c|c|c|c|c|c|c|c|c|c|c|c|}
		
		& \multicolumn{2}{c|}{Nash-checking} & \multicolumn{2}{c|}{SP-checking} & \multicolumn{2}{c|}{Nash RV} & \multicolumn{2}{c|}{Ach. Nash RV} & \multicolumn{2}{c|}{SP RV} & \multicolumn{2}{c|}{Ach. SP RV} \\
		\cline{2-13}
		& det. & n.-det. & det. & n.-det. & det. & n.-det. & det. & n.-det. & det. & n.-det. & det. & n.-det. \\
		\hline
		Parity & \multicolumn{4}{c|}{poly.} & \multicolumn{4}{c|}{$\coNP$-comp.} & \multicolumn{4}{c|}{$\coNP$-comp.} \\
		\hline
		QR & \multicolumn{4}{c|}{poly.} & \multicolumn{4}{c|}{$\coNP$-comp.} & \multicolumn{4}{c|}{$\PSpace$-comp.} \\
		\hline
		Energy & poly. & \multicolumn{3}{c|}{$\coNP$-comp. $(*)$} &
		\multicolumn{4}{c|}{undecidable, co-RE $(*)$} &
		\multicolumn{4}{c|}{undecidable $(*)$} \\
		\hline
		DS & \multicolumn{4}{c|}{poly.} & \multicolumn{4}{c|}{TDS-hard, RE $(*)$} & \multicolumn{4}{c|}{TDS-hard, RE $(*)$} \\
		\hline
		MP & \multicolumn{4}{c|}{poly.} & \multicolumn{4}{c|}{$\coNP$-comp.} & \multicolumn{2}{c|}{$\coNP$-comp.} & \multicolumn{2}{c|}{$\Poly^\NP$-comp. $(*)$} \\
		\hline
	\end{tabular}
    \caption{Synthesis of our results}
    \label{table_synthesis}
\end{table*}

A synthesis of those results can be found in Table~\ref{table_synthesis}: \emph{QR} is a short for \emph{quantitative reachability}, \emph{DS} for \emph{discounted-sum}, \emph{MP} for \emph{mean-payoff}, \emph{SP} for \emph{subgame-perfect}, and \emph{RV} for \emph{rational verification}.
Several of them are direct consequences of our general constructions, combined with results drawn from the literature; while some others required more effort and substantial original work.
We highlighted the latter with a star $(*)$ as they form the main technical contributions of our paper.

\subsection{Related works}

During the last decade, multiplayer games and their applications to reactive synthesis have raised a growing attention: the reader may refer to~\cite{DBLP:conf/lata/BrenguierCHPRRS16,DBLP:conf/dlt/Bruyere17,DBLP:journals/siglog/Bruyere21, DBLP:conf/tacas/FismanKL10, DBLP:journals/amai/KupfermanPV16}, and their references.
The concept of {\em rational verification} appears in~\cite{DBLP:journals/corr/abs-2207-02637}, where Gutierrez, Najib, Perelli, and Wooldridge give the complexity of several related problems.
They use a definition that is slightly different from ours: their problem consists in deciding, given a game and a specification, whether all NEs (or one of them) in that game satisfy the specification, without any player representing the system (Leader in our setting).
Still, as we show with Corollary~\ref{cor_reductions}, that problem is strongly related to ours.
In~\cite{DBLP:conf/atal/Steeples0W21}, they also study if $\omega$-regular properties are enforced by NEs induced by mean-payoff objectives.
The objectives considered in those papers are only $\omega$-regular objectives.
Moreover, both in~\cite{DBLP:journals/corr/abs-2207-02637} and in~\cite{DBLP:conf/atal/Steeples0W21} only NEs are considered, while our main contributions are about SPEs, that are arguably better suited for reasoning about sequential games~\cite{Osborne04}, but also require substantially more complex techniques.
In~\cite{filiot_et_al:LIPIcs:2020:12534}, Filiot, Gentilini, and Raskin study \emph{Stackelberg values} of mean-payoff and discounted-sum two-player non-zero sum games, i.e. the payoff that Leader gets when the other player, \emph{Follower}, plays the \emph{best response} that is available with regards to his own objective.
This is a synthesis problem while we consider a verification problem. They consider only one player in the environment while we consider the more general case of $n$ players.

In~\cite{Ummels05}, and later in~\cite{DBLP:conf/fossacs/Ummels08}, Ummels studies SPEs in parity games.
He proves that they always exist, and that deciding whether there exists an SPE in a given parity game that generates a payoff vector between two given thresholds (the \emph{constrained existence problem}, very close to the \emph{universal threshold problem} studied in this paper) is $\ExpTime$-easy and $\NP$-hard.
In~\cite{DBLP:conf/concur/BrihayeBGRB19}, Brihaye, Bruyère, Goeminne, Raskin, and van den Bogaard, study the same problem in quantitative reachability games, and prove that it is $\PSpace$-complete.

In~\cite{DBLP:journals/mor/FleschP17}, Flesch and Predtetchinski give a general procedure to characterize SPEs.
In~\cite{Concur}, Brice, Raskin, and van den Bogaard introduce the \emph{negotiation function}, a tool that turns Flesch and Predtetchinski's procedure into effective algorithms for a large class of games.
In~\cite{CSL}, they use it to close the gap left by Ummels, proving that the constrained existence problem is $\NP$-complete in parity games, with methods that they use later in~\cite{Icalp} to prove that the same problem is also $\NP$-complete in mean-payoff games.
An alternative procedure to solve such SPE problems is proposed in~\cite{thesis_noemie}, where Meunier constructs a two-player zero-sum game in which one player has a winning strategy if and only if there exists an SPE satisfying the desired constraint in the input game.
That technique is nevertheless often costly, because the size of the constructed game is proportional to the number of possible payoff vectors; and for the same reason, it cannot be applied to games with infinite payoff spaces.

Energy objectives have also been widely studied, in connection with the study of vector additions systems with states and Petri nets, but almost always in a two-player zero-sum setting: see for instance~\cite{DBLP:conf/formats/BouyerFLMS08, DBLP:journals/iandc/VelnerC0HRR15, DBLP:journals/amai/KupfermanPV16}.
As for discounted-sum objectives, they are defined for instance by Zwick and Paterson in~\cite{ZwickPaterson}, again in a two-player zero-sum setting.
They are strongly related to the target discounted-sum problem, which is a long-standing open problem, as shown in~\cite{DBLP:conf/lics/BokerHO15} by Boker, Henzinger, and Otop.
To the best of our knowledge, no algorithmic results are known for those classes of objectives in a multiplayer non-zero sum setting.

\subsection{Structure of the paper}
In Section~\ref{sec_background}, we introduce the necessary background.
In Section~\ref{sec_tools}, we present our two general constructions, the deviation game and the product game.
In Section~\ref{sec_parity}, we exploit it to study parity games; in Section~\ref{sec_qr}, quantitative reachability games; in~\ref{sec_energy}, energy games; in Section~\ref{sec_ds}, discounted-sum games; and in Section~\ref{sec_mp}, mean-payoff games.

	\section{Background} \label{sec_background}

\subsection{Graphs, games and strategies}

We call \emph{graph} a finite directed graph, i.e. a pair $(V, E)$ where $V$ is a finite set of \emph{vertices} and $E \subseteq V \times V$ is a set of \emph{edges}.
The edge $(u, v)$, written $uv$, is an \emph{outgoing edge} of $u$.
A \emph{path} in $(V, E)$ is a finite or infinite sequence $\alpha = \alpha_0 \alpha_1 \dots \in V^* \cup V^\omega$ such that for every index $k$, we have $\alpha_k \alpha_{k+1} \in E$.
We write $\Occ(\alpha)$ (resp. $\Inf(\alpha)$) for the set of vertices that occur (resp. that occur infinitely often) in $\alpha$.
For a given index $k$, we write $\alpha_{\leq k} = \alpha_{< k+1} = \alpha_0 \dots \alpha_k$, and $\alpha_{\geq k} = \alpha_{> k-1} = \alpha_k \alpha_{k+1} \dots$
A \emph{cycle} is a finite path $c = c_0 \dots c_n$ with $c_n c_0 \in E$.
A finite path $\alpha$ is \emph{simple} if for every two indices $k \neq \l$, we have $\alpha_k \neq \alpha_\l$.

	We call \emph{non-initialized game} a tuple
	$\Game = \left(\Pi, V, (V_i)_{i \in \Pi}, E, \mu\right)$, where:
		
		\begin{itemize}
			\item $\Pi$ is a finite set of \emph{players};
			
			\item $(V, E)$ is a graph, in which every vertex has at least one outgoing edge;
			
			\item $(V_i)_{i \in \Pi}$ is a partition of $V$, in which $V_i$ is the set of vertices \emph{controlled} by player $i$;
			
			\item a \emph{play} (resp. \emph{history}) in the game $\Game$ is an infinite (resp. finite) path in the graph $(V, E)$, and the set of plays (resp. histories) in $\Game$ is denoted by $\Plays \Game$ (resp. $\Hist \Game$);
			
			\item the \emph{payoff function} $\mu: \Plays \Game \to \RR^\Pi$ maps each play $\pi$ to the tuple $\mu(\pi) = (\mu_i(\pi))_{i \in \Pi}$.
		\end{itemize}

	Given a set of players $P \subseteq \Pi$, we often write $V_P = \bigcup_{i \in P} V_i$.
	When $i$ is a player and when the context is clear, we write $-i$ for the set $\Pi \setminus \{i\}$.
	We often assume that a special player, called \emph{Leader} and denoted by the symbol $\LL$, belongs to the set $\Pi$.
	An \emph{initialized game} is a pair $(\Game, v_0)$, often written $\Game_{\|v_0}$, where $\Game$ is a non-initialized game and $v_0 \in V$ is a vertex called \emph{initial vertex}.
	When the context is clear, we use the word \emph{game} for both initialized and non-initialized games.	
		A play (resp. history) in the initialized game $\Game_{\|v_0}$ is a play (resp. history) that has $v_0$ as first vertex.
		The set of plays (resp. histories) in $\Game_{\|v_0}$ is denoted by $\Plays \Game_{\|v_0}$ (resp. $\Hist \Game_{\|v_0}$).
		We also write $\Hist_i \Game$ (resp. $\Hist_i \Game_{\|v_0}$) for the set of histories in $\Game$ (resp. $\Game_{\|v_0}$) whose last vertex is controlled by player $i$.

		A \emph{strategy} for player $i$ in the initialized game $\Game_{\|v_0}$ is a mapping $\sigma_i: \Hist_i \Game_{\|v_0} \to V$, such that $v\sigma_i(hv)$ is an edge of $(V, E)$ for every $hv$.
		A history $h$ is \emph{compatible} with a strategy $\sigma_i$ if and only if $h_{k+1} = \sigma_i(h_0 \dots h_k)$ for all $k$ such that $h_k \in V_i$.
		This definition naturally extends to plays.
		A \emph{strategy profile} for $P \subseteq \Pi$ is a tuple $\bsigma_P = (\sigma_i)_{i \in P}$, where each $\sigma_i$ is a strategy for player $i$ in $\Game_{\|v_0}$.
		A play, or a history, is \emph{compatible} with $\bsigma_P$ if it is compatible with every $\sigma_i$ for $i \in P$.
  Since the $\sigma_i$'s domains are pairwise disjoint, we sometimes consider $\bsigma_P$ as one function: for $hv \in \Hist \Game_{\|v_0}$ such that $v \in \bigcup_{i \in P} V_i$, we liberally write $\bsigma_P(hv)$ for $\sigma_i(hv)$ with $i$ such that $v \in V_i$.
		A \emph{complete} strategy profile, usually written $\bsigma$, is a strategy profile for $\Pi$.
		Exactly one play is compatible with the strategy profile $\bsigma$: we call it its \emph{outcome} and write $\< \bsigma \>$ for it.
		When $\btau_P$ and $\btau'_Q$ are two strategy profiles with $P \cap Q = \emptyset$, we write $(\btau_P, \btau'_Q)$ for the strategy profile $\bsigma_{P \cup Q}$ such that $\sigma_i = \tau_i$ for $i \in P$, and $\sigma_i = \tau'_i$ for $i \in Q$.

\subsection{Notable classes of games} \label{sec_notable_classes}
 
Here, we will focus on five game classes.
The two first classes are studied mostly as examples: our results will easily come from the general constructions that we present in Section~\ref{sec_tools}, and from the existing literature.
The first class, parity games, is a class of \emph{Boolean games}, i.e. games in which all payoffs are equal either to $0$ or to $1$.
For such games, we say that player $i$ \emph{loses} the play $\pi$ when $\mu_i(\pi) = 0$, and \emph{wins} it when $\mu_i(\pi) = 1$.
The other games are called \emph{quantitative}.

Parity games are Boolean games in which each player wins a play if that play satisfies some \emph{parity condition}, a canonical encoding for $\omega$-regular conditions.

	\begin{definition}[Parity]
		The game $\Game$ is a \emph{parity game} if for each player $i$, there exists a mapping $\kappa_i: V \to \NN$, called \emph{color mapping}, such that for every play $\pi$, we have $\mu_i(\pi) = 1$ if the color $\min_{v \in \Inf(\pi)} \kappa_i(v)$ is even, and $\mu_i(\pi) = 0$ if it is odd.
	\end{definition}

Quantitative reachability games constitute a quantitative version of the classical class of (Boolean) reachability games: each player seeks to reach a given target as fast as possible.

\begin{definition}[Quantitative reachability]
    The game $\Game$ is a \emph{quantitative reachability game} if for each player $i$, there exists a \emph{target set} $T_i \subseteq V$, such that for every play $\pi$, we have:
    $$\mu_i(\pi) = \frac{1}{1 + \inf\{n \in \NN ~|~ \pi_n \in T_i\}},$$
    with the conventions $\inf \emptyset = +\infty$ and $\frac{1}{+\infty} = 0$.
\end{definition}

The three following game classes will require more substantial work.
In those classes, each player $i$'s payoff is based on a \emph{reward mapping} $r_i: E \to \QQ$.
Intuitively, the reward mapping gives the (positive or negative) reward that player $i$ gets for each action.
In energy games, the players seek to keep the aggregated sum of those rewards, their \emph{energy level}, always nonnegative.
That quantity symbolizes any resource that an agent could have to store: fuel, money, \dots

\begin{definition}[Energy]
    In a graph $(V, E)$, we associate to each reward mapping $r$ the \emph{energy level function} $\EL_r: \Hist \Game \to \NN \cup \{\bot\}$ defined by:
    \begin{itemize}
        \item $\EL_r(h_0) = 0$;
        \item $\EL_r(h_{\leq n+1}) = \EL_r(h_{\leq n}) + r(h_nh_{n+1})$ if $\EL_r(h_{\leq n}) \neq \bot$, and $\EL_r(h_{\leq n}) + r(h_nh_{n+1}) \geq 0$;
        \item $\EL_r(h_{\leq n+1}) = \bot$ otherwise.
    \end{itemize}

	The game $\Game$ is an \emph{energy game} if there exists a tuple $(r_i)_{i \in \Pi}$ of reward mappings such that for each $i$ and every $\pi$, we have $\mu_i(\pi) = 0$ if $\EL_{r_i}(\pi_{\leq n}) = \bot$ for some $n$, and $\mu_i(\pi) = 1$ otherwise.
	When the context is clear, we write $\EL_i$ for $\EL_{r_i}$.
\end{definition}

In discounted-sum games, each player's payoff is obtained by summing the rewards that the player obtains with some discount factor applied as the play goes along.

	\begin{definition}[Discounted-sum]
	    In a graph $(V, E)$, we define for each reward mapping $r$ and each \emph{discount factor} $\lambda \in (0, 1)$ the \emph{discounted sum function} $\DS_r^\lambda: h \mapsto \sum_k \lambda^k r(h_kh_{k+1})$.
		Then, we write
	    $\DS^\lambda_r (\pi) = \lim_n \DS^\lambda_r(\pi_{\leq n}).$
		The game $\Game$ is a \emph{discounted-sum game} if there exists a discount factor $\lambda \in (0, 1) \cap \QQ$ and a tuple $(r_i)_{i \in \Pi}$ of reward mappings such that for each $i$ and every $\pi$, we have $\mu_i(\pi) = \DS^\lambda_{r_i}(\pi)$.
		When the context is clear, we write $\DS_i$ for $\DS_{r_i}^\lambda$.
	\end{definition}

In mean-payoff games, a players' payoff is equal to their asymptotic average reward.
	
	\begin{definition}[Mean-payoff]
	    In a graph $(V, E)$, we define for each reward mapping $r$ the \emph{mean-payoff function} $\MP_r: h_0 \dots h_n \mapsto \frac{1}{n} \sum_k r\left(h_k h_{k+1}\right)$.
	    Then, we write
	    $\MPi_r (\pi) = \liminf_n \MP_r(\pi_{\leq n}).$
		The game $\Game$ is a \emph{mean-payoff game} if there exists a tuple $(r_i)_{i \in \Pi}$ of reward mappings, such that for each player $i$, we have $\mu_i = \MPi_{r_i}$.
		When the context is clear, we write $\MP_i$ for $\MP_{r_i}$, and $\MPi_i$ for $\MPi_{r_i}$.
	\end{definition}

Every game $\Game$ from one of those five classes can be encoded with a finite number of bits.
We write $\lv \Game \rv$ for that number.

An example of mean-payoff game is given in Figure~\ref{fig_ex_game}, with two players: player $\Circle$, who controls the vertices $a$ and $c$, and player $\Box$, who controls the vertex $b$.
The initial vertex is $v_0 = a$.
We wrote above each edge the rewards that both players get when that edge is taken.
Three types of plays are possible in that game: the one that loops on the vertex $a$ gives both players the payoff $0$; the ones that loop on the vertex $b$ give both players the payoff $1$; and the ones that loop on the vertex $c$ give both players the payoff $0$.

\begin{figure} 
			\centering
			\begin{tikzpicture}[->,>=latex,shorten >=1pt, initial text={}, scale=0.8, every node/.style={scale=0.7}]
				\node[state, initial left] (a) at (0, 0) {$a$};
				\node[state, rectangle] (b) at (2, 0) {$b$};
                \node[state] (c) at (4, 0) {$c$};
				\path (a) edge node[above] {$\stack{\playcircle}{0}\stack{\Box}{0}$} (b);
                \path (b) edge node[above] {$\stack{\playcircle}{0}\stack{\Box}{0}$} (c);
				\path (a) edge [loop above] node {$\stack{\playcircle}{0}\stack{\Box}{0}$} (a);
				\path (b) edge [loop above] node {$\stack{\playcircle}{1}\stack{\Box}{1}$} (b);
                \path (c) edge [loop above] node {$\stack{\playcircle}{0}\stack{\Box}{0}$} (c);
			\end{tikzpicture}
			\caption{An example of mean-payoff game}
			\label{fig_ex_game}
		\end{figure}
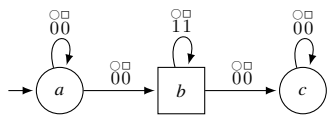

\subsection{Equilibria and rational responses}

In this paper, we study rational behaviors of players: we have, therefore, to define our rationality concepts.
Let us start with the most classical one: \emph{Nash equilibrium}.
		The strategy profile $\bsigma$ is a \emph{Nash equilibrium} (resp. {\em $\LL$-fixed Nash equilibrium}) --- or ($\LL$-fixed) \emph{NE} for short --- in $\Game_{\|v_0}$ if for each player $i$ (resp. each player $i \neq \LL$) and every strategy $\sigma'_i$, called \emph{deviation of $\sigma_i$}, we have $\mu_i\left(\< \sigma'_i, \bsigma_{-i} \>\right) \leq \mu_i\left(\< \bsigma \>\right)$.
		When it is not the case, we call \emph{profitable deviations} the deviations that do not satisfy that inequality.
	
As an example, in the game given in Figure~\ref{fig_ex_game}, two types of NEs can be found: those that eventually loop on the vertex $b$, and give both players the payoff $1$; and those that loop on $a$, but in which player $\Circle$ has no profitable deviation, because if she goes to the vertex $b$, player $\Box$ threatens to go to the vertex $c$ (and player $\Box$ has no profitable deviation, because he does never make any choice).
However, player $\Box$'s threat is not \emph{credible}, since going to the vertex $c$ would give him the payoff $0$, while he could stay on the vertex $b$ and get the payoff $1$.
A stronger rationality concept, that avoids that phenomenon, is the one of \emph{subgame-perfection}.

		Let $hv$ be a history in the game $\Game$. The \emph{subgame} of $\Game$ after $hv$ is the game $\Game_{\|hv} = \left(\Pi, V, (V_i)_i, E, \mu_{\|hv}\right)_{\|v}$, where $\mu_{\|hv}$ maps each play $\pi$ to its payoff in $\Game$, assuming that the history $hv$ has already been played, i.e. to the payoff $\mu_{\|hv}(\pi) = \mu(h\pi)$.
		If $\sigma_i$ is a strategy in $\Game_{\|v_0}$, its \emph{substrategy} after $hv$ is the strategy $\sigma_{i\|hv}: h' \mapsto \sigma_i(hh')$ in the game $\Game_{\|hv}$.

		The strategy profile $\bsigma$ is a \emph{($\LL$-fixed) subgame-perfect equilibrium} --- or \emph{($\LL$-fixed) SPE} for short --- in $\Game_{\|v_0}$ if and only if for every history $h$ in $\Game_{\|v_0}$ (resp. every history $h$ compatible with $\sigma_{\LL}$), the strategy profile $\bsigma_{\|h}$ is a ($\LL$-fixed) Nash equilibrium in the subgame $\Game_{\|h}$.

	NEs and SPEs entail two notions of rationality for the environment's responses to a strategy $\sigma_\LL$ of Leader.
A strategy profile $\bsigma_{-\LL}$ is a \emph{Nash response} to $\sigma_\LL$ if the strategy profile $\bsigma = (\sigma_\LL, \bsigma_{-\LL})$ is an $\LL$-fixed NE, and a \emph{subgame-perfect response} if it is an $\LL$-fixed SPE.
The set of Nash (resp. subgame-perfect) responses to $\sigma_\LL$ is written $\NR(\sigma_\LL)$ (resp. $\SPR(\sigma_\LL)$).

Finally, let $\rho \in \{\Nash, \subgameperfect\}$.
We call \emph{$\rho$-equilibria} the NEs if $\rho = \Nash$, and the SPEs if $\rho = \subgameperfect$.
We will similarly talk about \emph{$\LL$-fixed $\rho$-equilibria}, and \emph{$\rho$-responses}.
We write $\rho\R(\sigma_\LL)$ for the set of $\rho$-responses to a strategy $\sigma_\LL$.

\subsection{Mealy machines} \label{sec_def_Mealy_machine}
	
		A \emph{Mealy machine for player $i$} on a game $\Game$ is a tuple $\Mach = (Q, q_0, \Delta)$, where $Q$ is a finite set of \emph{states}, where $q_0 \in Q$ is the \emph{initial state}, and where $\Delta \subseteq (Q \times V_{-i} \times Q) \cup (Q \times V_i \times Q \times V)$ is a finite set of \emph{transitions}, such that for every $(p, u, q, v) \in \Delta$, we have $uv \in E$, and such that for every $p \in Q$ and $u \in V$, there exists a transition $(p, u, q)$ or $(p, u, q, v) \in \Delta$.
        Specialist readers will have noted that this definition is more general than the classical one, in which it is often assumed that for each $p$ and $u$, there exists exactly one such transition: hereafter, such a machine will be called \emph{deterministic}.
        Results about deterministic Mealy machines can be applied to \emph{programs}, which are supposed to run deterministically; we chose to take a more general definition to capture also \emph{protocols}, which may be given to an agent who would still have some room for manoeuvre in how they apply it.
		
		A strategy $\sigma_i$ in $\Game_{\|v_0}$ is \emph{compatible} with $\Mach$ if there exists a mapping $h \mapsto q_h$ that maps every history $h$ in $\Game_{\|v_0}$ to a state $q_h \in Q$, such that for every $hv \in \Hist_{-i} \Game_{\|v_0}$, we have $(q_h, v, q_{hv}) \in \Delta$, and for every $hv \in \Hist_i \Game_{\|v_0}$, we have $(q_h, v, q_{hv}, \sigma_i(hv)) \in \Delta$.
		The set of strategies in $\Game_{\|v_0}$ compatible with $\Mach$ is written $\Comp_{\|v_0}(\Mach)$.
		If $\Mach$ is deterministic, then there is exactly one strategy compatible with $\Mach$; we call it a \emph{finite-memory} strategy.

We define analogously Mealy machines that capture a set of strategy profiles for several players, including for the whole set $\Pi$.
Note also that every Mealy machine $\Mach$ can be encoded with a finite number of bits: we write $\lv \Mach \rv$ for that number.

Figure~\ref{fig_ex_1player_machine} depicts a one-player Mealy machine on the game of Figure~\ref{fig_ex_game}.
        Each arrow from a state $p$ to a state $q$ labeled $u|v$ denotes the existence of a transition $(p, u, q, v)$ (from the state $p$, the machine reads the vertex $u$, switches to the state $q$ and outputs the vertex $v$).
        Each arrow from a state $p$ to a state $q$ labeled $u$ denotes the existence of a transition $(p, u, q)$ (from $p$, the machine reads $u$, switches to $q$ and outputs nothing).
        It is a machine for player $\Box$, that is not deterministic: from the state $q_0$, reading the vertex $b$, the machine stays in $q_0$ but it can output either $b$ or $c$.
        The strategies that are compatible with it can be described as follows: when player $\Box$ has to play, if the vertex $a$ was seen an odd number of times, then he stays in $b$; in the opposite case, he can either stay in $b$ or eventually go to $c$.

Figure~\ref{fig_ex_multiplayer_machine} depicts a deterministic multiplayer Mealy machine on the same game.
The strategy profile that is compatible with it loops on the vertex $a$, and after a possible deviation of player $\Circle$ that would lead to the vertex $b$, loops once on $b$, before going to $c$.

		\begin{figure}
			\centering
			\begin{tikzpicture}[->,>=latex,shorten >=1pt, initial text={}, scale=0.8, every node/.style={scale=0.7}]
				\node[state, rectangle, rounded corners, initial above] (q0) at (0, 0) {$q_0$};
				\node[state, rectangle, rounded corners] (q1) at (3, 0) {$q_1$};
				
				\path (q0) edge[bend left=20] node[above] {$a$} (q1);
                \path (q1) edge[bend left=20] node[below] {$a$} (q0);
				\path (q0) edge[loop left] node[left] {$\begin{matrix} b|b\\ b|c\\ c \end{matrix}$} (q0);
                \path (q1) edge[loop right] node[right] {$\begin{matrix} b|b\\ c \end{matrix}$} (q1);
			\end{tikzpicture}
			\caption{A non-deterministic one-player Mealy machine} \label{fig_ex_1player_machine}
		\end{figure}
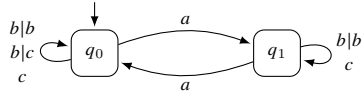
	\begin{figure}
			\centering
			\begin{tikzpicture}[->,>=latex,shorten >=1pt, initial text={}, scale=0.8, every node/.style={scale=0.7}]
				\node[state, rectangle, rounded corners, initial above] (q0) at (0, 0) {$q_0$};
				\node[state, rectangle, rounded corners] (q1) at (3, 0) {$q_1$};
				
				\path (q0) edge node[above] {$b|b$} (q1);
				\path (q0) edge[loop left] node[left] {$\begin{matrix} a|a\\ c|c \end{matrix}$} (q0);
                \path (q1) edge[loop right] node[right] {$\begin{matrix} a|a\\ b|c\\ c|c \end{matrix}$} (q1);
			\end{tikzpicture}
			\caption{A deterministic multiplayer Mealy machine} \label{fig_ex_multiplayer_machine}
		\end{figure}
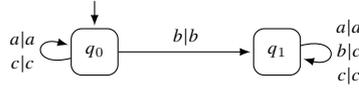

\subsection{Decision problems}

Let us now define checking and rational verification.
We define it for each game class $\Cl$, for each $\rho \in \{\Nash, \subgameperfect\}$, and in both the deterministic and the non-deterministic setting.

\begin{pb}[(Deterministic) $\rho$-checking problem in the class $\Cl$]
	Given a game $\Game_{\|v_0} \in \Cl$ and a (deterministic) Mealy machine $\Mach$ on $\Game$, is every $\bsigma \in \Comp_{\|v_0}(\Mach)$ a $\rho$-equilibrium?
\end{pb}
	
\begin{pb}[(Deterministic) $\rho$-rational verification problem in the class $\Cl$]
	Given a game $\Game_{\|v_0} \in \Cl$, a threshold $t \in \QQ$ and a (deterministic) Mealy machine $\Mach$ on $\Game$, is every $\LL$-fixed $\rho$-equilibrium $\bsigma$ with $\sigma_\LL \in \Comp_{\|v_0}(\Mach)$ such that $\mu_\LL(\< \bsigma \>) > t$?
\end{pb}

	    \subsection{A characterization of Nash equilibria}

Before moving on to our results, let us end this section with a lemma that will be used often along this paper, stating that an NE outcome is a play in which no player has a payoff smaller than what they can enforce when they deviate, and when all the other players make everything they can to punish them.
That is a classical result that can be found under various wordings (see for instance Theorem~1 in~\cite{Concur}); here, we present it under the form that will be the most useful for us hereafter.

\begin{lemma} \label{lm_ne}
    In every game $\Game$, the NE outcomes are exactly the plays $\pi$ such that for each player $i$, for every vertex $\pi_k \in V_i$, there exists a strategy profile $\btau^k_{-i}$ from $\pi_k$ such that $\sup_{\tau_i} \mu_i(\pi_{<k} \< \btau^k_{-i}, \tau_i \>) \leq \mu_i(\pi)$.
\end{lemma}

\begin{proof}
    Let $\pi$ be a Nash equilibrium outcome, and let $\bsigma$ be an NE such that $\< \bsigma \> = \pi$: then, for each $i$ and every $\pi_k$, the strategy profile $\bsigma_{\|\pi_{\leq k}}$ satisfies the hypothesis given above.
    Conversely, if such strategy profiles exist, then we can define $\bsigma$ as the strategy profile that follows the play $\pi$ and that, after a one-shot deviation $\pi_{\leq k} v$ with $v \neq \pi_{k+1}$, follows the strategy profile $\tau^k_{-i}$, and any strategy for player $i$.
\end{proof}
	
In the next section, we present two general constructions that we will use in the rest of the paper.

	\section{General constructions} \label{sec_tools}

Although very intuitive, the checking and rational verification problems are quite hard to study as they are.
Indeed, their instances include two graph structures: a game and a Mealy machine.
We therefore need preliminary results, that will reduce those problems to simpler ones.

	   \subsection{To solve checking problems: the deviation games}

Deciding the checking problems amounts to searching for a profitable deviation, either to the outcome (Nash-checking problem) or in some subgame (subgame-perfect-checking problem).
That can be achieved through a new game structure, called \emph{deviation game}, in which a play simulates two parallel plays in the original game: one in which the players have to follow the outputs of the Mealy machine, and one in which one of them is allowed to deviate from it.
In that game, two fresh players will measure the payoffs of the deviating player in the original game: one, called \emph{Adam}, is measuring the player's payoff in the non-deviating play, and the other one, called \emph{Eve}, is measuring their payoff in the deviating one.

	\begin{definition}[Deviation games] \label{def_deviation_games}
		Let $\Game_{\|v_0}$ be a game, and let $\Mach$ be a multiplayer Mealy machine in $\Game_{\|v_0}$.
		The associated \emph{Nash deviation game} is the game
		$\NDev(\Game, \Mach)_{\|(q_0, v_0)} = (\{\AA, \EE\}, V', (V'_\AA, V'_\EE), E', \mu')_{\|(q_0, v_0)}$,
		where:
		\begin{itemize}
		    \item the players are \emph{Adam}, written $\AA$, and \emph{Eve}, written $\EE$.
		
			\item The vertex space is $V' = \{(q_0, v_0)\} \cup (Q \times V \times \Pi) \cup (Q \times V \times \Pi \times Q \times V)$, and Eve controls every vertex.

        \item The set $E'$ contains:
        \begin{itemize}
            \item the edge $(q_0, v_0)(q, v, i)$ for each player $i \in \Pi$ and each transition $(q_0, v_0, q, v) \in \Delta$ (the player $i$ is chosen as the deviating player);

            \item the edge $(q_0, v_0)(q, v, i, q, v')$ for each player $i \in \Pi$, each transition $(q_0, v_0, q, v) \in \Delta$, and each edge $v_0v' \in E$ with $v' \neq v$ (the player $i$ is chosen as the deviating player, and starts to deviate immediately);

            \item the edge $(p, u, i)(q, v, i)$ for each $(p, u, q, v) \in \Delta$ (player $i$ has not started to deviate yet);

            \item the edge $(p, u, i)(q, v, i, q, v')$ for each transition $(p, u, q, v) \in \Delta$ with $u \in V_i$ and each edge $uv' \in E$ with $v' \neq v$ (player $i$ starts to deviate);

            \item the edge $(p, u, i, p', u')(q, v, i, q', v')$ for every two transitions $(p, u, q, v), (p', u', q', w) \in \Delta$ with either $w = v'$ or $u' \in V_i$ (player $i$ is deviating).
        \end{itemize}
			
			\item Let $\pi$ be a play in this game, of the form:
   \begin{align*}\pi &= (q_0, v_0) (q_1, v_1, i) \dots (q_{k-1}, v_{k-1}, i) \\ &(q_k, v_k, i, q'_k, v'_k) (q_{k+1}, v_{k+1}, i, q'_{k+1}, v'_{k+1}) \dots
   \end{align*}
			Then, we define $\mu'_\AA(\pi) = \mu(v_0 \dots v_k v_{k+1} \dots)$, and $\mu'_\EE(\pi) = \mu(v_0 \dots v_k v'_{k+1} \dots)$.
            When $\pi$ has the form $\pi = (q_0, v_0) (q_1, v_1, i) (q_2, v_2, i) \dots$, i.e. when player $i$ does never deviate, we define $\mu'_\AA(\pi) = \mu'_\EE(\pi) = \mu(v_0 v_1 v_2 \dots)$.
		\end{itemize}
		
		The \emph{subgame-perfect deviation game}
		$\SPDev(\Game, \Mach)_{\|(q_0, v_0)}$ is defined similarly with, additionnally:
        \begin{itemize}
            \item the edge $(p, u, i)(q, w, i)$ for each player $i \in \Pi$, each transition $(p, u, q, v) \in \Delta$ and each edge $uw \in E$ (player $i$ has not started to deviate, and Eve is looking for a subgame in which a profitable deviation exists);

            \item the edge $(q_0, v_0)(q, w, i)$ for each $(q_0, v_0, q, v) \in \Delta$ and every $v_0w \in E$ (player $i$ is chosen as the deviating player, and Eve is looking for a subgame).
        \end{itemize}
	\end{definition}

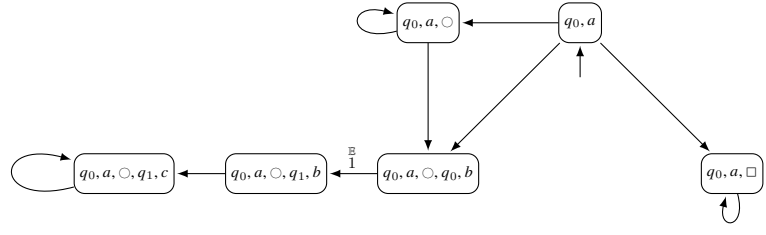
\begin{figure}
			\centering
			\begin{tikzpicture}[->,>=latex,shorten >=1pt, initial text={}, scale=1, every node/.style={scale=0.6}]
                \newcommand{\deltay}{1}
				\node[state, rectangle, rounded corners, initial below] (q0a) at (0, 0*\deltay) {$q_0, a$};
				\node[state, rectangle, rounded corners] (q0aC) at (-2, 0*\deltay) {$q_0, a, \Circle$};
				\node[state, rectangle, rounded corners] (q0aCq0b) at (-2, -2*\deltay) {$q_0, a, \Circle, q_0, b$};
				\node[state, rectangle, rounded corners] (q0aCq1b) at (-4, -2*\deltay) {$q_0, a, \Circle, q_1, b$};
				\node[state, rectangle, rounded corners] (q0aCq1c) at (-6, -2*\deltay) {$q_0, a, \Circle, q_1, c$};
				\node[state, rectangle, rounded corners] (q0aB) at (2, -2*\deltay) {$q_0, a, \Box$};
				
				\path (q0a) edge (q0aC);
                \path (q0aC) edge[loop left] (q0aC);
				\path (q0a) edge (q0aCq0b);
				\path (q0aC) edge (q0aCq0b);
				\path (q0aCq0b) edge node[above] {$\stack{\EE}{1}$} (q0aCq1b);
				\path (q0aCq1b) edge (q0aCq1c);
                \path (q0aCq1c) edge[loop left] (q0aCq1c);
				\path (q0a) edge (q0aB);
                \path (q0aB) edge[loop below] (q0aB);
			\end{tikzpicture}
			\caption{A Nash deviation game}
			\label{fig_ndev}
		\end{figure}
		
		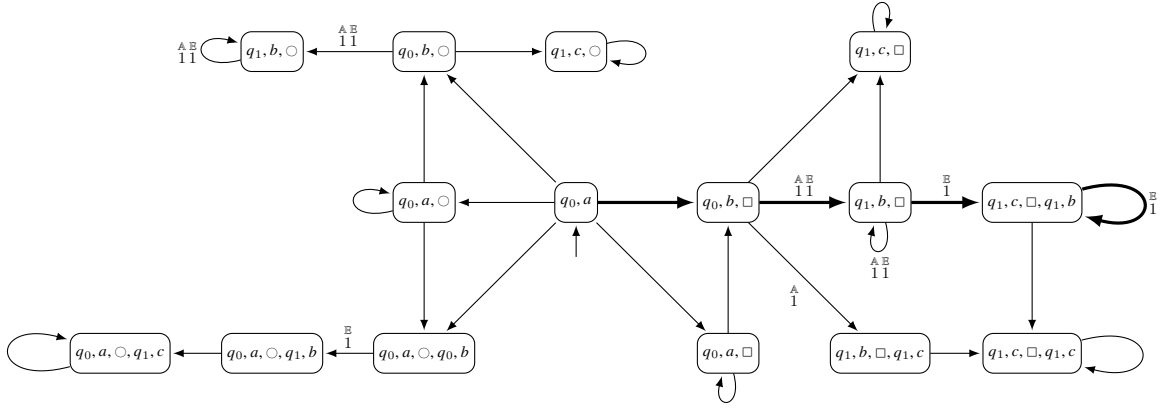
\begin{figure*}
			\centering
			\begin{tikzpicture}[->,>=latex,shorten >=1pt, initial text={}, scale=1, every node/.style={scale=0.6}]
                \newcommand{\deltay}{1}
				\node[state, rectangle, rounded corners, initial below] (q0a) at (0, 0*\deltay) {$q_0, a$};
				\node[state, rectangle, rounded corners] (q0aC) at (-2, 0*\deltay) {$q_0, a, \Circle$};
				\node[state, rectangle, rounded corners] (q0aCq0b) at (-2, -2*\deltay) {$q_0, a, \Circle, q_0, b$};
				\node[state, rectangle, rounded corners] (q0aCq1b) at (-4, -2*\deltay) {$q_0, a, \Circle, q_1, b$};
				\node[state, rectangle, rounded corners] (q0aCq1c) at (-6, -2*\deltay) {$q_0, a, \Circle, q_1, c$};
				\node[state, rectangle, rounded corners] (q0aB) at (2, -2*\deltay) {$q_0, a, \Box$};
    
				\node[state, rectangle, rounded corners] (q0bC) at (-2, 2*\deltay) {$q_0, b, \Circle$};
				\node[state, rectangle, rounded corners] (q1bC) at (-4, 2*\deltay) {$q_1, b, \Circle$};
				\node[state, rectangle, rounded corners] (q1cC) at (0, 2*\deltay) {$q_1, c, \Circle$};
				\node[state, rectangle, rounded corners] (q0bB) at (2, 0*\deltay) {$q_0, b, \Box$};
				\node[state, rectangle, rounded corners] (q1cB) at (4, 2*\deltay) {$q_1, c, \Box$};
				\node[state, rectangle, rounded corners] (q1bB) at (4, 0*\deltay) {$q_1, b, \Box$};
				\node[state, rectangle, rounded corners] (q1bBq1c) at (4, -2*\deltay) {$q_1, b, \Box, q_1, c$};
				\node[state, rectangle, rounded corners] (q1cBq1c) at (6, -2*\deltay) {$q_1, c, \Box, q_1, c$};
				\node[state, rectangle, rounded corners] (q1cBq1b) at (6, 0*\deltay) {$q_1, c, \Box, q_1, b$};

				\path (q0a) edge (q0aC);
                \path (q0aC) edge[loop left] (q0aC);
				\path (q0a) edge (q0aCq0b);
				\path (q0aC) edge (q0aCq0b);
				\path (q0aCq0b) edge node[above] {$\stack{\EE}{1}$} (q0aCq1b);
				\path (q0aCq1b) edge (q0aCq1c);
                \path (q0aCq1c) edge[loop left] (q0aCq1c);
				\path (q0a) edge (q0aB);
                \path (q0aB) edge[loop below] (q0aB);
                
				\path (q0a) edge (q0bC);
                \path (q0aC) edge (q0bC);
				\path (q0bC) edge node[above] {$\stack{\AA}{1} \,  \stack{\EE}{1}$} (q1bC);
    			\path (q1bC) edge[loop left] node[left] {$\stack{\AA}{1}\, \stack{\EE}{1}$} (q1bC);
        		\path (q0bC) edge (q1cC);
            	\path (q1cC) edge[loop right] (q1cC);
				\path[very thick] (q0a) edge (q0bB);
    \path (q0aB) edge (q0bB);
				\path (q0bB) edge (q1cB);
    \path (q1bB) edge (q1cB);
    \path[very thick] (q0bB) edge node[above] {$\stack{\AA}{1}\, \stack{\EE}{1}$} (q1bB);
    \path (q1cB) edge[loop above] (q1cB);
    \path (q0bB) edge node[below left] {$\stack{\AA}{1}$} (q1bBq1c);
    \path (q1bB) edge[loop below] node[below] {$\stack{\AA}{1}\, \stack{\EE}{1}$} (q1bB);
    \path (q1bBq1c) edge (q1cBq1c);
    \path (q1cBq1c) edge[loop right] (q1cBq1c);
    \path[very thick] (q1bB) edge node[above] {$\stack{\EE}{1}$} (q1cBq1b);
    \path[very thick] (q1cBq1b) edge[loop right] node[right] {$\stack{\EE}{1}$} (q1cBq1b);
    \path (q1cBq1b) edge (q1cBq1c);
			\end{tikzpicture}
			\caption{A subgame-perfect deviation game}
			\label{fig_spdev}
		\end{figure*}

As an example, if $\Game_{\|v_0}$ is the game of Figure~\ref{fig_ex_game} and if $\Mach$ is the machine given in Figure~\ref{fig_ex_multiplayer_machine}, then Figure~\ref{fig_ndev} represents the game $\NDev(\Game_{\|v_0}, \Mach)$, and Figure~\ref{fig_spdev} represents the game $\SPDev(\Game_{\|v_0}, \Mach)$.
For the sake of readability, the vertices that are not accessible from $(q_0, v_0)$ have been omitted.
These games are also mean-payoff games, and the rewards are given on the figures --- the rewards equal to $0$ have been omitted.
On Figure~\ref{fig_spdev}, the thick arrows highlight a play that gives a better payoff to Eve than to Adam; as we will see now, that means that the only strategy profile compatible with the machine $\Mach$ is not an SPE.
No such play can be found on the game of Figure~\ref{fig_ndev}, because that strategy profile is an NE.
In other words, the game $\Game_{\|v_0}$ and the machine $\Mach$ form a positive instance of the (deterministic) Nash checking problem, but not of the (deterministic) subgame-perfect checking problem.

	\begin{theorem}[App.~\ref{pf_deviation_game}] \label{thm_deviation_game}
		There exists a strategy profile $\bsigma \in \Comp_{\|v_0}(\Mach)$ that is \emph{not} an NE (resp. SPE) if and only if there exists a play $\pi$ in the game $\NDev(\Game, \Mach)_{\|(q_0, v_0)}$ (resp. $\SPDev(\Game, \Mach)_{\|(q_0, v_0)}$) such that $\mu_\AA(\pi) < \mu_\EE(\pi)$.
	\end{theorem}
	
In other words, the checking problems amount to solve the following problem in the deviation games.

\begin{pb}[Privilege problem in the class $\Cl$]
    Given a game $\Game_{\|v_0} \in \Cl$ with two players $\AA$ and $\EE$, called \emph{Adam} and \emph{Eve}, is every play $\pi$ in $\Game_{\|v_0}$ such that $\mu_\AA(\pi) \geq \mu_\EE(\pi)$?
\end{pb}

Moreover, the size of the deviations games are bounded by a polynomial function of $\lv \Game \rv$ and $\lv \Mach \rv$; and, when the game $\Game$ belongs to a class $\Cl$ among the five classes defined in Section~\ref{sec_notable_classes}, then all deviation games that can be constructed from it also belong to the class $\Cl$.
Hence the following.

\begin{corollary} \label{cor_privilege_problem}
    Let $\Cl$ be a class of games, among the classes of parity games, mean-payoff games, quantitative reachability games, energy games, and discounted-sum games.
    Then, in the class $\Cl$, the Nash-checking and the subgame-perfect checking problems, deterministic or not, reduce to the privilege problem.
\end{corollary}

        \subsection{To solve rational verification problems: the product game}

Responding rationally to Leader's strategies that are compatible with $\Mach$ amounts to play rationally in a larger game, in which the machine $\Mach$ has been incorporated.

	\begin{definition}[Product game] \label{defi_product_game}
		Let $\Game_{\|v_0}$ be a game, and let $\Mach$ be a Mealy machine for Leader in $\Game$.
		Their \emph{product game} is the game $\Game_{\|v_0} \otimes \Mach = (\Pi \cup \{\DD\}, V', (V'_i)_i, E', \mu')_{\|(v_0, q_0)}$ where the player $\DD$, called \emph{Demon}, chooses how the machine $\Mach$ will run.
		Formally:
		\begin{itemize}
			\item $V' = (V \times Q) \cup (V \times Q \times Q)$;
			
			\item $V'_\LL = \emptyset$, $V'_i = V_i \times Q \times Q$ for every $i \in \Pi \setminus \{\LL\}$, and $V'_\DD = (V \times Q) \cup (V_\LL \times Q \times Q)$;
			
			\item the set $E'$ contains:
			\begin{itemize}
				\item the edge $(u, p)(u, p, q)$ for each $(p, u, q) \in \Delta$ (if $u \not\in V_\LL$), or $(p, u, q, v) \in \Delta$ (if $u \in V_\LL$);
				
				\item the edge $(u, p, q)(v, q)$ for each $(p, u, q, v) \in \Delta$ (if $u \in V_\LL$);
				
				\item the edge $(u, p, q)(v, q)$ for each $(p, u, q) \in \Delta$, and each $uv \in E$ (if $u \not\in V_\LL$);
			\end{itemize}
			
			\item each payoff function $\mu'_i$ maps every play $(\pi_0, q_0) (\pi_0, q_0, q_1) (\pi_1, q_1) \dots$ to the payoff $\mu_i(\pi_0\pi_1 \dots)$ if $i \neq \DD$, and to the payoff $0$ if $i = \DD$.
		\end{itemize}
	\end{definition}

\begin{figure}
\centering
\begin{tikzpicture}[->,>=latex,shorten >=1pt, initial text={}, scale=0.8, every node/.style={scale=0.65}]
\newcommand{\deltay}{0.75}

   
	\node[state, rectangle, dotted, rounded corners] (bq1) at (-2, 2*\deltay) {$b, q_1$};
	\node[state] (aq0q1) at (0, 2*\deltay) {$a, q_0, q_1$};
    \node[state, rectangle, dotted, rounded corners, initial above] (aq0) at (0, 4*\deltay) {$a, q_0$};
    \node[state, rectangle, dotted, rounded corners] (bq1q1) at (-2, 4*\deltay) {$b, q_1, q_1$};
    \node[state, rectangle, dotted, rounded corners] (aq1) at (2, 2*\deltay) {$a, q_1$};
    \node[state] (aq1q0) at (2, 4*\deltay) {$a, q_1, q_0$};
    \node[state, rectangle, dotted, rounded corners] (bq0q0) at (4, 2*\deltay) {$b, q_0, q_0$};
    \node[state, rectangle, dotted, rounded corners] (bq0) at (4, 4*\deltay) {$b, q_0$};
    \node[state, rectangle, dotted, rounded corners] (cq0) at (6, 2*\deltay) {$c, q_0$};
    \node[state, rectangle, dotted, rounded corners] (cq0q0) at (6, 4*\deltay) {$c, q_0, q_0$};
				
				\path (aq0) edge (aq0q1);
				\path[very thick] (aq0q1) edge (bq1);
				\path (bq1) edge[bend left] node[left] {$\stack{\playcircle}{1}\, \stack{\Box}{1}$} (bq1q1);
				\path (bq1q1) edge[bend left] node[right] {$\stack{\playcircle}{1}\, \stack{\Box}{1}$} (bq1);
				\path (aq0q1) edge (aq1);
				\path (aq1) edge (aq1q0);
				\path (aq1q0) edge (aq0);
				\path (aq1q0) edge (bq0);
				\path (bq0) edge[bend left] node[right] {$\stack{\playcircle}{1}\, \stack{\Box}{1}$} (bq0q0);
				\path (bq0q0) edge[bend left] node[left] {$\stack{\playcircle}{1}\, \stack{\Box}{1}$} (bq0);
				\path (bq0q0) edge (cq0);
				\path (cq0) edge[bend left] (cq0q0);
				\path (cq0q0) edge[bend left] (cq0);
			\end{tikzpicture}
			\caption{A product game}
			\label{fig_product_game}
		\end{figure}
		
		Figure~\ref{fig_product_game} depicts the game $\Game_{\|v_0} \otimes \Mach$, when $\Game_{\|v_0}$ is the game of Figure~\ref{fig_ex_game} and $\Mach$ the machine of Figure~\ref{fig_ex_1player_machine}.
		Leader is then assimilated to player $\Box$, and Demon's vertices are represented by dotted boxes.
The unreachable vertices have been omitted, and we have given only the non-zero rewards.
Since, from the vertex $(a, q_0, q_1)$, player $\Circle$ has always the possibility to go to the vertex $(b, q_1)$ and to get the payoff $1$, it can be shown that every NE and every SPE in that game gives player $\Box$ the payoff $1$.
As we will see now, that means that the strategies compatible with the machine $\Mach$ guarantee the payoff $1$ to player $\Box$ against Nash-rational or subgame-perfect rational responses, i.e. that $\Game_{\|v_0}, 1-\epsilon,$ and $\Mach$, for every $\epsilon > 0$, form a positive instance of the Nash and subgame-perfect rational verification problems.

	\begin{theorem}[App.~\ref{pf_product_game}] \label{thm_product_game}
		Let $\rho \in \{\Nash, \subgameperfect\}$.
		Let $\Game_{\|v_0}$ be a game, let $\Mach$ be a Mealy machine for Leader in $\Game$, and let $t \in \QQ$.
		Then, every $\rho$-response $\bsigma_{-\LL}$ to every strategy $\sigma_\LL \in \Comp_{\|v_0}(\Mach)$ satisfies $\mu_\LL(\< \bsigma \>) > t$ if and only if every $\rho$-equilibrium $\btau$ in the game $\Game_{\|v_0} \otimes \Mach$ satisfies $\mu'_\LL(\< \btau \>) > t$.
	\end{theorem}

Thus, solving the $\rho$-rational verification problem in the game $\Game_{\|v_0}$ amounts to solve the \emph{$\rho$-universal threshold problem} ($\rho$-UT problem) in $\Game_{\|v_0} \otimes \Mach$.

\begin{pb}[$\rho$-universal threshold problem in the class $\Cl$]
		Given a game $\Game_{\|v_0} \in \Cl$, a player $i \in \Pi$, and a threshold $t \in \QQ$, is every $\rho$-equilibrium $\bsigma$ in $\Game_{\|v_0}$ such that $\mu_i(\< \bsigma \>) > t$?
	\end{pb}

Moreover, the size of the product game is bounded by a polynomial function of $\lv \Game \rv$ and $\lv \Mach \rv$; and when the game $\Game$ belongs to a class $\Cl$ among the three classes defined in Section~\ref{sec_notable_classes}, then all product games constructed from it also belong to $\Cl$.
Hence the following.
	
	\begin{corollary} \label{cor_reductions}
		Let $\Cl$ be a game class among energy games, discounted-sum games, and mean-payoff games.
    Then, in the class $\Cl$, for a given $\rho \in \{\Nash, \subgameperfect\}$, the $\rho$-UT problem, the $\rho$-rational verification problem, and the deterministic $\rho$-rational verification problem are reducible to each other in polynomial time.
	\end{corollary}

	\begin{proof}
		
		\begin{itemize}
			\item \emph{The deterministic $\rho$-rational verification problem reduces to the $\rho$-rational verification problem,}
			because a non-deterministic Mealy machine is a Mealy machine.

			\item \emph{The $\rho$-UT problem reduces to the deterministic $\rho$-rational verification problem.}
			
			Let $\Game_{\|v_0}$, $i$ and $t$ form an instance of the $\rho$-UT problem.
			We define the game $\Game'_{\|v_0}$ as equal to the game $\Game_{\|v_0}$, where Leader has been added to the player set, but controls no vertex.
			We define $\mu_\LL = \mu_i$.
			If $\Game$ belongs to the class $\Cl$, so does $\Game'$.
			Let $\Mach$ be the one-state deterministic Mealy machine on $\Game'$ that never outputs anything.
			Then, a strategy profile $\bsigma$ in $\Game'_{\|v_0}$ is an $\LL$-fixed $\rho$-equilibrium, if and only if it is an $\LL$-fixed $\rho$-equilibrium with $\sigma_\LL \in \Comp_{\|v_0}(\Mach)$, if and only if the strategy profile $\bsigma_{-\LL}$ is a $\rho$-equilibrium in the game $\Game_{\|v_0}$.
			As a consequence $\Game_{\|v_0}$, $i$, and $t$ form a positive instance of the $\rho$-UT problem, if and only if $\Game'_{\|v_0}$, $\Mach$, and $t$ form a positive instance of the deterministic $\rho$-rational verification problem.
            Moreover, the latter can be constructed from the former in polynomial time.

			\item \emph{The $\rho$-rational verification problem reduces to the $\rho$-UT problem,} by Theorem~\ref{thm_product_game}, and since the product game $\Game_{\|v_0} \otimes \Mach$ can be constructed from $\Game_{\|v_0}$ and $\Mach$ in polynomial time. \hspace{1em plus 1fill}\qedhere
		\end{itemize}
	\end{proof}

	\section{Parity games} \label{sec_parity}

Let us now apply those general constructions to our first class of games, and the easiest to study: parity games.
 
	    \subsection{Checking problems}

By Corollary~\ref{cor_privilege_problem}, in parity games, the checking problems reduce to the privilege problem, which consists in finding a play in a given game that satisfies a parity condition and that falsify another.
That can be done in polynomial time.

	\begin{theorem} \label{thm_parity_checking}
		In the class of parity games, the Nash-checking and the subgame-perfect checking problems, deterministic or not, can be decided in polynomial time.
	\end{theorem}
	
	\begin{proof}
		By Corollary~\ref{cor_privilege_problem}, those four problems reduce to the privilege problem.
		Let $\Game_{\|v_0}$ be a parity game in which there exists a play $\pi$ such that $\mu_\AA(\pi) < \mu_\EE(\pi)$.
		Then, we have $\mu_\AA(\pi) = 0$, and $\mu_\EE(\pi) = 1$.
		Therefore, there exist two vertices $u, v \in \Inf(\pi)$, such that $\kappa_\AA(u)$ is odd, $\kappa_\EE(v)$ is even, and there exists a path from $u$ to $v$ and a path from $v$ to $u$ that both traverse only vertices $w$ such that $\kappa_\AA(w) \geq \kappa_\AA(w)$ and $\kappa_\EE(w) \geq \kappa_\EE(w)$.
		Conversely, if such vertices $u$ and $v$ exist, then there exists a play $\pi$ satisfying $\mu_\AA(\pi) < \mu_\EE(\pi)$.
		The existence of such vertices can be checked in polynomial time.
	\end{proof}

	\subsection{Rational verification}

As for rational verification problems, they reduce by Corollary~\ref{cor_reductions} to UT problems, which are subproblems of problems already studied in~\cite{Ummels05}, \cite{DBLP:conf/fossacs/Ummels08}, and later in~\cite{CSL}.
In a nutshell, those problems belong to the class $\coNP$, because when there exists a NE or SPE outcome that makes some player $i$ lose, there exists one that has a simple form, and that can be guessed in polynomial time.
The lower bounds can be obtained by a slight adaptation on a reduction from $\coSat$ that was already presented in~\cite{Ummels05}.
	
	\begin{theorem}[App.~\ref{pf_parity_verif}] \label{thm_parity_verif}
		In the class of parity games, the Nash rational and the subgame-perfect rational verification problems, deterministic or not, are $\coNP$-complete.
	\end{theorem}

	\section{Quantitative reachability} \label{sec_qr}

Again, Corollaries~\ref{cor_privilege_problem} and~\ref{cor_reductions} enable us to solve efficiently the problems we are interested in; readily in the case of checking problems and subgame-perfect rational verification, and with some further work in the case of Nash rational verification.

	\subsection{Checking problems}
	
	\begin{theorem} \label{thm_reach_checking}
		In quantitative reachability games, the Nash-checking and the subgame-perfect checking problems, deterministic or not, can be decided in polynomial time.
	\end{theorem}
	
	\begin{proof}
		By Corollary~\ref{cor_privilege_problem}, those four problems reduce to the privilege problem.
	    In a quantitative reachability game, a play in which Adam's payoff is strictly smaller than Eve's one is a play that reaches Eve's target set without traversing Adam's one.
	    The existence of such a play can be decided by polynomial time classical graph search algorithms.
\end{proof}

	\subsection{Subgame-perfect rational verification}

\begin{theorem} \label{thm_reach_sp_verif}
    In quantitative reachability games, the  subgame-perfect rational verification problem, deterministic or not, is $\PSpace$-complete.
\end{theorem}

\begin{proof}
    It has been proved in~\cite{DBLP:conf/concur/BrihayeBGRB19} that the complement of the UT problem was $\PSpace$-complete\footnote{As for parity games, the problem studied in that paper was more general, but the lower bound still holds by a slight adaptation of the proof.}.
    The result follows by Corollary~\ref{cor_reductions}.
\end{proof}

	\subsection{Nash rational verification}

Negative instances of the Nash rational verification problem can be recognized by guessing an NE outcome, and checking it using Lemma~\ref{lm_ne}; hence that problem is $\coNP$-easy.
The matching lower bound can be established by reduction from the problem $\coSat$.

\begin{theorem}[App.~\ref{pf_reach_nash_verif}] \label{thm_reach_nash_verif}
    In quantitative reachability games, the  Nash rational verification problem, deterministic or not, is $\coNP$-complete.
\end{theorem}

	\section{Energy games} \label{sec_energy}

Let us now tackle game classes that require additional new techniques: first energy objectives.

	\subsection{Checking problems}

Energy games are the only one of our five classes in which the checking problems cannot be solved in polynomial time (unless $\Poly = \NP$);
except the simplest of them, the deterministic Nash-checking problem. 
Indeed, in that case, only one play must be compared to potential profitable deviations: the outcome deterministically generated by the Mealy machine.

\begin{theorem}[App.~\ref{pf_energy_det_ne_checking}] \label{thm_energy_det_ne_checking}
    In energy games, the deterministic Nash-checking problem can be decided in polynomial time.
\end{theorem}

In the other cases, either the non-determinacy or the need to study subgames entails a $\coNP$ lower bound, that we can prove by reduction from the problem $\SubsetSum$.
The matching upper bound can be obtained by an algorithm that searches for configurations (a vertex, a state of the memory and an energy vector) from which a profitable deviation exists, and that is accessible from the initial configuration.
We can then use the fact that reachability in 1-dimensional vector addition systems with states is $\NP$-easy --- see~\cite{DBLP:conf/concur/HaaseKOW09}.

	\begin{theorem}[App.~\ref{pf_energy_checking}] \label{thm_energy_checking}
		In energy games, the Nash-checking, the subgame-perfect-checking, and the deterministic subgame-perfect-checking problems are $\coNP$-complete.
	\end{theorem}

\subsection{Nash rational verification}

Rational verification problems are undecidable in this class, as we will show by reduction from the halting problem of two-counter machines (the reader who is not familiar with those machines may refer to App.~\ref{app_two_counter_machines}).
However, Nash rational verification is co-recursively enumerable.

\begin{theorem}[App.~\ref{pf_energy_nash_verif}] \label{thm_energy_nash_verif}
    In energy games, the  Nash rational verification problem, deterministic or not, is undecidable and co-recursively enumerable.
\end{theorem}

\begin{proof}[Proof sketch]
    We prove here that the Nash UT problem is undecidable and co-recursively enumerable.
    The theorem will follow by Corollary~\ref{cor_reductions}.

\begin{itemize}
    \item \emph{Undecidability.}
    We show undecidability by reduction from the halting problem of a two-counter machine.
    Let $\Kount$ be a two-counter machine.
	We define an energy game $\Game_{\|q_0^1}$ with five players --- players $C_1^\top$, $C_1^\bot$, $C_2^\top$, $C_2^\bot$, and $\WW$, called \emph{Witness} --- by assembling the gadgets presented in Figure~\ref{fig_gadget} --- the rewards that are not presented are equal to $0$, and the players controlling relevant vertices are written in blue.
	Then, a play in $\Game_{\|v_0}$ that does not reach the vertex $\blacktriangle$ simulates a sequence of transitions of $\Kount$, that can be a valid run or not: at each step, the counter $C_i$ is captured by the energy level of player $C_i^\top$, always equal to the energy level of player $C_i^\bot$.
	For each counter $C_i$, the player $C_i^\bot$ will have a profitable deviation if that play fakes a test to $0$, by going to the vertex $\blacktriangle$; and the player $C_i^\top$ will lose, and therefore have a profitable deviation by staying in $q_0^i$ if it fakes a positive test.
 Thus, as shown in the complete version of this proof, every NE outcome in the game $\Game_{\|q_0^1}$ is won by Witness if and only if the machine $\Kount$ does not terminate.
    As a consequence, the halting problem of two-counter machines reduces to the Nash UT problem in energy games, which is therefore undecidable.
	
\begin{figure}
    \centering
    \begin{subfigure}[b]{0.24\textwidth}
			\centering
			\begin{tikzpicture}[->,>=latex,shorten >=1pt, initial text={}, scale=0.8, every node/.style={scale=0.7}]
				\node (1) at (3, 0) {};
				\node[state, initial left] (q0) at (0, 0) {$q_0^1$};
                \node[state] (q0') at (2, 0) {$q_0^2$};
                \node[blue] (pq0) at (0, -0.6) {$C_1^\top$};
                \node[blue] (pq0') at (2, -0.6) {$C_2^\top$};
				\path (q0) edge (q0');
                \path (q0') edge (1);
				\path (q0) edge[loop above] (q0);
                \path (q0') edge[loop above] (q0');
			\end{tikzpicture}
			\caption{Initial state}
			\label{fig_gadget_initial}
		\end{subfigure}
		\begin{subfigure}[b]{0.2\textwidth}
			\centering
			\begin{tikzpicture}[->,>=latex,shorten >=1pt, initial text={}, scale=0.8, every node/.style={scale=0.7}]
				\node[state, initial above] (qf) at (0, 0) {$q_\f$};
				\path (qf) edge[loop right] node[right] {$\stack{C_1^\bot}{-1}\,\stack{C_2^\bot}{-1}\,\stack{\WW}{-1}$} (qf);
			\end{tikzpicture}
			\caption{Final state}
			\label{fig_gadget_final}
		\end{subfigure}
		\begin{subfigure}[b]{0.2\textwidth}
			\centering
			\begin{tikzpicture}[->,>=latex,shorten >=1pt, initial text={}, scale=1, every node/.style={scale=0.8}]
				\node[state, initial left] (q) at (0, 0) {$q$};
				\node (2) at (1, 0) {};
				\path (q) edge node[above] {$\stack{C^\top}{1}\,\stack{C^\bot}{1}$} (2);
			\end{tikzpicture}
			\caption{Incrementations}
			\label{fig_gadget_incrementation}
		\end{subfigure}
	\begin{subfigure}[b]{0.3\textwidth}
		\centering
		\begin{tikzpicture}[->,>=latex,shorten >=1pt, initial text={}, scale=1, every node/.style={scale=0.8}]
            \newcommand{\deltay}{0.5}
			\node (1) at (-1, 0) {};
			\node[state] (q) at (0, 0) {$q$};
			\node (2) at (2, 0*\deltay) {(if $C > 0$)};
			\node[state] (q') at (0, 2*\deltay) {$q'$};
			\node (3) at (2, 2*\deltay) {(if $C = 0$)};
			\node[state] (s) at (-1.2, 2*\deltay) {$\blacktriangle$};
   \node[blue] (pq) at (-0.4, -0.4) {$C^\top$};
   \node[blue] (pq') at (-0.4, 2*\deltay-0.4) {$C^\bot$};
			
			\path (1) edge (q);
			\path (q) edge node[above] {$\stack{C^\top}{-1}\,\stack{C^\bot}{-1}$} (2);
			\path (q) edge (q');
			\path (q') edge (3);
			\path (q') edge node[above] {$\stack{C^\bot}{-1}$} (s);
			\path (s) edge [loop left] (s);
		\end{tikzpicture}
		\caption{Tests}
		\label{fig_gadget_test}
	\end{subfigure}
    \caption{Gadgets}
    \label{fig_gadget}
\end{figure}

    \item \emph{Co-recursive enumerability.}
    As shown in the complete version of this proof, in an energy game $\Game_{\|v_0}$, if there exists an NE that makes some player $i$ lose, then there exists a finite-memory one.
    Thus, a semi-algorithm that recognizes the negative instances of the UT problem consists in enumerating the finite-memory complete strategy profiles on $\Game_{\|v_0}$, and for each of them, to check (by diagonalization):
    \begin{itemize}
        \item whether it is an NE: that is decidable (in polynomial time), by Theorem~\ref{thm_energy_checking};

        \item whether it makes player $i$ lose: that is recursively enumerable, by constructing step by step its outcome and computing the energy levels on the fly.
    \end{itemize}
    We have a negative instance of the UT problem if and only if at least one finite-memory strategy profile satisfies those two conditions.
    The Nash UT problem is therefore co-recursively enumerable. \hspace{1em plus 1fill}\qedhere
\end{itemize}
\end{proof}

\subsection{Subgame-perfect rational verification}

In the subgame-perfect setting, the previous construction could also prove undecidability.
But we choose to present a refinement of it, that proves a stronger result.

	\begin{theorem}[App.~\ref{pf_energy_sp_verif}] \label{thm_energy_sp_verif}
		In energy games, the subgame-perfect rational verification problem, deterministic or not, is undecidable, even when Leader plays against only two players.
	\end{theorem}

Again, the proof shows that, in particular, that problem is not recursively enumerable in energy games.
It might still be the case that it is co-recursively enumerable.
That would in particular be the case if finite memory was sufficient for an SPE to make any player $i$ lose, when that is possible, as in the case of NEs.
Unfortunately, one cannot follow this approach, because that statement is false: in order to be able to punish some player, without making another player lose, an SPE may have to memorize their energy levels, and therefore require infinite memory, as it will be the case in the example that follows.
We leave therefore the question open.

	\begin{proposition}[App.~\ref{pf_energy_infinite_memory}] \label{pptn_energy_infinite_memory}
		In the energy game presented in Figure~\ref{fig_energy_infinite_memory}, there exists an SPE that makes player $\Box$ lose, but no finite memory SPE can achieve that result.
	\end{proposition}

		\begin{figure}
			\centering
			\begin{tikzpicture}[->,>=latex,shorten >=1pt, initial text={}, scale=0.8, every node/.style={scale=0.7}]
				\node[state, initial left] (a) at (0, 0) {$a$};
				\node[state, rectangle] (b) at (2, 1) {$b$};
				\node[state, diamond] (c) at (2, -1) {$c$};
				\node[state] (d) at (4, 0) {$d$};
				\node[state] (e) at (6, 0) {$e$};
				\path (a) edge[bend right=20] (b);
				\path (b) edge[bend right=20] node[above left] {$\stackrel{\playcircle}{1}\stackrel{\Box}{1} \stackrel{\Diamond}{1}$} (a);
				\path (a) edge[bend left=20] (c);
				\path (c) edge[bend left=20] node[below left] {$\stackrel{\playcircle}{1}\stackrel{\Box}{1} \stackrel{\Diamond}{1}$} (a);
				\path (b) edge node[above] {$\stackrel{\playcircle}{1}$} (d);
				\path (c) edge node[below] {$\stackrel{\playcircle}{1}$} (d);
				\path (d) edge [loop above] node {$\stackrel{\playcircle}{-1} \, \stackrel{\Box}{-1} \, \stackrel{\Diamond}{-1}$} (d);
				\path (d) edge (e);
				\path (e) edge [loop right] (e);
			\end{tikzpicture}
			\caption{A game where infinite memory is necessary to make player $\Box$ lose}
			\label{fig_energy_infinite_memory}
		\end{figure}
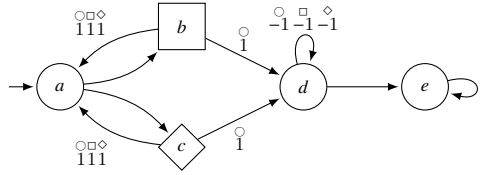

		\section{Discounted-sum games} \label{sec_ds}

	\subsection{Checking problems}

We will now move to discounted-sum objectives.
As in mean-payoff games, the privilege problem, and therefore the checking problems, can be solved in polynomial time by a Bellman-Ford-like algorithm.

	\begin{theorem}[App.~\ref{pf_ds_checking}] \label{thm_ds_checking}
		In discounted-sum games, the Nash-checking and the subgame-perfect checking problems, deterministic or not, can be decided in polynomial time.
	\end{theorem}

        \subsection{Rational verification}

Before moving to rational verification problems, let us define the following decision problem.

	\begin{pb}[Target discounted-sum problem]
		Given four quantities $\lambda, a, b, t \in \QQ$ with $0 < \lambda < 1$, is there a sequence $(u_n)_{n \in \NN} \in \{a, b\}^\omega$ such that $\sum_{n \in \NN} u_n \lambda^n = t$?
	\end{pb}
	
	Although it is a quite natural problem that appears in many different fields, the target discounted-sum (TDS) problem turns out to be surprisingly hard to solve, and its decidability status is still open.
 The interested reader may refer to~\cite{DBLP:conf/lics/BokerHO15} for more details.
	The following theorem shows that rational verification problems are at least as difficult.
	
	\begin{theorem} \label{thm_ds_verif_hardness}
		The TDS problem reduces to the complements of the (deterministic) Nash rational and subgame-perfect rational verification problems in discounted-sum games.
	\end{theorem}
	
	\begin{proof}
	    We present here a reduction to the complements of the Nash universal and subgame-perfect UT problems; the result follows by Corollary~\ref{cor_reductions}.
	    Let $a, b, t \in \QQ$, let $\lambda \in \QQ \cap (0, 1)$, and let $\Game_{\|v_0}$ be the discounted-sum game of Figure~\ref{fig_reduction_tds2}, with discount factor $\lambda$.
		In that game, there exists an NE $\bsigma$ with $\mu_\playcircle(\< \bsigma \>) < 0$, if and only if there exists an SPE $\bsigma$ with $\mu_\playcircle(\< \bsigma \>) < 0$, if and only if $a, b, t$, and $\lambda$ form a positive instance of the TDS problem.
	
		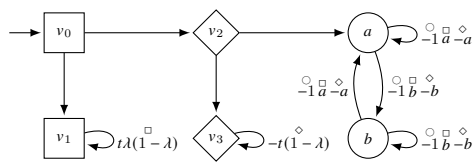
\begin{figure}
			\centering
			\begin{tikzpicture}[->,>=latex,shorten >=1pt, initial text={}, scale=1, every node/.style={scale=0.6}]
                \newcommand{\deltay}{0.7}
				\node[state, rectangle, initial left] (v0) at (0, 2*\deltay) {$v_0$};
				\node[state, rectangle] (v1) at (0, 0*\deltay) {$v_1$};
				\node[state, diamond] (v2) at (2, 2*\deltay) {$v_2$};
				\node[state, diamond] (v3) at (2, 0*\deltay) {$v_3$};
				\node[state] (a) at (4, 2*\deltay) {$a$};
				\node[state] (b) at (4, 0*\deltay) {$b$};
				\path (a) edge[bend left=20] node[right] {$\stack{\playcircle}{-1}  \, \stack{\Box}{b}  \, \stack{\Diamond}{-b}$} (b);
				\path (b) edge[bend left=20] node[left] {$\stack{\playcircle}{-1}  \, \stack{\Box}{a}  \, \stack{\Diamond}{-a}$} (a);
				\path (b) edge[loop right] node[right] {$\stack{\playcircle}{-1}  \, \stack{\Box}{b}  \, \stack{\Diamond}{-b}$} (b);
				\path (a) edge[loop right] node[right] {$\stack{\playcircle}{-1}  \, \stack{\Box}{a}  \, \stack{\Diamond}{-a}$} (a);
				\path (v0) edge (v1);
				\path (v0) edge (v2);
				\path (v2) edge (v3);
				\path (v2) edge (a);
				\path (v1) edge[loop right] node[right] {$\stack{\Box}{t \lambda (1-\lambda)}$} (v1);
				\path (v3) edge[loop right] node[right] {$\stack{\Diamond}{-t (1 - \lambda)}$} (v3);
			\end{tikzpicture}
			\caption{A game constructed from an instance of TDS}
			\label{fig_reduction_tds2}
		\end{figure}
		
		Indeed, if such an NE or SPE exists, it necessarily reaches the vertex $a$.
		But then, player $\Box$ must get at least the payoff $\mu_\Box(v_0v_1^\omega) = t \lambda^2$, and player $\Diamond$ the payoff $\mu_\Diamond(v_0v_2v_3^\omega) = -t \lambda^2$, otherwise they would have a profitable deviation.
		If such a play exists, then we have a positive instance of the TDS problem.
		Conversely, from such a positive instance, one can construct a play from $v_0$ in which player $\playcircle$ gets the payoff $\frac{\lambda^2}{1-\lambda}$, player $\Box$ the payoff $t \lambda^2$, and player $\Diamond$ the payoff $-t \lambda^2$, and none of them has a profitable deviation in any subgame. \hspace{1em plus 1fill}\qedhere
	\end{proof}
	
	The previous theorem suggests that finding algorithms solving those problems is a very ambitious objective.
	However, in the sequel, we will show that like the TDS problem, the rational verification problems are recursively enumerable.
The key idea is the following: a property of discounted-sum objectives is that when a play gives to some player a payoff that is strictly smaller than some threshold, that can be seen on finite prefixes of those plays.
Therefore, although strategy profiles are in general infinite objects that exist in uncountable number, profitable deviations can be found by analyzing their behaviors on a finite (but unbounded) number of histories.

	\begin{theorem}[App.~\ref{pf_ds_verif_easiness}] \label{thm_ds_verif_easiness}
		In discounted-sum games, the Nash rational and the subgame-perfect rational verification problems, deterministic or not, are recursively enumerable.
	\end{theorem}

	\section{Mean-payoff games} \label{sec_mp}

Let us now end with mean-payoff games.
For the checking as well as for the classical rational verification problems, our results can be derived from the general constructions of Section~\ref{sec_tools} and from the existing literature.
However, we will show that in mean-payoff games, there are examples that highlight a limit of our definition of rational verification, and that play in favour of a more subtle one, leading to distinct complexity results.

	\subsection{Checking problems}

As in parity games, Corollary~\ref{cor_privilege_problem} enables us to solve the checking problems in polynomial time, since the privilege problem itself reduces to the search of a negative cycle.

	\begin{theorem} \label{thm_mp_checking}
		In the class of mean-payoff games, the Nash-checking and the subgame-perfect checking problems, deterministic or not, can be decided in polynomial time.
	\end{theorem}
	
	\begin{proof}
		By Corollary~\ref{cor_privilege_problem}, those four problems reduce to the privilege problem.
		Let $\Game_{\|v_0}$ be a mean-payoff game.
		If every simple cycle $c$ accessible from $v_0$ is such that $\MP_\AA(c) \geq \MP_\EE(c)$, then we also have $\mu_\AA(\pi) \geq \mu_\EE(\pi)$ for every play $\pi$; conversely, if there is a simple cycle $c$ accessible from $v_0$ such that $\MP_\AA(c) < \MP_\EE(c)$, then there is a play $\pi$ in $\Game_{\|v_0}$ such that $\mu_\AA(\pi) < \mu_\EE(\pi)$ --- a play that reaches that cycle and loops there forever.
		Such a cycle can also be seen as a negative cycle for the reward function $r = r_\AA - r_\EE$.
		Note that the equality $\MPi_r(\pi) = \MPi_\AA(\pi) - \MPi_\EE(\pi)$ does not hold for every play $\pi$ (because the limit inferior and the additive inverse do not commute in general), but it does when $\pi$ has the form $c^\omega$.
		
		Thus, an algorithm that solves the privilege problem in polynomial time consists in searching for such a negative cycle, using Bellman-Ford's algorithm.
	\end{proof}

        \subsection{Classical rational verification}

The rational verification problems, as they are defined so far, can also be solved using already existing algorithms.

	\begin{theorem}[App.~\ref{pf_mp_verif}] \label{thm_mp_verif}
		In the class of mean-payoff games, the Nash rational and the subgame-perfect rational verification problems, deterministic or not, are $\coNP$-complete.
	\end{theorem}

        \subsection{The temptation of chaos} \label{sec_chaos}

It is now worth noting that the definition we gave of rational verification entails, in the case of mean-payoff games, results that may be considered as counter-intuitive.
For instance, consider the game of Figure~\ref{fig_chaos}, where Leader owns no vertex, and consider the only (vacuous) strategy available for Leader.
Does that strategy guarantee a payoff greater than $1$?
		Intuitively, it does not, since Leader always receives the payoff $0$.
		But still, that strategy, that game, and that threshold form a positive instance of subgame-perfect rational verification, because no $\LL$-fixed SPE exists in that game (see~\cite{Concur}).
More generally, the definition we give of rational verification considers that a \emph{good} strategy for Leader is a strategy such that for every response of the environment that is rational, the generated outcome observes some specification.
But a strategy is then good, in that sense, if \emph{no} rational response of the environment exists: that is the phenomenon that we can call \emph{temptation of chaos}.
While that case does never occur in energy and discounted-sum games, where rational responses are always guaranteed to exist (as we will see below), it must be considered in mean-payoff games.
	
		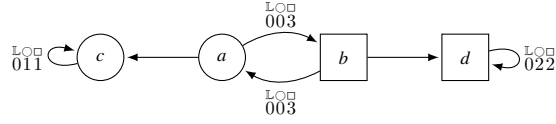
\begin{figure}
			\centering
			\begin{tikzpicture}[->,>=latex,shorten >=1pt, initial text={}, scale=0.8, every node/.style={scale=0.7}]
				\node[state] (a) at (0, 0) {$a$};
				\node[state] (c) at (-2, 0) {$c$};
				\node[state, rectangle] (b) at (2, 0) {$b$};
				\node[state, rectangle] (d) at (4, 0) {$d$};
				\path (a) edge (c);
				\path (a) edge[bend left] node[above] {$\stackrel{\LL}{0}\stackrel{\playcircle}{0} \stackrel{\Box}{3}$} (b);
				\path (b) edge[bend left] node[below] {$\stackrel{\LL}{0}\stackrel{\playcircle}{0} \stackrel{\Box}{3}$} (a);
				\path (b) edge (d);
				\path (d) edge [loop right] node {$\stackrel{\LL}{0}\stackrel{\playcircle}{2} \stackrel{\Box}{2}$} (d);
				\path (c) edge [loop left] node {$\stackrel{\LL}{0}\stackrel{\playcircle}{1} \stackrel{\Box}{1}$} (c);
			\end{tikzpicture}
			\caption{The temptation of chaos: an illustration}
			\label{fig_chaos}
		\end{figure}

\subsection{Achaotic rational verification}
 
	To avoid such phenomena, we introduce an alternative definition of rational verification, \emph{achaotic rational verification}: a good strategy for Leader will be a strategy that guarantees the given threshold against every response that is \emph{as rational as possible}.
	To define that problem, we need quantitative relaxations to the notions of NEs and SPEs.
		Let $\Game_{\|v_0}$ be a game.
  Let $\epsilon \geq 0$.
		The strategy profile $\bsigma$ is an \emph{$\epsilon$-NE} (resp. {\em $\LL$-fixed $\epsilon$-NE}) in $\Game_{\|v_0}$ if and only if for each $i \in \Pi$ (resp. $\Pi \setminus \{\LL\}$) and every deviation $\sigma'_i$ of $\sigma_i$, the inequality $\mu_i\left(\< \sigma'_i, \bsigma_{-i} \>\right) \leq \mu_i\left(\< \bsigma \>\right) + \epsilon$ holds: no deviation is profitable \emph{by more than $\epsilon$}.
  Note that $0$-NEs coincide with NEs.
We derive from that notion, as expected, the notions of ($\LL$-fixed) $\epsilon$-SPEs, $\epsilon$-Nash and $\epsilon$-subgame-perfect responses, and the notations $\epsilon\NR(\sigma_\LL)$, $\epsilon\SPR(\sigma_\LL)$, and $\epsilon\rho\R(\sigma_\LL)$.
We can now define our decision problem.

\begin{pb}[Achaotic (deterministic) $\rho$-rational verification in the class $\Cl$]
	Given a game $\Game_{\|v_0} \in \Cl$, a threshold $t \in \QQ$, and a Mealy machine (resp. a deterministic Mealy machine) $\Mach$ on $\Game$, does there exist $\epsilon \geq 0$ satisfying:
 \begin{itemize}
     \item $\epsilon\rho\R(\sigma_\LL) \neq \emptyset$ for some strategy $\sigma_\LL \in \Comp_{\|v_0}(\Mach)$;

     \item and $\mu_\LL(\< \sigma_\LL, \bsigma_{-\LL} \>) > t$ for every $\sigma_\LL \in \Comp_{\|v_0}(\Mach)$, and every $\bsigma_{-\LL} \in \epsilon\rho\R(\sigma_\LL)$?
 \end{itemize}
\end{pb}

We will prove below that in mean-payoff games, there exists a least quantity $\epsilon_{\min}$ such that $\epsilon_{\min}\rho$-responses to a given strategy $\sigma_\LL$ exist.
For instance, in the example depicted by Figure~\ref{fig_chaos}, we have $\epsilon_{\min} = 1$.
Thus, we can rephrase the achaotic rational verification problems as follows: given a game $\Game_{\|v_0}$, a threshold $t \in \QQ$ and a Mealy machine $\Mach$, do we have $\mu_\LL(\< \sigma_\LL, \bsigma_{-\LL} \>) > t$ for every $\sigma_\LL \in \Comp_{\|v_0}(\Mach)$ and every $\bsigma_{-\LL} \in \epsilon_{\min}\rho\R(\sigma_\LL)$?

Among the problems we study here, this new definition is relevant in only one case: subgame-perfect rational verification in mean-payoff games.
In all other cases, the rational verification problems are equivalent to their achaotic versions, because Nash and subgame-perfect responses are guaranteed to exist.
	
\begin{proposition}[App.~\ref{pf_ach_verif}] \label{prop_ach_verif}
    Let $\Cl$ be a class of games, among the classes of energy games and discounted-sum games.
	    Let $\rho \in \{\Nash, \subgameperfect\}$.
	    Then, the positive instances of the achaotic $\rho$-rational verification problem in $\Cl$ are exactly the positive instances of the $\rho$-rational verification problem.
    Similarly, the positive instances of the achaotic Nash-rational verification problem in mean-payoff games are exactly the positive instances of the $\rho$-rational verification problem.
\end{proposition}


Now, an optimal algorithm for that problem in mean-payoff games requires the following lemma: in each game, there exists a least $\epsilon$ such that $\epsilon$-SPEs exist, and it can be written with a polynomially bounded number of bits.
To prove that, we need to use the notion of \emph{negotiation function}, defined in~\cite{Concur}: a function from vertex labellings to vertex labellings whose least $\epsilon$-fixed point (i.e., the least vertex labelling $\lambda$ that is a fixed point of that function up to $\epsilon$) characterizes $\epsilon$-SPEs.
Our result can be obtained by revisiting a proof of~\cite{Icalp}, that was designed to bound the number of bits required to write that least $\epsilon$-fixed point, for a fixed $\epsilon$.
Hereafter, we write $\lv \epsilon \rv$ for the number of bits required to write $\epsilon$ in a usual encoding.

	\begin{lemma}[App.~\ref{pf_epsilon}] \label{lm_epsilon}
		There exists a polynomial $P_1$ such that in every mean-payoff game $\Game_{\|v_0}$, there exists $\epsilon_{\min}$ with $\lv \epsilon_{\min} \rv \leq P_1(\lv \Game \rv)$ such that $\epsilon_{\min}$-SPEs exist in $\Game_{\|v_0}$, and $\epsilon$-SPEs, for every $\epsilon < \epsilon_{\min}$, do not.
	\end{lemma}

We are now equipped to prove the following theorem.
 
	\begin{theorem}[App.~\ref{pf_mp_ach_verif}] \label{thm_mp_ach_verif}
		In the class of mean-payoff games, the achaotic subgame-perfect rational verification problem, deterministic or not, is $\Poly^\NP$-complete.
	\end{theorem}
	
	\begin{proof}[Proof sketch]
		Using Lemma~\ref{lm_epsilon} and the same arguments as in the proof of Theorem~\ref{thm_product_game}, those two problems are interreducible with the following one: given $\Game_{\|v_0}$ and a $t \in \QQ$, does every $\epsilon_{\min}$-SPE $\bsigma$ in $\Game_{\|v_0}$ satisfy $\mu_\LL(\< \bsigma \>) > t$?
		Let us prove $\Poly^\NP$-completeness for that problem.
		
		\begin{itemize}
			\item \emph{Easiness.}
						By~\cite{Icalp}, there is an $\NP$ algorithm deciding, given $\epsilon$ and $\Game_{\|v_0}$, whether there is an $\epsilon$-SPE in $\Game_{\|v_0}$, i.e. whether $\epsilon \geq \epsilon_{\min}$.
			Using Lemma~\ref{lm_epsilon} and the inequality $\epsilon_{\min} \leq 2 \max_{i, uv} |r_i(uv)|$, a dichotomous search can thus compute $\epsilon_{\min}$ using polynomially many calls to that algorithm.
			Then, one last call can decide whether there exists an $\epsilon_{\min}$-SPE $\bsigma$ such that $\mu_i(\< \bsigma \>) \leq t$.

			\begin{figure*}
				\centering
                \newcommand{\deltay}{0.7}
				\begin{tikzpicture}[->,>=latex,shorten >=1pt, initial text={}, scale=0.8, every node/.style={scale=0.57}]
					\node[state] (?x1) at (0, 0*\deltay) {$?x_1$};
					\node (x1) at (2, 1.5*\deltay) {\dots};
					\node (nx1) at (2, -1.5*\deltay) {\dots};
					\node[state] (?xi) at (4, 0*\deltay) {$?x_i$};
					\node[state] (xi) at (6, 1.5*\deltay) {$x_i$};
					\node[state] (nxi) at (6, -1.5*\deltay) {$\neg x_i$};
					\node (?xi+1) at (8, 0*\deltay) {\dots};
					\node[state] (xn) at (10, 1.5*\deltay) {$x_n$};
					\node[state] (nxn) at (10, -1.5*\deltay) {$\neg x_n$};
					\node[state] (C1) at (12, 0*\deltay) {$C_1$};
					\node (C2) at (14, 0*\deltay) {\dots};
					\node[state] (Cp) at (16, 0*\deltay) {$C_p$};
					\node[state, initial left] (a) at (0, -4*\deltay) {$a$};
					\node[state] (b) at (2, -4*\deltay) {$b$};
					\node[state] (c) at (4, -4*\deltay) {$c$};
					\node[state] (m1) at (6, -3.5*\deltay) {$\blacktriangledown$};
					\node[state, dashed] (m2) at (6, 3.5*\deltay) {$\blacktriangledown$};
					\node[state, dashed] (m3) at (14, -1.5*\deltay) {$\blacktriangledown$};
					
					\node[blue] (A) at (0, -4.6*\deltay) {$\AA$};
					\node[blue] (B) at (2, -3.4*\deltay) {$\BB$};
					\node[blue] (S1) at (-0.4, -0.4*\deltay) {$\SS$};
					\node[blue] (S2) at (3.6, -0.4*\deltay) {$\SS$};
					\node[blue] (pxi) at (6, 0.9*\deltay) {$x_i$};
					\node[blue] (pnxi) at (6, -0.9*\deltay) {$\neg x_i$};
					\node[blue] (pxn) at (10, 0.9*\deltay) {$x_n$};
					\node[blue] (pnxn) at (10, -0.9*\deltay) {$\neg x_n$};
					\node[blue] (pC1) at (12, -0.6*\deltay) {$C_1$};
					\node[blue] (pCp) at (16, -0.6*\deltay) {$C_p$};

					\path (a) edge[bend left] node[above] {$\stack{\AA}{0}~\stack{\BB}{3}~\stack{\WW}{1}$} (b);
					\path (b) edge[bend left] node[below] {$\stack{\AA}{0}~\stack{\BB}{3}~\stack{\WW}{1}$} (a);
					\path (b) edge (c);
					\path (c) edge[loop above] node[above] {$\stack{\AA}{2}~\stack{\BB}{2}~\stack{\WW}{1}$} (c);
					\path (a) edge (?x1);
					\path (?x1) edge (x1);
					\path (?x1) edge (nx1);
					\path (?xi) edge node[above left] {$\stack{x_i}{2m}~\stack{C}{m}~\stack{\AA}{2 - \frac{m}{2^{i+1}}}$} (xi);
					\path (?xi) edge node[below left] {$\stack{\neg x_i}{2m}~\stack{C}{m}~\stack{\AA}{2}$} (nxi);
					\path (xi) edge node[above right] {$\stack{x_i}{2m}~\stack{C}{m}~\stack{\AA}{2 - \frac{m}{2^{i+1}}}$} (?xi+1);
					\path (nxi) edge node[below right] {$\stack{\neg x_i}{2m}~\stack{C}{m}~\stack{\AA}{2}$} (?xi+1);
					\path (xn) edge node[above right] {$\stack{x_n}{2m}~\stack{C}{m}~\stack{\AA}{2 - \frac{m}{2^{n+1}}}$} (C1);
					\path (nxn) edge node[below right] {$\stack{\neg x_n}{2m}~\stack{C}{m}~\stack{\AA}{2}~\stack{\WW}{1}$} (C1);
					\path (C1) edge node[above] {$\stack{\AA}{2}$} (C2);
					\path (C2) edge node[above] {$\stack{\AA}{2}$} (Cp);
					\path (Cp) edge[bend right=40] node[above] {$\stack{\AA}{2}$} (?x1);
					\path (nxi) edge (m1);
					\path (nxn) edge[bend right=15] (m1);
					\path (xi) edge (m2);
					\path (xn) edge[bend left=15] (m2);
					\path (C1) edge[bend left=15] (m3);
					\path (Cp) edge[bend right=15] (m3);
					\path (m1) edge[loop below] node[below] {$\stack{\AA}{1}~\stack{x_1, \neg x_1, \dots, x_n, \neg x_n}{4}~ \stack{C_1, \dots, C_p}{2}~\stack{\WW}{1}$} (m1);
				\end{tikzpicture}
				\caption{The game $\Game_{\|a}$}
				\label{fig_G_pnp}
			\end{figure*}
			
			\item \emph{Hardness.}
			We proceed by reduction from the following $\Poly^\NP$-complete problem: given a Boolean formula $\phi$ in conjunctive normal form over the ordered variables $x_1, \dots, x_n$, is the lexicographically first valuation $\nu_{\min}$ satisfying $\phi$ such that $\nu_{\min}(x_n) = 1$? (and in particular, does such a valuation exist?)
			Let us write $\phi = \bigwedge_{j=1}^p C_j$.
			We construct a game $\Game_{\|a}$, with a player called \emph{Witness} and written $\WW$, in which there exists an $\epsilon_{\min}$-SPE $\bsigma$ such that $\mu_\WW(\< \bsigma\>) \leq 0$ if and only if $\phi$ is satisfiable and $\nu_{\min}(x_n) = 1$.
			That game, depicted in Figure~\ref{fig_G_pnp} (unmentioned rewards are equal to $0$, and we write $m = 2n+p$), has $2n + p + 4$ players: the literal players $x_1, \neg x_1, \dots, x_n, \neg x_n$; the clause players $C_1, \dots, C_p$; the player \emph{Solver}, written $\SS$; the player \emph{Witness}, written $\WW$; the player \emph{Alice}, written $\AA$; and the player \emph{Bob}, written $\BB$.

   This game is based on the classical example of mean-payoff game in which SPEs do not exist, already presented in Section~\ref{sec_chaos}.
   In the latter, from the vertex $a$, Alice can access a sink vertex, where Bob and her both get the payoff $1$.
   Here, they access instead to a region where the choices of Solver define a valuation of $x_1, \dots, x_n$ --- unless one of the literal players chooses to go to the sink vertex $\blacktriangledown$, which will be a profitable deviation if Solver makes inconsistent choices (one literal and, later, its negation).
   That valuation $\nu$ defines Alice's payoff $\mu_\AA(\pi) = 2 - \sum_{i=1}^n \frac{\nu(x_i)}{2^i}$, and therefore defines how much deviating and reaching $c$ is profitable for her.
   Consequently, as we show in the complete version of this proof, the valuation $\nu_{\min}$ is the binary encoding of the quantity $\epsilon_{\min}$, and there is an $\epsilon_{\min}$-SPE in which Witness gets the payoff $0$ or less if and only if $\nu_{\min}(x_n) = 1$. \hspace{1em plus 1fill}\qedhere
   \end{itemize}
\end{proof}

	\bibliography{bibli}

	\appendix

    \section{Proof of Theorem~\ref{thm_deviation_game}} \label{pf_deviation_game}

\begin{customthm}{\ref{thm_deviation_game}}
		There exists a strategy profile $\bsigma \in \Comp_{\|v_0}(\Mach)$ that is \emph{not} an NE (resp. SPE) if and only if there exists a play $\pi$ in the game $\NDev(\Game, \Mach)_{\|(q_0, v_0)}$ (resp. $\SPDev(\Game, \Mach)_{\|(q_0, v_0)}$) such that $\mu_\AA(\pi) < \mu_\EE(\pi)$.
	\end{customthm}

\begin{proof}
	    We present first the proof for SPEs.
	
		\begin{itemize}
			\item \emph{If there exists $\bsigma \in \Comp_{\|v_0}(\Mach)$ that is not an SPE, then there exists a play $\pi$ in $\SPDev(\Game, \Mach)_{\|(q_0, v_0)}$ such that $\mu_\AA(\pi) < \mu_\EE(\pi)$.}
			
			If $\bsigma$ is not an SPE, then there exists a history $hv = h_0 \dots h_{n-1} v$, a player $i$, and a play $\chi'$ compatible with $\bsigma_{-i\|hv}$ such that if $\chi = \< \bsigma_{\|hv} \>$, then we have $\mu_i(h\chi) < \mu_i(h\chi')$.
			Without loss of generality, we can assume $v \in V_i$ and $\chi_1 \neq \chi'_1$.
			Then, there exists a play:
			\begin{align*}
			    \pi &= (q_0, h_0) (q_1, h_1, i) \dots (q_n, v, i) (q_{n+1}, \chi_1, i, q'_{n+1}, \chi'_1) \\
			&(q_{n+2}, \chi_2, i, q'_{n+2}, \chi'_2) \dots
			\end{align*}
			in $\SPDev(\Game, \Mach)_{\|(q_0, v_0)}$ such that $\mu_\AA(\pi) < \mu_\EE(\pi)$.

			\item \emph{If there exists such a play in $\SPDev(\Game_{\|v_0}, \Mach)$, then there exists $\bsigma \in \Comp_{\|v_0}(\Mach)$ that is not an SPE.}
			
			If such a play exists, then it has the form:
			\begin{align*}
			    \pi &= (q_0, h_0)(q_1, h_1, i) \dots (q_{n-1}, h_{n-1}, i) (q_n, \chi_0, i) \\
			    & (q_{n+1}, \chi_1, i, q'_{n+1}, \chi'_1) (q_{n+2}, \chi_2, i, q'_{n+2}, \chi'_2) \dots
			\end{align*}
			where $\mu_i(h\chi) < \mu_i(h\chi')$ (with $\chi'_0 = \chi_0$).
			Then, there exists a strategy profile $\bsigma \in \Comp_{\|v_0}(\Mach)$ such that $\chi = \< \bsigma_{\|h \chi_0} \>$, and $\chi'$ is compatible with $\bsigma_{-i\|h \chi_0}$.
			That strategy profile is therefore not an SPE.
		\end{itemize}

        The proof for NEs follows the same structure, with histories $hv$ that are compatible with a strategy $\bsigma \in \Comp_{\|v_0}(\Mach)$. \hspace{1em plus 1fill}\qedhere
	\end{proof}

\section{Proof of Theorem~\ref{thm_product_game}} \label{pf_product_game}

\begin{customthm}{\ref{thm_product_game}}
		Let $\rho \in \{\Nash, \subgameperfect\}$.
		Let $\Game_{\|v_0}$ be a game, let $\Mach$ be a Mealy machine for Leader in $\Game$, and let $t \in \QQ$.
		Then, every $\rho$-response $\bsigma_{-\LL}$ to every strategy $\sigma_\LL \in \Comp_{\|v_0}(\Mach)$ satisfies $\mu_\LL(\< \bsigma \>) > t$ if and only if every $\rho$-equilibrium $\btau$ in the game $\Game_{\|v_0} \otimes \Mach$ satisfies $\mu'_\LL(\< \btau \>) > t$.
\end{customthm}
	
	\begin{proof}
    We present first the proof when $\rho = \subgameperfect$; again, the proof for $\rho = \Nash$ follows the same structure.
	    
		\begin{itemize}
			\item \emph{If for every $\sigma_\LL \in \Comp_{\|v_0}(\Mach)$ and every response $\bsigma_{-\LL} \in \SPR(\sigma_\LL)$, we have $\mu_\LL(\< \bsigma \>) > t$, then every SPE $\btau$ in $\Game_{\|v_0} \otimes \Mach$ satisfies $\mu'_\LL(\< \btau \>) > t$.}
			
			Let $\btau$ be an SPE in the game $\Game_{\|v_0} \otimes \Mach$.
			Let us define a strategy profile $\bsigma$ in $\Game_{\|v_0}$ as follows: for every history $h = h_0 \dots h_n$ in $\Game_{\|v_0}$, let $H = (h_0, q_0)(h_0, q_0, q_1) \dots (h_n, q_{n-1}, q_n)$ be the unique history of that form in $\Game_{\|v_0} \otimes \Mach$ such that $(h_k, q_k, q_{k+1}) = \tau_\DD((h_0, q_0) \dots (h_k, q_k))$ for each $k$, and let $(v, q_n) = \btau(H)$.
			Then, we define $\bsigma(h) = v$.
			
			Since the only edges available for Demon in the game $\Game_{\|v_0} \otimes \Mach$ are those that are compatible with $\Mach$, we have $\sigma_\LL \in \Comp_{\|v_0}(\Mach)$.
			Moreover, the strategy profile $\bsigma$ is an $\LL$-fixed SPE: let $h = h_0 \dots h_n$ be a history from $v_0$ that is compatible with $\sigma_\LL$, let $i \in \Pi \setminus \{\LL\}$ and let $\sigma'_i$ be a deviation of $\sigma_i$.
   We want to prove that $\mu_i(\< \bsigma_{-i\|h}, \sigma'_{i\|h} \>) \leq \mu_i(\< \bsigma_{\|h} \>)$.
			Then, let us define the history $H$ as above, and let $\tau'_i$ be the strategy that simulates $\sigma'_i$ in the game $\Game_{\|v_0} \otimes \Mach$, i.e. that maps each history of the form $H (v_1, q_n) (v_1, q_n, q_{n+1}) \dots (v_k, q_{n+k-1}, q_{n+k})$ to the vertex $(\sigma'_i(hv_1 \dots v_k), q_{n+k})$.
			Since the strategy profile $\btau$ is an SPE, we have $\mu'_i(H_{< 2n+1} \< \btau_{-i\|H}, \tau'_{i\|H} \>) \leq \mu'_i(H_{< 2n+1} \< \btau_{\|H} \>)$, and therefore $\mu_i(h_{<n} \< \bsigma_{-i\|h}, \sigma'_{i\|h} \>) \leq \mu_i(h_{<n} \< \bsigma_{\|h} \>)$.
			
			Thus, the strategy $\bsigma$ is an $\LL$-fixed SPE that satisfies $\sigma_\LL \in \Comp_{\|v_0}(\Mach)$.
			By hypothesis, it comes that we have $\mu_\LL(\< \bsigma \>) > t$, and therefore $\mu'_\LL(\< \btau \>) > t$.

			\item \emph{If every SPE $\btau$ in $\Game_{\|v_0} \otimes \Mach$ satisfies $\mu'_\LL(\< \btau \>) > t$, then for every $\sigma_\LL \in \Comp_{\|v_0}(\Mach)$ and every response $\bsigma_{-i} \in \SPR(\sigma_\LL)$, we have $\mu_\LL(\< \bsigma \>) > t$.}
			
			Indeed, let $\bsigma$ be an $\LL$-fixed SPE in $\Game_{\|v_0}$ with $\sigma_\LL \in \Comp_{\|v_0}(\Mach)$.
			We write $h \mapsto q_h$ for the mapping establishing the compatibility of the strategy $\sigma_\LL$ with the machine $\Mach$, as defined in Section~\ref{sec_def_Mealy_machine}.
			Let us define a strategy profile $\btau$ in $\Game_{\|v_0} \otimes \Mach$ as follows: for every history $H = (h_0, q_0) (h_0, q_0, q_1) \dots (h_n, q_{n-1})$, we define $\btau(H) = (h_n, q_{n-1}, q_{h_0 \dots h_n})$.
			For every history $H = (h_0, q_0) \dots (h_n, q_{n-1}, q_n)$, we define $\btau(H) = (\bsigma(h_0 \dots h_n), q_n)$.
			
			Then, the strategy profile $\btau$ is an SPE: let $H = H_0 \dots H_m$ be a history in $\Game_{\|v_0} \otimes \Mach$, let $i$ be a player and let $\tau'_i$ be a deviation of $\tau_i$.
			If $i = \LL$, then we have $\mu'_i(H_{<m} \< \btau_{-i\|H}, \tau'_{i\|H} \>) \leq \mu'_i(H_{<m} \< \btau_{\|H} \>)$, because Leader does not control any vertex in $\Game_{\|v_0} \otimes \Mach$, hence actually $\tau'_i = \tau_i$.
			Likewise if $i = \DD$, because Demon gets the payoff $0$ in every play.
			Now, if $i \neq \DD, \LL$, let us consider without loss of generality that $H$ has the form $H = (h_0, q_0) (h_0, q_0, q_1) \dots (h_n, q_n)$ (if the last vertex is controlled by Demon, it can be removed).
			Let $(\pi_0, q_n) (\pi_0, q_n, q_{n+1}) (\pi_1, q_{n+1}) \dots = \< \btau_{-i\|H}, \tau'_{i\|H} \>$.
			The play $\pi = \pi_0 \pi_1 \dots$ is compatible with the strategy profile $\bsigma_{-i\|h}$.
			Therefore, since $\bsigma$ is an $\LL$-fixed SPE, we have $\mu_i(h_{<n} \pi) \leq \mu_i(h_{<n} \< \bsigma_{\|h} \>)$, i.e. $\mu'_i(H_{<2n+1} \< \btau_{-i\|H}, \tau'_{i\|H} \>) \leq \mu'_i(H_{<2n+1} \< \btau_{\|H} \>)$.
			
			Thus, the strategy profile $\btau$ is an SPE: by hypothesis, it comes that $\mu'_\LL(\< \btau \>) > t$, and therefore $\mu_\LL(\< \bsigma \>) > t$. \hspace{1em plus 1fill}\qedhere
		\end{itemize}
	\end{proof}

\section{Proof of Theorem~\ref{thm_parity_verif}} \label{pf_parity_verif}

 	\begin{customthm}{\ref{thm_parity_verif}}
		In the class of parity games, the Nash rational and the subgame-perfect rational verification problems, deterministic or not, are $\coNP$-complete.
	\end{customthm}
	
	\begin{proof}
	    \begin{itemize}     
	        \item \emph{Subgame-perfect rational verification.}

            In~\cite{CSL}, a problem that was more general than the complement of the subgame-perfect universal threshold problem in parity games was proved to be $\NP$-complete.
            In our setting, the upper bound is therefore already established.
            As for the lower bound, we need to present a slight modification of the reduction that was given in the aforementioned paper, itself inspired from the construction provided in~\cite{Ummels05}.
            Let $\phi = \bigwedge_{j=1}^p C_j$ be a formula of the propositional logic, over the variables $x_1, \dots, x_n$.
		      We define the parity game $\Game^\phi_{\|C_1}$ as follows:
		
		\begin{itemize}
			\item the players are the variables $x_1, \dots, x_n$, their negations, the player \emph{Solver}, denoted by $\SS$, and the player \emph{Witness}, denoted by $\WW$;
			
			\item the vertices controlled by Solver are all the clauses $C_j$, and the sink vertex $\blacktriangledown$;
			
			\item the vertices controlled by player $\l = \pm x_i$ are the pairs $(C_j, L)$, where $\l$ is a literal of $C_j$;

        \item Witness controls no vertices;
			
			\item there are edges from each clause vertex $C_j$ to all the vertices $(C_j, \l)$; from each pair vertex $(C_j, \l)$ to the vertex $C_{j+1}$, and to the sink vertex $\blacktriangledown$; and from the sink vertex $\blacktriangledown$ to itself;
			
			\item for Solver, every vertex has the color $\kappa_\SS(v) = 2$, except the vertex $\blacktriangledown$, which has color $1$; for Witness, every vertex has the color $\kappa_\WW(v) = 1$, except the vertex $\blacktriangledown$, which has color $2$; for each literal player $\l$, every vertex has the color $\kappa_\l(v) = 2$, except the vertices of the form $(C, \bar{\l})$, that have the color $1$.
		\end{itemize}

    The arena of the game $\Game^\phi$, when $\phi$ is the tautology $\bigwedge_{j=1}^6 (x_j \vee \neg x_j)$, is given by Figure~\ref{fig_Gphi}.

  \begin{figure}
    \centering
      \begin{tikzpicture}[->,>=latex,shorten >=1pt, scale=0.4, every node/.style={scale=0.6}]
			\node[state] (C1) at (0:5) {$C_1$};
			\node[state] (C11) at (30:6) {$x_1$};
			\node[state] (C12) at (30:4) {$\neg x_1$};
			\node[state] (C2) at (60:5) {$C_2$};
			\node[state] (C21) at (90:6) {$x_2$};
			\node[state] (C22) at (90:4) {$\neg x_2$};
			\node[state] (C3) at (120:5) {$C_3$};
			\node[state] (C31) at (150:6) {$x_3$};
			\node[state] (C32) at (150:4) {$\neg x_3$};
			\node[state] (C4) at (180:5) {$C_4$};
			\node[state] (C41) at (210:6) {$x_4$};
			\node[state] (C42) at (210:4) {$\neg x_4$};
			\node[state] (C5) at (240:5) {$C_5$};
			\node[state] (C51) at (270:6) {$x_5$};
			\node[state] (C52) at (270:4) {$\neg x_5$};
			\node[state] (C6) at (300:5) {$C_6$};
			\node[state] (C61) at (330:6) {$x_6$};
			\node[state] (C62) at (330:4) {$\neg x_6$};
			\node[state] (b) at (0,0) {$\blacktriangledown$};
			
			\path[->] (C1) edge (C11);
			\path[->] (C1) edge (C12);
			\path[->] (C11) edge (C2);
			\path[->] (C12) edge (C2);
			\path[->] (C11) edge[bend left=30] (b);
			\path[->] (C12) edge (b);
			\path[->] (C2) edge (C21);
			\path[->] (C2) edge (C22);
			\path[->] (C21) edge (C3);
			\path[->] (C22) edge (C3);
			\path[->] (C21) edge[bend left=30] (b);
			\path[->] (C22) edge (b);
			\path[->] (C3) edge (C31);
			\path[->] (C3) edge (C32);
			\path[->] (C31) edge (C4);
			\path[->] (C32) edge (C4);
			\path[->] (C31) edge[bend left=30] (b);
			\path[->] (C32) edge (b);
			\path[->] (C4) edge (C41);
			\path[->] (C4) edge (C42);
			\path[->] (C41) edge (C5);
			\path[->] (C42) edge (C5);
			\path[->] (C41) edge[bend left=30] (b);
			\path[->] (C42) edge (b);
			\path[->] (C5) edge (C51);
			\path[->] (C5) edge (C52);
			\path[->] (C51) edge (C6);
			\path[->] (C52) edge (C6);
			\path[->] (C51) edge[bend left=30] (b);
			\path[->] (C52) edge (b);
			\path[->] (C6) edge (C61);
			\path[->] (C6) edge (C62);
			\path[->] (C61) edge (C1);
			\path[->] (C62) edge (C1);
			\path[->] (C61) edge[bend left=30] (b);
			\path[->] (C62) edge (b);
			\path (b) edge[loop above] (b);
		\end{tikzpicture}
		\caption{The game $\Game^\phi$.}
		\label{fig_Gphi}
  \end{figure}

            Intuitively: from the vertex $C_j$, Solver chooses which literal $\l$ of $C_j$ will be satisfied by the valuation she tries to construct.
            Then, from the vertex $(C_j, \l)$, player $\l$ may go to the vertex $\blacktriangledown$ and win; if she does not, then Solver has to make her win, which means that she cannot choose the literal $\bar{\l}$ for another clause (at least not infinitely often).
            Solver wins (and Witness loses) if and only if the vertex $\blacktriangledown$ is never reached.
            As proved in~\cite{CSL}, there exists an SPE that is won by Solver in this game if and only if the formula $\phi$ is satisfiable; that is, \emph{every} SPE is won by Witness if and only if $\phi$ is \emph{not} satisfiable.
            Consequently, the subgame-perfect universal threshold problem in parity game is $\coNP$-complete.
		    The $\coNP$-completeness of the subgame-perfect rational verification problem follows by Corollary~\ref{cor_reductions}.
		    
		    	        \item \emph{Nash rational verification.}
	        
	        By Corollary~\ref{cor_reductions}, the Nash rational verification problem in parity games is $\coNP$-complete if and only if the Nash universal threshold problem is.
	        
	        That problem is $\coNP$-hard by the same reduction as above: in the game $\Game^\phi$, we can show that every NE outcome is an SPE outcome (for example using Theorem~1 of~\cite{CSL}).
	        
	        By Lemma~\ref{lm_ne}, the NE outcomes in a parity game are exactly the plays $\pi$ such that, for each player $j$ such that $\mu_j(\pi) = 0$, for each vertex $v \in \Occ(\pi) \cap V_j$, there exists a strategy profile $\bsigma_{-j}$ from $v$ such that for every strategy $\sigma_j$, we have $\mu_j(\< \bsigma_{-j}, \sigma_j \>) = 0$.
	        Deciding the existence of such a strategy profile is an $\NP \cap \coNP$-easy problem --- see for example~\cite{DBLP:conf/cav/EmersonJS93}.
	        
	        Moreover, according to~\cite{CSL} again, if there exists such a play $\pi$ satisfying $\mu_i(\pi) \leq t$ for a given player $i$ and a given threshold $t$, there exists one of the form $h c^\omega$, where the lengths of $h$ and $c$ are bounded by a polynomial function of the size of $\Game$.
	        Therefore, an $\NP$ algorithm that recognizes a negative instance of the Nash universal threshold problem is the following: first, guess such a play $hc^\omega$, a set $W \subseteq V$, and for each $v \in W$, a certificate of the fact that from $v$, there exists a strategy profile that prevents the player controlling $v$ to get the payoff $1$ (since that problem is $\NP \cap \coNP$-easy, and therefore $\NP$-easy, such a certificate exists, has polynomial size and can be checked in polynomial time).
	        Then, in polynomial time, check those certificates, check that $\mu_i(hc^\omega) \leq t$, and check that the players controlling vertices of $W$ are the only players who lose along $h c^\omega$. \hspace{1em plus 1fill}\qedhere
	    \end{itemize}
\end{proof}

\section{Proof of Theorem~\ref{thm_reach_nash_verif}} \label{pf_reach_nash_verif}

 \begin{customthm}{\ref{thm_reach_nash_verif}}
    In quantitative reachability games, the  Nash rational verification problem, deterministic or not, is $\coNP$-complete.
\end{customthm}

\begin{proof}    
    We prove here that the Nash universal threshold problem is $\coNP$-complete.
    The theorem will follow by Corollary~\ref{cor_reductions}.
    
    \begin{itemize}
        \item \emph{Hardness.}
        
        Let us proceed by reduction from the $\NP$-complete problem $\Sat$ to the complement of the Nash universal threshold problem.
        Let $\phi = \bigwedge_{j=1}^p C_j$ be a formula from the propositional logic in conjunctive normal form, over the variables $x_1, \dots, x_n$.
        We wish to define a quantitative reachability game $\Game_{\|v_0}$, a player $i$ and a threshold $t$, such that the formula $\phi$ is satisfiable if and only if there exists an NE $\bsigma$ satisfying $\mu_i(\< \bsigma \>) \leq t$. 
        
        Let $\Game_{\|v_0}$ be the $p+1$-player game presented in Figure~\ref{fig_G_ne_reach}.
        Its players are \emph{Solver}, written $\SS$, and the clauses $C_1, \dots, C_p$.
        When it is a relevant information, the player controlling a vertex is written in blue.
        The target set of Solver is $\{\blacktriangle\}$, and the target set of each clause player $C_j$ is $\{\blacktriangle\} \cup \{\l ~|~ \l \text{ is a literal of } C_j\}$.
        
        \begin{figure*}
				\centering
				\begin{tikzpicture}[->,>=latex,shorten >=1pt, initial text={}, scale=0.7, every node/.style={scale=0.7}]
					\node[state, initial left] (?x1) at (0, 0) {$?x_1$};
					\node[state] (x1) at (2, 1.5) {$x_1$};
					\node[state] (nx1) at (2, -1.5) {$\neg x_1$};
					\node (?xi) at (4, 0) {};
					\node (dots) at (6, 0) {\dots};
					\node[state] (?xn) at (8, 0) {$?x_n$};
					\node[state] (xn) at (10, 1.5) {$x_n$};
					\node[state] (nxn) at (10, -1.5) {$\neg x_n$};
					\node[state] (C1) at (12, 0) {$C_1$};
					\node (Ci) at (14, 0) {\dots};
					\node[state] (Cp) at (16, 0) {$C_p$};
					\node[state] (s) at (14, 2) {$\blacktriangle$};
					\node[state] (m) at (18, 0) {$\blacktriangledown$};
					
					\node[blue] (S1) at (0, -0.7) {$\SS$};
					\node[blue] (S1) at (8, -0.7) {$\SS$};
					\node[blue] (C1') at (12, -0.7) {$C_1$};
					\node[blue] (Cp') at (16, -0.7) {$C_p$};
					
					\path (?x1) edge (x1);
					\path (?x1) edge (nx1);
					\path (x1) edge (?xi);
					\path (nx1) edge (?xi);
					\path (?xn) edge (xn);
					\path (?xn) edge (nxn);
					\path (xn) edge (C1);
					\path (nxn) edge (C1);
					\path (C1) edge (Ci);
					\path (Ci) edge (Cp);
					\path (C1) edge (s);
					\path (Cp) edge (s);
					\path (Cp) edge (m);
					\path (s) edge[loop above] (s);
					\path (m) edge[loop right] (m);
				\end{tikzpicture}
				\caption{The game $\Game_{\|?x_1}$}
				\label{fig_G_ne_reach}
			\end{figure*}
		
		Finally, we define $t = 0$ and $i = \SS$.
		That construction can be made in polynomial time: let us now prove that there exists an NE $\bsigma$ satisfying $\mu_\SS(\< \bsigma \>) = 0$ if and only if the formula $\phi$ is satisfiable.
		
		\begin{itemize}
		    \item \emph{If there exists such an NE $\bsigma$.}
		    
		    Then, let $\pi = \< \bsigma \>$, and let us define the valuation $\nu$ that maps each variable $x$ to $1$ if $\pi$ traverses the vertex $x$, and to $0$ otherwise.
		    
		    Let $C_j$ be a clause of $\phi$.
		    Since $\mu_\SS(\< \bsigma \>) = 0$, necessarily, the play $\pi$ traverses the vertex $C_j$, and does not reach the vertex $\blacktriangle$.
		    Therefore, if going to $\blacktriangle$ is not a profitable deviation for player $C_j$, it means that a vertex of her target set has already been visited before.
		    In other words, there is at least one literal $\l$ of $C_j$ such that the vertex $\l$ is traversed by $\pi$.
		    If it is a positive literal $\l = x$, then we have $\nu(x) = 1$ and it is satisfied.
		    If it is a negative literal $\l = \neg x$, then the vertex $x$ is not visited, hence $\nu(x) = 0$, and $\l$ is also satisfied, hence $C_j$ is satisfied, hence $\phi$ is satisfied by $\nu$.

		    \item \emph{If $\phi$ is satisfiable.}
		    
		    Let $\nu$ be a valuation that satisfies it.
		    Let us define the strategy profile $\bsigma$ as follows: after every history of the form $h?x$, we have $\sigma_\SS(h?x) = x$ if $\nu(x) = 1$, and $\sigma_\SS(h?x) = \neg x$ otherwise.
		    Then, for every history of the form $hC_j$, we have $\sigma_{C_j}(hC_j) = C_{j+1}$, or $\sigma_{C_j}(hC_j) = \blacktriangledown$ if $j = p$.
		    
		    Thus, we have $\mu_\SS(\< \bsigma \>) = 0$, as desired.
		    Let us now prove that $\bsigma$ is an NE.
		    Since no clause player plans to go to $\blacktriangle$ after any history, Solver does not have any profitable deviation.
		    Let now $C_j$ be a clause player: since the clause $C_j$ is satisfied by $\nu$, one of its literals is satisfied --- let us write it $\l$.
		    Then, the vertex $\l$, which belongs to the target set of player $C_j$, is traversed by the play $\< \bsigma \>$ before it reaches a vertex controlled by player $C_j$; hence that player has no profitable deviation either.
		    The strategy profile $\bsigma$ is an NE satisfying $\mu_\SS(\< \bsigma \>) = 0$.
		\end{itemize}

		\item \emph{Easiness.}
		
		Let $n = \card V$ and $p = \card \Pi$.
  
		A non-deterministic algorithm that decides this problem in polynomial time operates as follows: first, for each player $j$ and each vertex $u \in V_j$, compute the value $\lambda(u) = \inf_{\bsigma_{-j}} \sup_{\sigma_j} \mu_j(\< \bsigma \>)$, the best payoff that player $j$ can enforce from $u$, or in other words (by Martin's theorem) the worst payoff the other players can force player $j$ to get.
        That can be achieved by a classical fixed-point algorithm: first, for each $v \in V$, initialize $\lambda_0(v) = 1$ if $v \in T_j$ and $\lambda_0(v) = 0$ otherwise.
        Then, for each $k$ and each $v$, define:
        $$\lambda_{k+1}(v) = \begin{cases}
            \underset{vw \in E}{\max} \frac{1}{1 + \frac{1}{\lambda_k(w)}} &\text{if } v \in V_j \\
            \underset{vw \in E}{\min} \frac{1}{1 + \frac{1}{\lambda_k(w)}} &\text{otherwise}.
        \end{cases}$$
        Compute that sequence until $\lambda_k = \lambda_{k+1}$: then, we have $\lambda(u) = \lambda_k(u)$.
        Since for each $k$ and $v$, we have $\lambda_k(v) \in \{0, \frac{1}{n}, \frac{1}{n -1}, \dots, 1\}$, that computation requires at most $\card V$ iterations, and can therefore be done in deterministic polynomial time.
        
        Once that computation is done, the Nash rational verification problem reduces, by Lemma~\ref{lm_ne}, to deciding whether there exists a play $\pi$ in $\Game_{\|v_0}$ such that $\mu_i(\< \bsigma \>) \leq t$, and that for each player $j$ and each vertex $\pi_k \in V_j$, we have either $\Occ(\pi_{< k}) \cap T_j \neq \emptyset$ (player $j$ has already reached their target set) or $\mu_j(\pi_{\geq k}) \geq \lambda(\pi_k)$.
        If such a play exists, there exists one that is a lasso $hc^\omega$ where $|h| \leq 2np+n$ and $c$ is a simple cycle.
        
        Indeed, let $\pi$ be such a play.
        For each $k$, let $P_k$ be the set of the players for which a constraint must be satisfied at step $k$ --- i.e., the set of players $j$ such that, for some $\l \leq k$, we have $\pi_\l \in V_j$ and $\lambda(\pi_\l) \neq 0$.
        Then, since for every $v$, we have $\lambda(v) \in \{0, \frac{1}{n}, \frac{1}{n -1}, \dots, 1\}$, once a constraint appears for some player $j$, it must be satisfied within at most $n$ steps --- in other words, if $j \in P_k$, then $j \not\in P_\l$ for $\l \geq k + n$.
        Therefore, let us consider the prefix $\pi_{\leq 2np+n}$: necessarily, there are indices $k < \l \leq 2np+n$ such that $\pi_k = \pi_\l$ and $P_k = \dots = P_\l = \emptyset$.
        Let us, moreover, choose $k$ and $\l$ such that the cycle $\pi_{k+1} \dots \pi_\l$ is a simple cycle (if it is not the case, we can increase $k$ and decrease $\l$).
        And then, the play $\chi = \pi_{\leq k} (\pi_{k+1} \dots \pi_\l)^\omega$ is also an NE play satisfying $\mu_i(\chi) \leq t$.
        
        As a consequence, at this step, the algorithm can guess non-deterministically such a lasso, and then check in polynomial time that player $i$ gets a payoff smaller than or equal to $t$, and that all the constraints are correctly satisfied. \hspace{1em plus 1fill}\qedhere
    \end{itemize}
\end{proof}

\section{Proof of Theorem~\ref{thm_energy_det_ne_checking}} \label{pf_energy_det_ne_checking}

\begin{customthm}{\ref{thm_energy_det_ne_checking}}
    In energy games, the deterministic Nash-checking problem can be decided in polynomial time.
\end{customthm}

\begin{proof}
    Given an energy game $\Game_{\|v_0}$ and a deterministic multiplayer Mealy machine $\Mach$, the unique strategy profile $\bsigma \in \Comp_{\|v_0}(\Mach)$ is not an NE if and only if these two conditions are satisfied for some player $i$:
    \begin{itemize}
        \item $\mu_i(\< \bsigma \>) = 0$;
        
        \item there exists a play $\pi$ compatible with $\bsigma_{-i}$ such that $\mu_i(\pi) = 1$.
    \end{itemize}
    The first condition can be checked in polynomial time, since the play $\< \bsigma \>$ is a lasso whose size is bounded by a polynomial function of $\lv \Mach \rv$.
    The second condition is satisfied if and only if there exists a play giving player $i$ the payoff $1$ in the one-player game:
    $$\Game' = \left( \{i\}, V \times Q, (V \times Q), E', \mu'\right),$$
    where $\mu': (\pi_0, q_0)(\pi_1, q_1) \dots \mapsto \mu_i(\pi)$, and $E'$ contains the transitions $(u, p)(v, q)$ such that either $(p, u, q, v) \in \Delta$, or $u \in V_i$ and there exists $v'$ such that $(p, u, q, v') \in \Delta$.
    Checking the existence of such a play can be done in polynomial time according to Theorem 7 in~\cite{DBLP:conf/formats/BouyerFLMS08}.
\end{proof}

\section{Proof of Theorem~\ref{thm_energy_checking}} \label{pf_energy_checking}

	\begin{customthm}{\ref{thm_energy_checking}}
		In energy games, the Nash-checking, the subgame-perfect-checking, and the deterministic subgame-perfect-checking problems are $\coNP$-complete.
	\end{customthm}
	
\begin{proof}
    \begin{itemize}
        \item \emph{Hardness.}
        
        Let us proceed by reduction from the $\NP$-complete problem $\SubsetSum$.
        Let $\rho \in \{\Nash, \subgameperfect\}$.
        Let $S = \{s_1, \dots, s_n\} \subseteq \NN$ and let $t \in \NN$.
        We want to construct an energy game $\Game_{\|v_0}$ and a multiplayer Mealy machine $\Mach$, such there exists a strategy profile $\bsigma \in \Comp_{\|v_0}(\Mach)$ that is not a $\rho$-equilibrium if and only if there exists a subset $S' \subseteq S$ such that $\sum_{s \in S'} s = t$.
        
        Let $\Game_{\|v_0}$ be the two-player game presented in Figure~\ref{fig_G_subsetsum}, with two players, player $\Circle$ and player $\Box$.
        Once again, all the rewards that are not written, and in particular all player $\Circle$'s rewards, are equal to $0$.
        
        \begin{figure*}
				\centering
				\begin{tikzpicture}[->,>=latex,shorten >=1pt, initial text={}, scale=1, every node/.style={scale=0.7}]
					\node[state, initial left] (?s1) at (0, 0) {$?s_1$};
					\node[state] (s1) at (2, 1.5) {$s_1$};
					\node[state] (ns1) at (2, -1.5) {$\neg s_1$};
					\node (?si) at (4, 0) {};
					\node (dots) at (6, 0) {\dots};
					\node[state] (?sn) at (8, 0) {$?s_n$};
					\node[state] (sn) at (10, 1.5) {$s_n$};
					\node[state] (nsn) at (10, -1.5) {$\neg s_n$};
					\node[state, rectangle] (BB) at (12, 0) {$\Box$};
					\node[state] (s) at (14, 1.5) {$\blacktriangle$};
					\node[state] (m) at (14, -1.5) {$\blacktriangledown$};
					
					\path (?s1) edge node[above left] {$\stack{\Box}{s_1}$} (s1);
					\path (?s1) edge (ns1);
					\path (s1) edge (?si);
					\path (ns1) edge (?si);
					\path (?sn) edge node[above left] {$\stack{\Box}{s_n}$} (sn);
					\path (?sn) edge (nsn);
					\path (sn) edge (BB);
					\path (nsn) edge (BB);
					\path (BB) edge node[above left] {$\stack{\Box}{-t}$} (s);
					\path (BB) edge node[below left] {$\stack{\Box}{-t-1}$} (m);
					\path (s) edge[loop right] (s);
					\path (m) edge[loop right] (m);
				\end{tikzpicture}
				\caption{The game $\Game_{\|?s_1}$}
				\label{fig_G_subsetsum}
			\end{figure*}
			
			\begin{itemize}
			    \item Let us first use this game to prove the $\coNP$-hardness of the Nash-checking problem: let us define $\Mach$ as the one-state non-deterministic machine that enables all the actions for both players.
			    
			    Assume that there exists a subset $S' \subseteq S$ such that $\sum_{s \in S'} s = t$.
			    Then, let us consider the strategy profile $\bsigma$ defined by $\sigma_\Box(h\Box) = \blacktriangledown$ and $\sigma_\playcircle(h?s_i) = s_i$ if $s_i \in S'$, and $\sigma_\playcircle(h?s_i) = \neg s_i$ otherwise, for every $h$ and $i$.
			    That strategy profile is compatible with $\Mach$, and it is not an NE, since player $\Box$ has exactly the energy level $t$ when he reaches the vertex $\Box$, and therefore loses by going to $\blacktriangledown$ while he could win by going to $\blacktriangle$.
			    
			    Conversely, assume that there is a strategy profile $\bsigma$ compatible with $\Mach$ that is not an NE.
			    Let $\pi = \< \bsigma \>$.
			    Since player $\Circle$ cannot lose in $\Game$, the player who has a profitable deviation is player $\Box$.
			    Since player $\Box$ controls only the vertex $\Box$, that can appear only once along $\pi$, the only possibility is that player $\Box$ loses the play $\pi = \pi_{\leq 2n} \Box \blacktriangledown^\omega$ but wins the play $\pi_{\leq 2n} \Box \blacktriangle^\omega$.
			    That means that $\EL_\Box(\pi_{\leq 2n} \Box) = t$.
			    Therefore, if $S' = \{s \in s ~|~ \text{the vertex s appears in } \pi\}$, we have $\sum_{s \in S'} s = t$.
			    
			    Thus, the subset-sum problem reduces in polynomial time to the complement of the Nash-checking problem, which is therefore $\coNP$-hard.

			    \item Let us now use the same game to prove that the deterministic subgame-perfect-checking problem, and therefore the subgame-perfect-checking problem, are $\coNP$-hard.
			    Let us define $\Mach$ as a one-state deterministic machine that defines arbitrarily the actions of player $\playcircle$ and that forces player $\Box$, from the vertex $\Box$, to go to the vertex $\blacktriangledown$.
			    Let $\bsigma$ be the unique strategy profile compatible with $\Mach$.
			    
			    Then, if there exists a subset $S' \subseteq S$ such that $\sum_{s \in S'} s = t$, we can define the history $h = h_0 \dots h_{2n}$ with $h_{2i+1} = s_i$ if $s_i \in S'$ and $h_{2i+1} = \neg s_i$ otherwise.
			    Then, we have $\EL_\Box(h) = t$, hence in the subgame $\Game_{\|h}$, player $\Box$ loses by going to $\blacktriangledown$, while he could win by going to $\blacktriangle$, hence $\bsigma$ is not an SPE.
			    
			    Conversely, if $\bsigma$ is not an SPE, then let $h$ be such that $\bsigma_{\|h}$ is not an NE.
			    Necessarily, the player who has a profitable deviation is player $\Box$.
			    Thus, without loss of generality, we can assume that $h$ ends in the vertex $\Box$.
			    Then, if we define $S' = \{s ~|~ s \text{ appears along } h\}$, we have $\sum_{s \in S'} = t$.
			    
			    Thus, the problem $\SubsetSum$ reduces in polynomial time to the complement of the deterministic subgame-perfect-checking problem, which is therefore $\coNP$-hard.
			\end{itemize}

	    \item \emph{Easiness.}
	    
	    Let $i \in \Pi$.
	    Let us consider the graph $G_i$, whose vertex set is $V \times Q$, and whose edges are partitioned into \emph{green edges}, i.e. the edges $(u, p)(v, q)$ with $(p, u, q, v) \in \Delta$, and the \emph{red edges}, i.e. the edges $(u, p)(v, q)$ with $(p, u, q, w) \in \Delta$ for some $w \neq v$, with $uv \in E$ and $u \in V_i$.
	    We define the reward function $r: (u, p)(v, q) \mapsto r_i(uv)$.
	    That graph may be seen as the product game $\Game_{\|v_0} \otimes \Mach$, but where $\Mach$ is a multiplayer Mealy machine, hence Demon controls all the vertices.
	    
	    There exists a strategy profile $\bsigma \in \Comp_{\|v_0} (\Mach)$ in which player $i$ has a profitable deviation in some subgame (resp. in which player $i$ has a profitable deviation) if and only if in the graph $G_i$, there exists a finite path $p$ (resp. a path $p$ made of green edges) from the vertex $(v_0, q_0)$ to a vertex $(v, q)$ with $\EL_r = n$, and if from $(v, q)$ with energy $n$, there exist a path $(v, q) \pi$ that uses only green edges and that is lost by player $i$, and a path $(v, q) \pi'$ that can use all edges and that is won by player $i$, with $\pi_0 \neq \pi'_0$.
	    
	    A non-deterministic algorithm that recognizes negatives instances of the subgame-perfect-checking problem (resp. the Nash-checking problem) in polynomial time is therefore the following one: first, for each player $i$, for each vertex $w$ of the graph $G_i$, compute in deterministic polynomial time, using the algorithm provided by~\cite{DBLP:conf/formats/BouyerFLMS08}:
	    \begin{itemize}
	        \item the least integer $a(w)$ such that there exists a winning infinite path $\pi'$ starting from $w$ with energy $a(w)$;
	        
	        \item the greatest integer $b(w)$ such that there exists a losing infinite path $\pi$ starting from $w$ with energy $b(w)$ using only green edges.
	    \end{itemize}
	    
	    Then, use the $\NP$ algorithm provided by~\cite{DBLP:conf/concur/HaaseKOW09} to decide whether there exists a vertex $(v, q)$ in $G_i$ and two edges $(v, q)w \neq (v, q)w'$, such that from the vertex $(v_0, q_0)$, one can reach $(v, q)$ (resp. one can reach $(v, q)$ using only green edges) with energy level $n$ satisfying $n + r((v, q)w) \leq b(w)$ and $n + r((v, q)w') \geq a(w)$. \hspace{1em plus 1fill}\qedhere
    \end{itemize}
\end{proof}

\section{Two-counter machines} \label{app_two_counter_machines}

A \emph{two-counter machine} is a tuple
	$\Kount = \left(Q, q_0, q_\f, \Delta_{C_1^+}, \Delta_{C_2^+}, \Delta_{C_1^-}, \Delta_{C_2^-}\right),$
	where $Q$ is a finite set of \emph{states}, where $q_0 \in Q$ is the \emph{initial state} and $q_\f$ is the \emph{final state}, and where $\Delta_{C_1^+}, \Delta_{C_2^+} \in Q \times Q$ are the sets of \emph{incremental transitions}, and $\Delta_{C_1^-}, \Delta_{C_2^-} \subseteq Q \times Q \times Q$ are the sets of \emph{test transitions}, such that every state $q \in Q \setminus \{q_\f\}$ admits exactly one outgoing transition (either incremental or test), and such that the state $q_\f$ admits none.
	For each \emph{counter} $C \in \{C_1, C_2\}$, we define the \emph{interpretation} of $C$ as the function $\hC: q_0 Q^* \to \NN$ defined by $\hC(q_0) = 0$, and $\hC(q_0 \dots q_n) = \hC(q_0 \dots q_{n-1}) + 1$ if $(q_{n-1}, q_n) \in \Delta_{C^+}$, $\hC(q_0 \dots q_n) = \hC(q_0 \dots q_{n-1}) - 1$ if $(q_{n-1}, q_n, q) \in \Delta_{C^-}$ for some $q$, and $\hC(q_0 \dots q_n) = \hC(q_0 \dots q_{n-1})$ otherwise.
	A \emph{run} of $\Kount$ is a sequence $\eta_0 \eta_1 \dots$, either infinite or ending with $q_\f$, such that for every index $k$, we have either $(\eta_k, \eta_{k+1}) \in \Delta_{C^+}$, or $(\eta_k, \eta_{k+1}, q) \in \Delta_{C^-}$ for some $q$ and $\hC(\eta_0 \dots \eta_k) > 0$, or $(\eta_k, q, \eta_{k+1}) \in \Delta_{C^-}$ for some $q$ and $\hC(\eta_0 \dots \eta_k) = 0$, for some $C \in \{C_1, C_2\}$.
	Note that each two-counter machine admits exactly one run.
	A machine \emph{halts} if that unique run is finite.
	The halting problem of a two-counter machine is undecidable, and in particular not co-recursively enumerable.
	
	Figures~\ref{fig_counter_machine1} and~\ref{fig_counter_machine2} present two examples of two-counter machines.
		The arrow with the label $C_1^{+}$ from $q_0$ to $q_1$ indicates a transition $(q_0, q_1) \in \Delta_{C_1^{+}}$.
		The arrows with the label $C_1^{-}$ from $q_1$ to itself and with the label $C_1 = 0$ from $q_1$ to $q_\f$ indicate a transition $(q_1, q_1, q_\f) \in \Delta_{C_1^{-}}$.
		
		\begin{figure}
			\caption{A machine that halts on $(0, 0)$}
				\centering
				\begin{tikzpicture}[->,>=latex,shorten >=1pt, initial text={}, scale=1, every node/.style={scale=0.8}]
					\node[state, initial left] (q0) at (0, 0) {$q_0$};
					\node[state] (q1) at (2, 0) {$q_1$};
					\node[state, double] (qf) at (4, 0) {$q_\f$};
					
					\path (q0) edge node[above] {$C_1^{+}$} (q1);
					\path (q1) edge [loop above] node[above] {$C_1^{-}$} (q1);
					\path (q1) edge node[above] {$C_1=0$} (qf);
				\end{tikzpicture}
				\label{fig_counter_machine1}
        \end{figure}
        \begin{figure}
			\caption{A machine that does not halt on $(0, 0)$}
				\centering
				\begin{tikzpicture}[->,>=latex,shorten >=1pt, initial text={}, scale=1, every node/.style={scale=0.8}]
					\node[state, initial left] (q0) at (0, 0) {$q_0$};
					\node[state] (q1) at (2, 0) {$q_1$};
					\node[state, double] (qf) at (4, 0) {$q_\f$};
					
					\path[bend left] (q0) edge node[above] {$C_1^{+}$} (q1);
					\path[bend left] (q1) edge node[below] {$C_1^{-}$} (q0);
					\path (q1) edge node[above] {$C_1=0$} (qf);
				\end{tikzpicture}
				\label{fig_counter_machine2}	
		\end{figure}
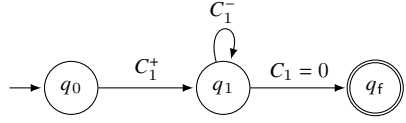

\section{Proof of Theorem~\ref{thm_energy_nash_verif}} \label{pf_energy_nash_verif}

\begin{customthm}{\ref{thm_energy_nash_verif}}
    In energy games, the  Nash rational verification problem, deterministic or not, is undecidable and co-recursively enumerable.
\end{customthm}

\begin{proof}
    We prove here that the Nash UT problem is undecidable and co-recursively enumerable.
    The theorem will follow by Corollary~\ref{cor_reductions}.

\begin{itemize}
    \item \emph{Undecidability.}

    We show undecidability by reduction from the halting problem of a two-counter machine.

\begin{figure}
    \centering
    \begin{subfigure}[b]{0.24\textwidth}
			\centering
			\begin{tikzpicture}[->,>=latex,shorten >=1pt, initial text={}, scale=0.8, every node/.style={scale=0.7}]
				\node (1) at (3, 0) {};
				\node[state, initial left] (q0) at (0, 0) {$q_0^1$};
                \node[state] (q0') at (2, 0) {$q_0^2$};
                \node[blue] (pq0) at (0, -0.6) {$C_1^\top$};
                \node[blue] (pq0') at (2, -0.6) {$C_2^\top$};
				\path (q0) edge (q0');
                \path (q0') edge (1);
				\path (q0) edge[loop above] (q0);
                \path (q0') edge[loop above] (q0');
			\end{tikzpicture}
			\caption{Initial state}
			\label{fig_gadget_initial_app}
		\end{subfigure}
		\begin{subfigure}[b]{0.2\textwidth}
			\centering
			\begin{tikzpicture}[->,>=latex,shorten >=1pt, initial text={}, scale=0.8, every node/.style={scale=0.7}]
				\node[state, initial above] (qf) at (0, 0) {$q_\f$};
				\path (qf) edge[loop right] node[right] {$\stack{C_1^\bot}{-1}\,\stack{C_2^\bot}{-1}\,\stack{\WW}{-1}$} (qf);
			\end{tikzpicture}
			\caption{Final state}
			\label{fig_gadget_final_app}
		\end{subfigure}
		\begin{subfigure}[b]{0.2\textwidth}
			\centering
			\begin{tikzpicture}[->,>=latex,shorten >=1pt, initial text={}, scale=1, every node/.style={scale=0.8}]
				\node[state, initial left] (q) at (0, 0) {$q$};
				\node (2) at (1, 0) {};
				\path (q) edge node[above] {$\stack{C^\top}{1}\,\stack{C^\bot}{1}$} (2);
			\end{tikzpicture}
			\caption{Incrementations}
			\label{fig_gadget_incrementation_app}
		\end{subfigure}
	\begin{subfigure}[b]{0.3\textwidth}
		\centering
		\begin{tikzpicture}[->,>=latex,shorten >=1pt, initial text={}, scale=1, every node/.style={scale=0.8}]
            \newcommand{\deltay}{0.5}
			\node (1) at (-1, 0) {};
			\node[state] (q) at (0, 0) {$q$};
			\node (2) at (2, 0*\deltay) {(if $C > 0$)};
			\node[state] (q') at (0, 2*\deltay) {$q'$};
			\node (3) at (2, 2*\deltay) {(if $C = 0$)};
			\node[state] (s) at (-1.2, 2*\deltay) {$\blacktriangle$};
   \node[blue] (pq) at (-0.4, -0.4) {$C^\top$};
   \node[blue] (pq') at (-0.4, 2*\deltay-0.4) {$C^\bot$};
			
			\path (1) edge (q);
			\path (q) edge node[above] {$\stack{C^\top}{-1}\,\stack{C^\bot}{-1}$} (2);
			\path (q) edge (q');
			\path (q') edge (3);
			\path (q') edge node[above] {$\stack{C^\bot}{-1}$} (s);
			\path (s) edge [loop left] (s);
		\end{tikzpicture}
		\caption{Tests}
		\label{fig_gadget_test_app}
	\end{subfigure}
    \caption{Gadgets}
    \label{fig_gadget_app}
\end{figure}

	Let $\Kount$ be a two-counter machine.
	We define an energy game $\Game_{\|q_0^1}$ with five players --- players $C_1^\top$, $C_1^\bot$, $C_2^\top$, $C_2^\bot$, and $\WW$, called \emph{Witness} --- by assembling the gadgets presented in Figure~\ref{fig_gadget_app} --- the rewards that are not presented are equal to $0$, and the players controlling relevant vertices are written in blue.
	For each state of $\Kount$, we define from one to two vertices, plus the additional vertex $\blacktriangle$.
	Then, a play in $\Game_{\|v_0}$ that does not reach the vertex $\blacktriangle$ simulates a sequence of transitions of $\Kount$, that can be a valid run or not: at each step, the counter $C_i$ is captured by the energy level of player $C_i^\top$, always equal to the energy level of player $C_i^\bot$.
	Let us now prove that the game $\Game_{\|q_0^1}$ admits an NE where Witness loses if and only if the machine $\Kount$ terminates.
	
	\begin{itemize}
		\item \emph{If such an NE exists.}
		
		Then, let us write $\bsigma$ for it, and let $\pi = \< \bsigma \>$.
		Since $\pi$ is lost by Witness, the play $\pi$ reaches the vertex $q_\f$, and therefore simulates a terminating sequence of transitions of $\Kount$ (note that one transition may be represented by several edges).
        We must now prove that that run is valid, i.e. that tests are simulated correctly.
        
        Let us assume that at some point along the play $\pi$, from a vertex $q$, player $C_i^\top$, with $i \in \{1, 2\}$, does not take the edge to the vertex $q'$ while her energy level is zero.
        Then, she loses, hence she has a profitable deviation at the beginning of the play, by looping on the vertex $q_0^i$.

        Let us now assume that she goes to the vertex $q'$, while her energy is positive.
        Then, player $C_i^\bot$'s energy is also positive: he can go to the vertex $\blacktriangle$ and win.
        That would be a profitable deviation, since, as the play $\pi$ reaches the vertex $q_\f$, it is lost by player $C_i^\bot$.

        Therefore, the play $\pi$ does not fake any test, and simulates correctly the machine $\Kount$, which terminates.
	
		\item \emph{If the machine $\Kount$ terminates.}
		
		Then, let us define a strategy profile $\bsigma$ in $\Game_{\|q_0^1}$ as follows: in tests of counter $C$, player $C^\top$ goes to $q'$ if and only if her energy level is positive; from $q'$, player $C^\bot$ never goes to the vertex $\blacktriangle$.
        Let $\pi = \< \bsigma \>$: since tests are simulated correctly, the play $\pi$ simulates the run of $\Kount$, and therefore reaches the vertex $q_\f$.
        It is therefore lost by Witness, $C_1^\bot$ and $C_2^\bot$.
  
		Witness has no profitable deviation, since he does not control any vertex.
        The only vertices controlled by each player $C^\bot$ are the vertices of the form $q'$.
        Those vertices are reached only when player $C^\top$'s, and therefore player $C^\bot$'s energy level is zero.
        Then, deviating and going to the vertex $\blacktriangle$ is not profitable for player $C^\bot$, since it makes him immediately lose.
        The strategy profile $\bsigma$ is therefore an NE, lost by Witness.
	\end{itemize}

    Every NE outcome in the game $\Game_{\|q_0^1}$ is won by Witness if and only if the machine $\Kount$ does not terminate.
    Therefore, the halting problem of two-counter machines reduces to the Nash UT problem in energy games, which is therefore undecidable.

    \item \emph{Co-recursive enumerability.}
    
    First let us prove that in an energy game $\Game_{\|v_0}$, if there exists an NE that makes some player $i$ lose, then there exists a finite-memory one.

    Indeed, in such a case, let $\pi = \< \bsigma \>$.
    By Dickson's lemma, there exist two indices $m$ and $n$, with $m < n$, such that $\pi_m = \pi_n$ and such that for every player $j$ that does not lose the play $\pi$, we have $\EL_j(\pi_{\leq m}) \leq \EL_j(\pi_{\leq n})$.
    Moreover, we can chose $m$ great enough to have $\EL_j(\pi_{\leq m}) = \bot$ for every player $j$ that loses the play $\pi$.
    Thus, the players winning the play $\pi$ are exactly the players winning the play $\chi = \pi_{< m} \left(\pi_m \dots \pi_{n - 1}\right)^\omega$.

    Now, let $k \leq n$, and let $j$ be the player controlling the vertex $\pi_k$.
    If $j$ is a player who loses the play $\pi$, then since $\bsigma$ is an NE, the strategy profile $\bsigma_{\|\pi_{\leq k}}$ is a strategy profile against which player $j$ cannot win.
    It is known (see for example~\cite{DBLP:conf/formats/BouyerFLMS08}, Lemma~10) that memoryless strategies are sufficient to falsify an energy objective.
    Therefore, let $\btau^k_{-j}$ be a memoryless strategy profile, from the vertex $\pi_k$, against which player $j$ cannot win.
    Let $\tau^k_j$ be an arbitrary memoryless strategy.
    If $j$ is not a player who loses the play $\pi$, then we define $\btau^k$ as an arbitrary memoryless strategy profile.

    Let $\Mach$ be the deterministic Mealy machine defined as follows: it has $2n$ states, namely $q_0, \dots, q_{n-1}$ and $q'_0, \dots, q'_{n-1}$.
    From each state $q_k$, the transition reading the vertex $\pi_k$ leads to the state $q_{k+1}$ (or $q_m$ if $k = n-1$), and outputs the vertex $\pi_{k+1}$ (or $\pi_m$ if $k = n-1$).
    The transition reading any other vertex $v$ (if $k \geq 1$) leads to the state $q'_k$ and outputs the vertex $\btau^{k-1}(\pi_{< k} v)$.
    From the state $q'_k$, the transition reading each vertex $v$ leads to $q'_k$, and outputs the vertex $\btau^{k-1}(h v)$, for any history $h$ (since $\btau^{k-1}$ is memoryless).
    Thus, the only strategy profile compatible with the machine $\Mach$ is the strategy profile that follows the play $\chi$, and that punishes any player who deviates by following the memoryless strategy profile $\btau^k$.
    That strategy profile is finite-memory, and is lost by player $i$, as desired.

    Thus, a semi-algorithm that recognizes the negative instances of the UT problem consists in enumerating the deterministic multiplayer Mealy machines on $\Game_{\|v_0}$, and for each of them, to check (by diagonalization):
    \begin{itemize}
        \item whether the only strategy profile compatible with it is an NE: that problem is decidable (in polynomial time) by Theorem~\ref{thm_energy_checking};

        \item whether that strategy profile makes player $i$ lose: that is recursively enumerable, by constructing step by step its outcome and computing the energy levels on the fly.
    \end{itemize}
    We have a negative instance of the UT problem if and only if at least one Mealy machine satisfies those two conditions.
    The Nash UT problem is therefore co-recursively enumerable. \hspace{1em plus 1fill}\qedhere
\end{itemize}
\end{proof}

\section{Proof of Theorem~\ref{thm_energy_sp_verif}} \label{pf_energy_sp_verif}

\begin{customthm}{\ref{thm_energy_sp_verif}}
		In energy games, the subgame-perfect rational verification problem, deterministic or not, is undecidable, even when Leader plays against only two players.
	\end{customthm}

\begin{proof}
    We prove first the undecidability of the subgame-perfect UT problem with two players.
	We proceed by reduction from the halting problem of a two-counter machine.
	
	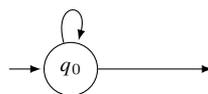
\begin{figure}
			\centering
			\begin{tikzpicture}[->,>=latex,shorten >=1pt, initial text={}, scale=1, every node/.style={scale=0.8}]
				\node (1) at (2, 0) {};
				\node[state, initial left] (q0) at (0, 0) {$q_0$};
				\path (q0) edge (1);
				\path (q0) edge[loop above] (q0);
			\end{tikzpicture}
			\caption{Gadget for the initial state}
			\label{fig_gadget_initial_sp}
		\end{figure}
		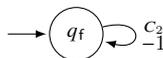
\begin{figure}
			\centering
			\begin{tikzpicture}[->,>=latex,shorten >=1pt, initial text={}, scale=1, every node/.style={scale=0.8}]
				\node (1) at (-1, 0) {};
				\node[state] (qf) at (0, 0) {$q_\f$};
				\path (1) edge (qf);
				\path (qf) edge[loop right] node[right] {$\stack{C_2}{-1}$} (qf);
			\end{tikzpicture}
			\caption{Gadget for the final state}
			\label{fig_gadget_final_sp}
		\end{figure}
		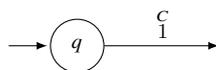
\begin{figure}
			\centering
			\begin{tikzpicture}[->,>=latex,shorten >=1pt, initial text={}, scale=1, every node/.style={scale=0.8}]
				\node[state] (q) at (0, 0) {$q$};
				\node (1) at (-1, 0) {};
				\node (2) at (2, 0) {};
				\path (1) edge (q);
				\path (q) edge node[above] {$\stack{C}{1}$} (2);
			\end{tikzpicture}
			\caption{Gadget for incrementations of counter $C \in \{C_1, C_2\}$}
			\label{fig_gadget_incrementation_sp}
		\end{figure}
	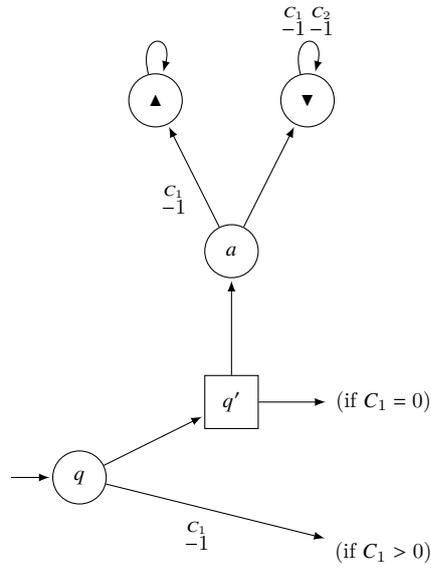
\begin{figure}
		\centering
		\begin{tikzpicture}[->,>=latex,shorten >=1pt, initial text={}, scale=1, every node/.style={scale=0.8}]
			\node (1) at (-1, 0) {};
			\node[state] (q) at (0, 0) {$q$};
			\node (2) at (4, -1) {(if $C_1 > 0$)};
			\node[state, rectangle] (q') at (2, 1) {$q'$};
			\node (3) at (4, 1) {(if $C_1 = 0$)};
			\node[state] (q'') at (2, 3) {$a$};
			\node[state] (top) at (1, 5) {$\blacktriangle$};
			\node[state] (bot) at (3, 5) {$\blacktriangledown$};
			
			\path (1) edge (q);
			\path (q) edge node[below left] {$\stack{C_1}{-1}$} (2);
			\path (q) edge (q');
			\path (q') edge (3);
			\path (q') edge (q'');
			\path (q'') edge node[below left] {$\stack{C_1}{-1}$} (top);
			\path (q'') edge (bot);
			\path (top) edge [loop above] (top);
			\path (bot) edge [loop above] node[above] {$\stack{C_1}{-1} \, \stack{C_2}{-1}$} (bot);
		\end{tikzpicture}
		\caption{Gadget for tests of counter $C_1$}
		\label{fig_gadget_test_C1_sp}
	\end{figure}
	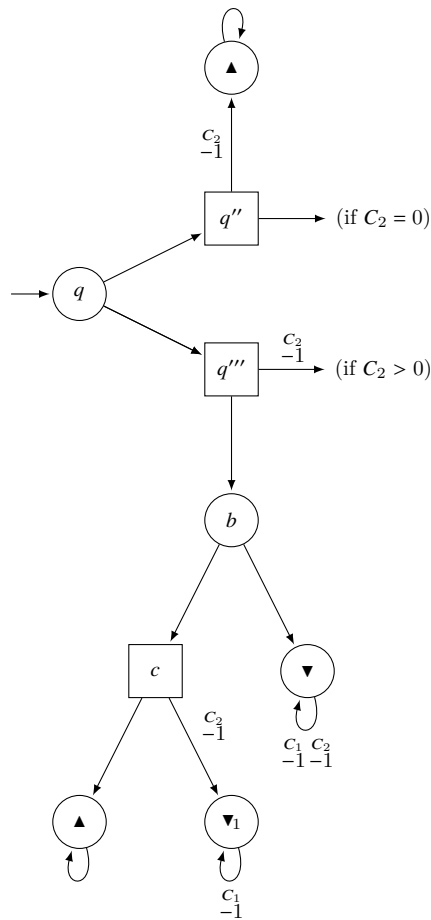
\begin{figure}
				\centering
				\begin{tikzpicture}[->,>=latex,shorten >=1pt, initial text={}, scale=1, every node/.style={scale=0.8}]
					\node (1) at (-1, 0) {};
					\node[state] (q) at (0, 0) {$q$};
					\node[state, rectangle] (q') at (2, 1) {$q''$};
					\node (2) at (4, 1) {(if $C_2 = 0$)};
					\node[state] (s1) at (2, 3) {$\blacktriangle$};
					\node[state, rectangle] (q'') at (2, -1) {$q'''$};
					\node (3) at (4, -1) {(if $C_2 > 0$)};
					\node[state] (q''') at (2, -3) {$b$};
					\node[state, rectangle] (q'''') at (1, -5) {$c$};
					\node[state] (s2) at (0, -7) {$\blacktriangle$};
					\node[state] (mc) at (2, -7) {$\triangleone$};
					\node[state] (m) at (3, -5) {$\blacktriangledown$};
					
					\path (1) edge (q);
					\path (q) edge (q');
					\path (q) edge (q'');
					\path (q') edge node[left] {$\stack{C_2}{-1}$} (s1);
					\path (s1) edge[loop above] (s1);
					\path (q') edge (2);
					\path (q) edge (q'');
					\path (q'') edge node[above] {$\stack{C_2}{-1}$} (3);
					\path (q'') edge (q''');
					\path (q''') edge (m);
					\path (q''') edge (q'''');
					\path (q'''') edge (s2);
					\path (q'''') edge node[above right] {$\stack{C_2}{-1}$} (mc);
					\path (m) edge [loop below] node[below] {$\stack{C_1}{-1} \, \stack{C_2}{-1}$} (m);
					\path (s2) edge [loop below] (s2);
					\path (mc) edge [loop below] node[below] {$\stack{C_1}{-1}$} (mc);
				\end{tikzpicture}
				\caption{Gadget for tests of counter $C_2$}
				\label{fig_gadget_test_C2_sp}
			\end{figure}

	Let $\Kount$ be a two-counter machine.
	We define an energy game $\Game_{\|q_0}$ with two players, player $C_1$ and player $C_2$, by assembling the gadgets presented in Figures~\ref{fig_gadget_initial_sp}, \ref{fig_gadget_final_sp}, \ref{fig_gadget_incrementation_sp}, \ref{fig_gadget_test_C1_sp}, and \ref{fig_gadget_test_C2_sp} --- the rewards that are not presented are equal to $0$, the round vertices are those controlled by player $C_1$, and the square ones are those controlled by player $C_2$.
	For each state of the machine $\Kount$, we define from one to three vertices, plus six additional vertices, written $a$, $b$, $c$, $\blacktriangle$, $\blacktriangledown$, and $\triangleone$.
	
	Then, a play in $\Game$ that does not reach one of the sink vertices $\blacktriangle$, $\blacktriangledown$, or $\triangleone$ simulates a sequence of transitions of the machine $\Kount$, that can be a run or not: at each step, the counter $C$ is captured by the energy level of player $C$.
	Let us now prove that the game $\Game_{\|q_0}$ admits an SPE where player $C_2$ loses if and only if the machine $\Kount$ terminates.
	
	\begin{itemize}
		\item \emph{If such an SPE exists.}
		
		Then, let us write $\bsigma$ for it, and let $\pi = \< \bsigma \>$.
  Let us show that $\pi$ simulates a run of $\Kount$.
		Since $\pi$ is lost by player $C_2$, there are \emph{a priori} three possibilities:
		
		\begin{itemize}			
			\item \emph{The play $\pi$ reaches the vertex $\blacktriangledown$.}
			Then, player $C_1$ loses, and has therefore a profitable deviation by looping on $\pi_0 = q_0$, which is impossible.
			
			\item \emph{Player $C_1$ makes player $C_2$ lose by going to a vertex of the form $q'''$, when his energy level is zero.} Thus, the play $\pi$ simulates a spurious run of $\Kount$, since such an action amounts to fake a test of $C_2$ above zero.
			But then, player $C_2$ has a profitable deviation by going to $b$: there, player $C_1$ cannot go to the vertex $\blacktriangledown$, because it would make her lose, while she can win by going to the vertex $c$; indeed, from there, player $C_2$ cannot go to the vertex $\triangleone$ because it would make him lose, since he has no more energy, while he can go to the vertex $\blacktriangle$ and win --- and let player $C_1$ win.
			Therefore, that case is also impossible.
			
			\item \emph{The play $\pi$ reaches the vertex $q_\f$.}
			Then, it simulates a correct run of the machine $\Kount$, that reaches the state $q_\f$.
			Indeed, we have already shown that $\pi$ cannot fake a test of $C_2$ above $0$.
			It cannot fake a test of $C_1$ above $0$ either, because then, player $C_1$ would lose, while she can win by looping on $q_0$.
			It cannot fake a test of $C_2$ to $0$, because then, from the vertex $q''$, if player $C_2$'s energy is greater than $0$, he has a profitable deviation by going to the vertex $\blacktriangle$.
			Finally, it cannot fake a test of $C_1$ to $0$, because then, from the vertex $q'$, if player $C_1$'s energy is greater than $0$, player $C_2$ has a profitable deviation by going to $a$, from where player $C_1$ cannot go to the vertex $\blacktriangledown$, because it would make her lose while she can win by going to $\blacktriangle$.
   Therefore, the machine $\Kount$ terminates.
		\end{itemize}
	
		\item \emph{If the machine $\Kount$ terminates.}
		
		Then, let us define a strategy profile $\bsigma$ in $\Game_{\|q_0}$ as follows:
		\begin{itemize}
			\item in tests of counter $C_1$, player $C_1$ goes to $q'$ if and only if her energy level is zero or $\bot$; from $q'$, player $C_2$ goes to the vertex $a$ if and only if player $C_1$'s energy level is positive;
			
			\item in tests of counter $C_2$, player $C_1$ goes to $q''$ if and only if the energy level of player $C_2$ is zero or $\bot$; from $q''$, player $C_2$ goes to the vertex $\blacktriangle$ if and only if his energy level is positive, and from $q'''$, he goes to $b$ if and only if it is zero or $\bot$;
			
			\item from the vertex $a$, player $C_1$ goes to the vertex $\blacktriangledown$ if and only if her energy level is zero or $\bot$;
			
			\item from the vertex $b$, player $C_1$ goes to the vertex $\blacktriangledown$ if and only if the energy level of player $C_2$ is positive;
			
			\item from the vertex $c$, player $C_2$ goes to the vertex $\triangleone$ if and only if his energy level is positive.
		\end{itemize}
		
		Let us show that $\bsigma$ is an SPE.
		Let $hv$ be a history from the vertex $q_0$: let us prove that $\bsigma_{\|hv}$ is an NE.
		
		\begin{itemize}
			\item \emph{If $v$ is a vertex of the form $q$.}
			
			Then, the play $\< \bsigma_{\|hv} \>$ simulates the correct run of the machine $\Kount$ from $q$ when the counters are initialized to $\EL_{C_1}(hv)$ and $\EL_{C_2}(hv)$ (if one of those energy levels is $\bot$, then it simulates the correct run of $\Kount$ where that counter is locked to $0$)
			Therefore, player $C_1$ wins (or has already lost), hence she cannot have a profitable deviation.
			As for player $C_2$, he cannot have a profitable deviation from a vertex of the form $q'$: when such a vertex is reached, player $C_1$'s energy is zero or $\bot$, hence if player $C_2$ chooses to go to $a$, player $C_1$ will go to $\blacktriangledown$ and he will lose.
			He cannot have a profitable deviation from a vertex of the form $q''$ either: when such a vertex is reached, he has a zero or $\bot$ energy level, hence if he chooses to go to the vertex $\blacktriangle$, he loses.
			Finally, he cannot have a profitable deviation from a vertex of the form $q'''$: when such a vertex is reached, he has a positive energy level, hence if he goes to the vertex $b$, player $C_1$ will go to the vertex $\blacktriangledown$ and he will lose.
			
			\item \emph{If $v$ is a vertex of the form $q'$.}
			
			Then, if player $C_1$'s energy level is positive, player $C_2$ goes to $a$, then player $C_1$ goes to $\blacktriangle$, and both player $C_1$ and $C_2$ win --- and have therefore no profitable deviation.
			Otherwise, the substrategy profile $\bsigma_{\|hv}$ simulates a correct run of $\Kount$, and we can use the same arguments as in the previous point.
			
			\item \emph{If $v$ is a vertex of the form $q''$.}
			
			Then, if player $C_2$'s energy level is positive, he goes to $\blacktriangle$ and wins --- and no player has a profitable deviation.
			Otherwise, the substrategy profile $\bsigma_{\|hv}$ simulates a correct run of $\Kount$, and we can use the same arguments than in the first point.
			
			\item \emph{If $v$ is a vertex of the form $q'''$.}
			
			Then, if player $C_2$'s energy level is zero or $\bot$, player $C_2$ goes to $b$, then player $C_1$ goes to $c$, and finally player $C_2$ goes to $\blacktriangle$, and both player $C_1$ and $C_2$ win --- they have therefore no profitable deviation.
			Otherwise, the substrategy profile $\bsigma_{\|hv}$ simulates a correct run of $\Kount$, and we can use the same arguments than in the first point.
			
			\item \emph{If $v = a$.}
			
			Then, either player $C_1$ has a positive energy level, and then she goes to $\blacktriangle$ and wins, or she has a zero or $\bot$ energy level, and then she cannot win.
			
			\item \emph{If $v = b$.}
			
			Then, either player $C_2$ has a zero or $\bot$ energy level, and then the play ends in $\blacktriangle$ and both players win (or have already lost), or he has a positive energy level, and then player $C_1$ cannot win, since player $C_2$ plans to go to $\triangleone$.
			
			\item \emph{If $v = c$.}
			
			Then, either $C_2$ has a positive energy level, and he goes to $\triangleone$ and wins, or he has a zero energy level, and he goes to $\blacktriangle$ and wins, or he has already lost.
			
			\item \emph{If $v = \blacktriangledown$, $\triangleone$, $q_\f$ or $\blacktriangle$,} then the proof is immediate.
		\end{itemize}
	
		Therefore, the strategy profile $\bsigma$ is an SPE, that simulates the run of the machine $\Kount$.
  That run terminates, hence the play $\< \bsigma \>$ reaches the vertex $q_\f$, and is lost by player $C_2$. 
	\end{itemize}

    Thus, the halting problem of two-counter machines reduces to the subgame-perfect UT problem in energy games with two-players, which is therefore undecidable.
    That same problem reduces, by adding a Leader controlling no vertex, to the deterministic subgame-perfect rational verification problem in energy games where Leader plays against two players.
    Consequently, that problem is itself undecidable.
\end{proof}

\section{Proof of Proposition~\ref{pptn_energy_infinite_memory}} \label{pf_energy_infinite_memory}

\begin{figure}
			\centering
			\begin{tikzpicture}[->,>=latex,shorten >=1pt, initial text={}, scale=0.8, every node/.style={scale=0.7}]
				\node[state, initial left] (a) at (0, 0) {$a$};
				\node[state, rectangle] (b) at (2, 1) {$b$};
				\node[state, diamond] (c) at (2, -1) {$c$};
				\node[state] (d) at (4, 0) {$d$};
				\node[state] (e) at (6, 0) {$e$};
				\path (a) edge[bend right=20] (b);
				\path (b) edge[bend right=20] node[above left] {$\stackrel{\playcircle}{1}\stackrel{\Box}{1} \stackrel{\Diamond}{1}$} (a);
				\path (a) edge[bend left=20] (c);
				\path (c) edge[bend left=20] node[below left] {$\stackrel{\playcircle}{1}\stackrel{\Box}{1} \stackrel{\Diamond}{1}$} (a);
				\path (b) edge node[above] {$\stackrel{\playcircle}{1}$} (d);
				\path (c) edge node[below] {$\stackrel{\playcircle}{1}$} (d);
				\path (d) edge [loop above] node {$\stackrel{\playcircle}{-1} \, \stackrel{\Box}{-1} \, \stackrel{\Diamond}{-1}$} (d);
				\path (d) edge (e);
				\path (e) edge [loop right] (e);
			\end{tikzpicture}
			\caption{A game where infinite memory is necessary to make player $\Box$ lose}
			\label{fig_energy_infinite_memory_bis}
		\end{figure}
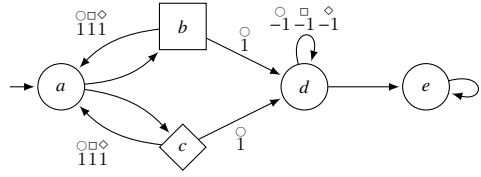

\begin{customprop}{\ref{pptn_energy_infinite_memory}}
		In the energy game presented in Figure~\ref{fig_energy_infinite_memory_bis}, there exists an SPE that makes player $\Box$ lose, but no finite memory SPE can achieve that result.
	\end{customprop}

\begin{proof}
	Consider the strategy profile $\bsigma$ defined by $\< \bsigma \> = acdde^\omega$, by $\< \bsigma_{\|hca} \> = abd^{\EL_\playcircle(hcabd)} e^\omega$ for every history $hca$, by $\< \bsigma_{\|hba} \> = acd^{\EL_\playcircle(hbacd)} e^\omega$ for every history $hba$, and by $\< \bsigma_{\|hd} \> = d^{\EL_\playcircle(hd)} e^\omega$ for every history $hd$.

    Intuitively: in every subgame, player $\Circle$ makes player $\Box$ and player $\Diamond$ lose.
    But to do so, she needs one of those players to cooperate with her: for example, in the main subgame $\Game_{\|a}$, she wants to make player $\Box$ lose, and to do so, she traverses the vertex $c$ to go to the vertex $d$.
    But then, player $\Diamond$ may deviate, and go back to the vertex $a$.
    Then, player $\Diamond$ has to be punished: and to do so, player $\Circle$ must go to the vertex $d$ through the vertex $b$, with player $\Box$'s cooperation\dots and so on.
    Once the vertex $d$ is reached (which will eventually be the case in every subgame), player $\Circle$'s energy level is equal to player $\Box$'s and player $\Diamond$'s one, plus 1.
    Thus, to make those players lose without losing herself, player $\Circle$ loops on the vertex $d$ exactly the right number of times, before going to the vertex $e$.

    Thus, the strategy profile $\bsigma$ is an SPE in which player $\Box$ loses --- but it is not a finite memory SPE.
    Let now $\btau$ be a finite memory SPE in $\Game_{\|a}$, compatible with a deterministic Mealy machine $\Mach$, and let us assume that $\mu_\Box(\< \btau \>) = 0$.
    Since, by hypothesis, we have $\mu_\Box(\< \btau \>) = 0$ (and therefore $\mu_\Diamond(\< \btau \>) = 0$, since players $\Box$ and $\Diamond$ always receive the same rewards), we also have, by induction and because $\btau$ is an SPE, the equality $\mu_\Box(\< \btau_{\|ha} \>) = \mu_\Diamond(\< \btau_{\|h a} \>) = 0$ for every history $h$ that is compatible with $\sigma_\playcircle$.
    
    Let now $n$ be the number of states of the machine $\Mach$, and let us consider a history $ha$ with $|h| = 2n$ (and therefore $\EL_\playcircle(ha) = \EL_\Box(ha) = \EL_\Diamond(ha) = n$).
    By the previous proposition, we know that $\mu_\Box(\< \btau_{\|ha} \>) = \mu_\Diamond(\< \btau_{\|h a} \>) = 0$, i.e. that the play $\< \btau_{\|ha} \>$ reaches the vertex $d$, and takes the edge $dd$ more than $n$ times.
    Since the machine $\Mach$ has only $n$ states, it means that that play actually loops on the vertex $d$ infinitely often, and is therefore also lost by player $\Circle$.
    Then, player $\Circle$ has a profitable deviation in that subgame, by going to the vertex $e$: contradiction.

    There exists an SPE that makes player $\Box$ lose in that game, but no finite-memory one. \hspace{1em plus 1fill}\qedhere
\end{proof}

\section{Proof of Theorem~\ref{thm_ds_checking}} \label{pf_ds_checking}

\begin{customthm}{\ref{thm_ds_checking}}
		In discounted-sum games, the Nash-checking and the subgame-perfect checking problems, deterministic or not, can be decided in polynomial time.
	\end{customthm}
	
	\begin{proof}
        By Corollary~\ref{cor_privilege_problem}, those four problems reduce to the privilege problem.
	
		Let $\Game_{\|v_0}$ be a discounted-sum game.
		Let $r$ be the reward function defined by $r(uv) = r_\EE(uv) - r_\AA(uv)$ for each edge $uv$.
        Then, a play $\pi$ in $\Game_{\|v_0}$ satisfies $\mu_\EE(\pi) > \mu_\AA(\pi)$ if and only if it satisfies $\DS^\lambda_r(\pi) > 0$.
        
        If such a play exists, then there exists one that has the form $hc^\omega$, with $|h|, |c| \leq n$, where $n$ is the number of vertices in $\Game$.
        Indeed, let $\pi$ be a play with $\DS^\lambda_r(\pi) > 0$.
        Let $k$ be the smaller index such that the vertex $\pi_k$ appears again in $\pi_{>k}$, and let $\l$ be the index of the second occurrence of $\pi_k$ along $\pi$.
        If we have $\DS^\lambda_r\left(\pi_{\leq k} (\pi_k \dots \pi_\l)^\omega\right) > 0$, then we can take $h = \pi_{\leq k}$ and $c = \pi_{k+1} \dots \pi_\l$.
        Otherwise, we can remove the sequence of $\pi_k \dots \pi_\l$ and start again with the play $\pi' = \pi_0 \dots \pi_k \pi_{\l+1} \dots$
        The procedure is necessarily finite since $\DS^\lambda_r(\pi) > 0$.

        Therefore, such a play exists if and only if there exists a vertex $v$ that can be reached from $v_0$ in $k \leq n$ steps with discounted sum $x$, and that can be reached from itself in $\l \leq n$ steps with discounted sum $y$, with $x$ and $y$ such that:
        $$x + \frac{y \lambda^k}{1 - \lambda^\l} > 0.$$
        The existence of such a vertex can be checked in polynomial time by Bellman-Ford's algorithm.
	\end{proof}

 \section{Proof of Theorem~\ref{thm_ds_verif_easiness}} \label{pf_ds_verif_easiness}

	\begin{customthm}{\ref{thm_ds_verif_easiness}}
		In discounted-sum games, the Nash rational and the subgame-perfect rational verification problems, deterministic or not, are recursively enumerable.
	\end{customthm}
	
	\begin{proof}
	    We present here a semi-algorithm that recognizes positive instances of the subgame-perfect universal (resp. Nash universal) threshold problem.
     The result will follow by Corollary~\ref{cor_reductions}.
	
		Given a discounted-sum game $\Game_{\|v_0}$ with discount factor $\lambda$, a player $i$ and a threshold $t$, we give an algorithm that stops if and only if there exists no SPE (resp. NE) $\bsigma$ in $\Game_{\|v_0}$ with $\mu_i(\< \bsigma \>) < t$.
		But first, let us give a preliminary result that will justify the correctness of our result.
		
		\begin{itemize}
			\item {\em Preliminary result: For every play $\pi$, each player $j$ and every index $n$, we have:
				$$\mu_j(\pi) \in \left[ \DS_j\left(\pi_{\leq n}\right) - M \lambda^n,    \DS_j\left(\pi_{\leq n}\right) + M \lambda^n \right],$$
				where:
    $$M = \frac{1}{1 - \lambda} \max_{uv \in E} ~ \max_{j \in \Pi} ~ |r_j(uv)|$$
    is a bound on the payoff (in absolute value) of every player.}
			
			Let us proceed by induction on $n$.
			When $n = 0$, for every play $\pi$ starting from $v_0$ and each player $j$, we have:
			\begin{align*}
				\mu_j(\pi) & = \sum_k r_j(\pi_k \pi_{k+1}) \lambda^k \\
				& \leq \sum_k \lambda^k \max_{uv \in E} \max_{j \in \Pi} |r_j(uv)| = M
			\end{align*}
			and symetrically $\mu_j(\pi) \geq -M$, which is the desired interval since $\DS_i(\pi_0) = 0$.
			
			Now, if the desired property is true for $n \geq 0$, let us prove that it is true for $n + 1$.
			Let $v \pi$ be a play.
			By induction hypothesis, we have $\mu_j(\pi) \in \left[ \DS_j(\pi_{\leq n}) - M \lambda^n, \DS_j(\pi_{\leq n}) + M \lambda^n \right].$
			Hence:
			\begin{align*}
				\mu_j(v\pi) & = r_j(v\pi_0) + \sum_k r_j(\pi_k \pi_{k+1})\lambda^{k+1} \\
				& = r_j(v\pi_0) + \lambda \mu_j(\pi) \\
				& \in \left[ r_j(v\pi_0) + \lambda \DS_j(\pi_{\leq n}) - M \lambda^{n+1}\right.,\\
				& ~~~ \left. r_j(v\pi_0) + \lambda \DS_j(\pi_{\leq n}) + M \lambda^{n+1} \right] \\
				& = \left[ \DS_j(v_0\pi_{\leq n}) - M \lambda^{n+1}, \DS_j(v_0\pi_{\leq n}) + M \lambda^{n+1} \right],
			\end{align*}
			as desired.
			
			\item \emph{Algorithm.}
			
			Let $\left(h^{(n)}\right)_{n \in \NN}$ be a recursive enumeration of the nonempty histories in $\Game_{\|v_0}$ by increasing order of lengths (we can, for example, order histories of the same length with the lexicographic order induced by some arbitrary order on vertices).
			
			Now, let $T$ be the infinite tree whose nodes of depth $n+1$ are all possible $n$-uples $\left(\bsigma(h^{(0)}), \dots, \bsigma(h^{(n)})\right)$, and where the children of the node $\left(\bsigma(h^{(0)}), \dots, \bsigma(h^{(n)})\right)$ are the nodes of the form $\left(\bsigma(h^{(0)}), \dots, \bsigma(h^{(n)}), \bsigma(h^{(n+1)})\right)$.
			Thus, every node partially defines a complete strategy profile $\bsigma$, and every infinite branch entirely defines it --- and every complete strategy profile is defined by an infinite branch.
			
			A node $\left(\bsigma(h^{(0)}), \dots, \bsigma(h^{(n)})\right)$ is called \emph{subgame-irrational} if there exist two indices $\l, m \leq n$, with $|h^{(\l)}| = |h^{(m)}| = p$, an index $k \leq p$ and a player $j$ such that:
			\begin{itemize}
				\item we have $h^{(\l)}_{\leq k} = h^{(m)}_{\leq k}$;
				
				\item the history $h^{(\l)}_{\geq k}$ is compatible with the strategy profile $\bsigma_{\|h^{(\l)}_{\leq k}}$, as partially defined there;
				
				\item the history $h^{(m)}_{\geq k}$ is compatible with the strategy profile $\bsigma_{-j\|h^{(m)}_{\leq k}}$;
				
				\item and finally, we have:
				$$\DS_j\left(h^{(m)}_{\geq k}\right) - \DS_j\left(h^{(\l)}_{\geq k}\right) > 2M \lambda^{p-k-1}.$$
			\end{itemize}
			
			The same node is called \emph{Nash-irrational} if, moreover, the history $h^{(\l)}_{\leq k}$ is compatible with $\bsigma$, as defined so far in the branch.
			
			The same node is called \emph{off-topic} if we have:
			$$\DS_i\left(h\right) > t - M \lambda^{|h|-1},$$
   where $h$ is the longest history that is compatible with $\bsigma$ as defined so far.
			
			Our algorithm consists in constructing the tree $T'$, obtained from the tree $T$ by cutting every branch after the first subgame-irrational (resp. irrational) or off-topic node, and terminating once that construction is finished.

			\item \emph{Correctness.}
			
			Let us consider a complete strategy profile $\bsigma$.
			Using the preliminary result, we know that we have $\mu_i(\< \bsigma \>) \leq t$ if and only if the corresponding branch contains no off-topic node.
			
			Let us now assume that that branch contains a subgame-irrational (resp. Nash-irrational) node $\left(\bsigma(h^{(0)}), \dots, \bsigma(h^{(n)})\right)$.
			Let us use the notations $k, \l, m, p$ and $j$ from the definition of subgame-irrationality (resp. Nash-irrationality).
			Let us define $\pi = \< \bsigma_{\|h^{(\l)}_{\leq k}} \>$: note that the history $h^{(\l)}_{\geq k}$ is a prefix of length $p-k$ of the play $\pi$.
			Similarly, let us extend the history $h^{(m)}_{\geq k}$ into a play $\pi'$, compatible with the strategy profile $\bsigma_{-j\|h^{(m)}_{\leq k}}$.
			By the preliminary result, the inequality:
			$$\DS_j\left(h^{(m)}_{\geq k}\right) - \DS_j\left(h^{(\l)}_{\geq k}\right) > 2 M \lambda^{p-k-1}$$
			implies $\mu_j(\pi') > \mu_j(\pi)$, and therefore the strategy profile $\bsigma$ is not an SPE (resp. not an NE, since the history $h^{(\l)}_{\leq k}$ is compatible with $\bsigma$).
			
			Conversely, if $\bsigma$ is not an SPE (resp. NE), then there exists a history $hv$ (resp. a history $hv$ compatible with $\bsigma$), a player $j$ and a strategy $\sigma'_j$ such that we have:
			$$\mu_j\left(\< \bsigma_{\|hv} \>\right) < \mu_j\left(\< \bsigma_{-j\|hv}, \sigma'_{j\|hv} \>\right).$$
			
			Then, let $\pi = \< \bsigma_{\|hv} \>$, and let $\pi' = \< \bsigma_{-j\|hv}, \sigma'_{j\|hv} \>$: since $\mu_j(\pi) < \mu_j(\pi')$, by the preliminary result, there exists an index $q$ such that:
			$$\DS_j\left(\pi_{\leq q}\right) + M \lambda^q < \DS_j\left(\pi'_{\leq q}\right) - M \lambda^q.$$
			Let now $\l$ and $m$ be the indices such that $h^{(\l)} = h\pi_{\leq q}$, and $h^{(m)} = h\pi'_{\leq q}$.
			Let $p = |h^{(\l)}| = |h^{(m)}|$, and let $k = |h|$.
			Then, we have $\DS_j\left(h^{(m)}_{\geq k}\right) - \DS_j\left(h^{(\l)}_{\geq k}\right) > 2M \lambda^{p-k-1}$: along the branch corresponding to $\bsigma$, the node of depth $\max\{\l, m\}$ is subgame-irrational (resp. Nash-irrational).
			
			As a consequence, the game $\Game_{\|v_0}$, the player $i$ and the threshold $t$ form a positive instance of the subgame-perfect universal (resp. Nash universal) threshold problem if and only if every branch of the tree $T$ contains either a subgame-irrational (resp. a Nash-irrational) or an off-topic node; i.e., by K\H{o}nig's Lemma, if and only if the tree $T'$ is finite; i.e., if and only if our algorithm terminates. \hspace{1em plus 1fill}\qedhere
		\end{itemize}
	\end{proof}

\section{Proof of Theorem~\ref{thm_mp_verif}} \label{pf_mp_verif}

\begin{customthm}{\ref{thm_mp_verif}}
		In the class of mean-payoff games, the Nash rational and the subgame-perfect rational verification problems, deterministic or not, are $\coNP$-complete.
	\end{customthm}
	
	\begin{proof}
	\begin{itemize}
        \item \emph{Lower bounds}

        The $\NP$-hardness of Nash and subgame-perfect universal threshold problems can be shown by the same proof as for parity games: given a formula $\phi$, the game $\Game^\phi$, defined in the proof of Theorem~\ref{thm_parity_verif}, can be transformed into a mean-payoff game as follows: when the color of a vertex $v$ for a given player $i$ is $2$, the reward granted to player $i$ on each edge leading to $v$ is defined as equal to $1$.
     When that color is $1$, that reward is defined as equal to $0$.
     Thus, we can prove that, again, Witness gets a payoff greater than $0$ in every SPE in that game if and only if $\phi$ is not satisfiable.
 
	    \item \emph{Subgame-perfect rational verification.}
	    
	    By~\cite{Icalp}, the complement of the universal threshold problem is $\NP$-easy.     
	The $\coNP$-completeness of the rational verification problem comes by Corollary~\ref{cor_reductions}.

		\item \emph{Nash rational verification}
		
		By~\cite{Icalp}, Theorem~38, given a mean-payoff game $\Game$, a player $i$, a threshold $t$, and a vertex labelling $\lambda: V \to \QQ \cup \{\pm \infty\}$, deciding the existence of a play $\pi$ such that $\mu_i(\pi) \leq t$, and such that $\mu_j(\pi) \geq \lambda(\pi_k)$ for each player $j$ and every vertex $\pi_k \in V_j$, is $\NP$-easy.
        Let us define $\lambda_1$ that maps each vertex $v$ to the value $\inf_{\bsigma_{-j}} \sup_{\sigma_j} \mu_j(\< \bsigma_{-j}, \sigma_j \>)$, where $j$ is the player controlling $v$.
        
        Since memoryless strategies are optimal in two-player zero-sum mean-payoff games (see for instance~\cite{DBLP:journals/tcs/BjorklundSV04}), the value $\lambda_1(v)$ is the payoff of player $j$ in a cycle of $\Game$, and can therefore be written with a number of bits that is polynomial in $\lv \Game \rv$.
        Therefore, there exists a polynomial $P$ such that $\lambda_1$ can be encoded using at most $P(\lv \Game \rv)$ bits.
        
        Thus, a non-deterministic polynomial algorithm that recognizes negative instances of the Nash universal threshold problem consists in guessing:
        \begin{itemize}
            \item a requirement $\lambda$, of size at most $P(\lv \Game \rv)$,

            \item a memoryless strategy profile $\btau^v_{-j}$ from $v$ for each $j$ and $v \in V_j$,
            
            \item and a certificate of the existence of a play $\pi$ satisfying the aforementioned properties,
	       \end{itemize}
        and then in checking, deterministically, that that certificate is valid and that each strategy profile $\tau^v_{-j}$ is such that for each strategy $\tau_j$, we have $\mu_j(\< \btau^v_{-j}, \tau_j\>) \leq \lambda(v)$, i.e. that $\lambda \geq \lambda_1$ --- and therefore, by Lemma~\ref{lm_ne}, that $\pi$ is an NE outcome. \hspace{1em plus 1fill}\qedhere
    \end{itemize}
\end{proof}

\section{Proof of Proposition~\ref{prop_ach_verif}} \label{pf_ach_verif}

\begin{customprop}{\ref{prop_ach_verif}}
	Let $\Cl$ be a class of games, among the classes of energy games and discounted-sum games.
	    Let $\rho \in \{\Nash, \subgameperfect\}$.
	    Then, the positive instances of the achaotic $\rho$-rational verification problem in $\Cl$ are exactly the positive instances of the $\rho$-rational verification problem.
    Similarly, the positive instances of the achaotic Nash-rational verification problem in mean-payoff games are exactly the positive instances of the $\rho$-rational verification problem.
\end{customprop}
	
\begin{proof}
	This proposition is a consequence of the following fact: in both of the settings given above, in every game $\Game_{\|v_0}$ of the studied class, for every Mealy machine $\Mach$ for Leader, there exists a strategy $\sigma_{\LL} \in \Comp_{\|v_0}(\Mach)$ that admits a $\rho$-response.
	
	Indeed, by Theorem~\ref{thm_product_game}, there exists a strategy $\sigma_{\LL} \in \Comp_{\|v_0}(\Mach)$ that admits a $\rho$-response if and only if there exists a $\rho$-equilibrium in $\Game_{\|v_0} \otimes \Mach$.

    \begin{itemize}
        \item In the first case, it is always true: in the two mentioned classes, SPEs are always guaranteed to exist --- and therefore so are NEs.
        Indeed, if $\Cl$ is the class of energy games, then $\Game_{\|v_0}$, and therefore $\Game_{\|v_0} \otimes \Mach$, is a Boolean game with Borel winning conditions.
	    By~\cite{GU08}, every such game contains an SPE.
        If $\Cl$ is the class of discounted-sum games, then $\Game_{\|v_0}$, and therefore $\Game_{\|v_0} \otimes \Mach$, is a game with payoff functions that are continuous for the canonical distance on infinite words.
	    By~\cite{DBLP:conf/csl/BrihayeBMR15}, every such game contains an SPE.

        \item In the second case, it is also always true as well: SPEs may not exist in mean-payoff games, but NEs always do --- see for example~\cite{DBLP:conf/lfcs/BrihayePS13}.
    \end{itemize}

    Therefore, in all those settings, if the game $\Game_{\|v_0}$ and the Mealy machine $\Mach$ are such that $\mu_\LL(\< \bsigma \>) > t$ for every $\sigma_\LL \in \Comp_{\|v_0}(\Mach)$ and all $\bsigma_{-\LL} \in \rho\R(\sigma_\LL)$, i.e. if we have a positive instance of the $\rho$-rational verification problem, then $\epsilon = 0$ is such that $0\rho\R(\sigma_\LL) \neq \emptyset$ for some $\sigma_\LL \in \Comp_{\|v_0}(\Mach)$, and simultaneously such that $\mu_\LL(\< \bsigma \>) > t$ for every $\sigma_\LL \in \Comp_{\|v_0}(\Mach)$ and all $\bsigma_{-\LL} \in 0\rho\R(\sigma_\LL)$ --- i.e., we have a positive instance of the achaotic $\rho$-rational verification problem.
    
    Conversely, if we have a positive instance of the achaotic $\rho$-rational verification problem, i.e. if there exists $\epsilon \geq 0$ such that $\epsilon\rho\R(\sigma_\LL) \neq \emptyset$ for some $\sigma_\LL \in \Comp_{\|v_0}(\Mach)$, and simultaneously such that $\mu_\LL(\< \bsigma \>) > t$ for every $\sigma_\LL \in \Comp_{\|v_0}(\Mach)$ and all $\bsigma_{-\LL} \in \epsilon\rho\R(\sigma_\LL)$, then it is also the case of $0$, since $0\rho$-responses are $\epsilon\rho$-responses; hence we have $\mu_\LL(\< \bsigma \>) > t$ for every $\sigma_\LL \in \Comp_{\|v_0}(\Mach)$ and all $\bsigma_{-\LL} \in \rho\R(\sigma_\LL)$, i.e. we have a positive instance of the $\rho$-rational verification problem.
\end{proof}

\section{Proof of Lemma~\ref{lm_epsilon}} \label{pf_epsilon}

\begin{customlm}{\ref{lm_epsilon}}
		There exists a polynomial $P_1$ such that in every mean-payoff game $\Game_{\|v_0}$, there exists $\epsilon_{\min}$ with $\lv \epsilon_{\min} \rv \leq P_1(\lv \Game \rv)$ such that $\epsilon_{\min}$-SPEs exist in $\Game_{\|v_0}$, and $\epsilon$-SPEs, for every $\epsilon < \epsilon_{\min}$, do not.
	\end{customlm}
	
	\begin{proof}
		First, let us show that $\epsilon_{\min}$ exists.
		According to~\cite{Concur}, $\epsilon$-SPEs are characterized by a mapping called \emph{negotiation function}, that maps each \emph{requirement}, i.e. each vertex labelling $\lambda: V \to \RR \cup \{\pm \infty\}$, to a requirement $\nego(\lambda)  \geq \lambda$.
		The same article defines \emph{$\lambda$-consistent plays} as plays $\pi$ satisfying $\mu_i(\pi) \geq \lambda(v)$ for each player $i$ and each vertex $v \in V_i$, and shows that each play is a $\delta$-SPE outcome if and only if it is $\lambda$-consistent for some requirement $\lambda$ that is a {\em $\delta$-fixed point} of the negotation function, i.e. that satisfies $\nego(\lambda)(v) \leq \lambda(v) + \delta$ for each $v$.
		Moreover, it proves that for every $\delta \geq 0$, we can define a least $\delta$-fixed point $\lambda_\delta$ of the negotiation function, which satisfies $\lambda_\delta(v_0) \neq +\infty$ if and only if $\delta$-SPEs exist.
		
		Let us now define $\epsilon = \inf\{\delta \geq 0 ~|~ \text{$\delta$-SPEs in $\Game_{\|v_0}$ exist}\}$.
		We only need to show that $\epsilon$ is also such that $\epsilon$-SPEs exist, or in other words, that $\lambda_\epsilon(v_0) \neq +\infty$.
		Let us note that for every $\delta, \delta'$ with $\delta \leq \delta'$, the requirement $\lambda_\delta$ is a $\delta'$-fixed point of the negotiation function, and therefore satisfies $\lambda_\delta \geq \lambda_{\delta'}$.
		
		Let us now define the requirement $\lambda: v \mapsto \sup_{\delta > \epsilon} \lambda_\delta(v)$.
		For each $v$ and every $\delta > \epsilon$, we have $\nego(\lambda_\delta)(v) \leq \lambda_\delta(v) - \delta$; since according to~\cite{Concur}, the negotiation function is Scott-continuous, we have $\nego(\lambda)(v) \leq \lambda(v) - \epsilon$, i.e. $\lambda$ is an $\epsilon$-fixed point of the negotiation function, and therefore $\lambda = \lambda_\epsilon$.
		Since $\lambda(v_0) = \sup_\delta \lambda_\delta(v_0) \neq +\infty$, we obtain that $\epsilon$-SPEs exist in $\Game_{\|v_0}$, and therefore that $\epsilon_{\min} = \epsilon$.
		
		Now, it has been shown in~\cite{Icalp}, in the proof of Theorem~1, that one can define for every $\lambda$ a finite union of polyhedra $X_\lambda \subseteq \RR^{V \times \Pi}$, where each of those polyhedra is defined by a set of inequations either of the form $x_{vi} \geq \lambda(w)$ or of size bounded by a polynomial function of $\lv \Game \rv$, and such that for each $i \in \Pi$ and $v \in V_i$, we have $\nego(\lambda)(v) = \min\left\{x_{vi} ~\left|~ \bbx \in X_\lambda\right.\right\}$.
		
		Let us now define the set $Y = \left\{\left(\lambda, \bbx\right) ~\left|~ \bbx \in X_\lambda\right.\right\} \subseteq \RR^V \times \RR^{V \times \Pi}$.
		The set $Y$ is itself a finite union of polyhedra defined by the same inequations than $X_\lambda$, or by inequations of bounded size (those of the form $x_{vi} \geq \lambda(w)$, where $\lambda(w)$ is no longer a constant but a coordinate of $\left(\lambda, \bbx\right)$).
		Let us also define the mapping:
		$$f: \left\{ \begin{matrix}
			Y & \to & \RR \\
			\left(\lambda, \bbx\right) & \mapsto & \underset{i \in  \Pi, v \in V_i}{\max} \left(x_{vi} - \lambda(v)\right)
		\end{matrix} \right..$$
		Then, we have:
		$$\epsilon_{\min} = \min_{\left(\lambda, \bbx\right)} f\left(\lambda, \bbx\right),$$
		and since $f$ is a piecewise linear mapping on a finite union of polyhedra that has a minimum, that minimum is reached on a vertex $\left(\lambda, \bbx\right)$ of one of those polyhedra.
		By Corollary~1 of~\cite{Icalp}, such a vertex has coordinates bounded by a polynomial of the maximal size of the inequations defining the polyhedron to which it belongs, i.e. by a polynomial of $\lv \Game \rv$.
		Therefore, that is also the case of $f(\lambda, \bbx) = \epsilon_{\min}$.
	\end{proof}

\section{Proof of Theorem~\ref{thm_mp_ach_verif}} \label{pf_mp_ach_verif}

			\begin{figure*}
				\centering
                \newcommand{\deltay}{0.7}
				\begin{tikzpicture}[->,>=latex,shorten >=1pt, initial text={}, scale=0.8, every node/.style={scale=0.57}]
					\node[state] (?x1) at (0, 0*\deltay) {$?x_1$};
					\node (x1) at (2, 1.5*\deltay) {\dots};
					\node (nx1) at (2, -1.5*\deltay) {\dots};
					\node[state] (?xi) at (4, 0*\deltay) {$?x_i$};
					\node[state] (xi) at (6, 1.5*\deltay) {$x_i$};
					\node[state] (nxi) at (6, -1.5*\deltay) {$\neg x_i$};
					\node (?xi+1) at (8, 0*\deltay) {\dots};
					\node[state] (xn) at (10, 1.5*\deltay) {$x_n$};
					\node[state] (nxn) at (10, -1.5*\deltay) {$\neg x_n$};
					\node[state] (C1) at (12, 0*\deltay) {$C_1$};
					\node (C2) at (14, 0*\deltay) {\dots};
					\node[state] (Cp) at (16, 0*\deltay) {$C_p$};
					\node[state, initial left] (a) at (0, -4*\deltay) {$a$};
					\node[state] (b) at (2, -4*\deltay) {$b$};
					\node[state] (c) at (4, -4*\deltay) {$c$};
					\node[state] (m1) at (6, -3.5*\deltay) {$\blacktriangledown$};
					\node[state, dashed] (m2) at (6, 3.5*\deltay) {$\blacktriangledown$};
					\node[state, dashed] (m3) at (14, -1.5*\deltay) {$\blacktriangledown$};
					
					\node[blue] (A) at (0, -4.6*\deltay) {$\AA$};
					\node[blue] (B) at (2, -3.4*\deltay) {$\BB$};
					\node[blue] (S1) at (-0.4, -0.4*\deltay) {$\SS$};
					\node[blue] (S2) at (3.6, -0.4*\deltay) {$\SS$};
					\node[blue] (pxi) at (6, 0.9*\deltay) {$x_i$};
					\node[blue] (pnxi) at (6, -0.9*\deltay) {$\neg x_i$};
					\node[blue] (pxn) at (10, 0.9*\deltay) {$x_n$};
					\node[blue] (pnxn) at (10, -0.9*\deltay) {$\neg x_n$};
					\node[blue] (pC1) at (12, -0.6*\deltay) {$C_1$};
					\node[blue] (pCp) at (16, -0.6*\deltay) {$C_p$};

					\path (a) edge[bend left] node[above] {$\stack{\AA}{0}~\stack{\BB}{3}~\stack{\WW}{1}$} (b);
					\path (b) edge[bend left] node[below] {$\stack{\AA}{0}~\stack{\BB}{3}~\stack{\WW}{1}$} (a);
					\path (b) edge (c);
					\path (c) edge[loop above] node[above] {$\stack{\AA}{2}~\stack{\BB}{2}~\stack{\WW}{1}$} (c);
					\path (a) edge (?x1);
					\path (?x1) edge (x1);
					\path (?x1) edge (nx1);
					\path (?xi) edge node[above left] {$\stack{x_i}{2m}~\stack{C}{m}~\stack{\AA}{2 - \frac{m}{2^{i+1}}}$} (xi);
					\path (?xi) edge node[below left] {$\stack{\neg x_i}{2m}~\stack{C}{m}~\stack{\AA}{2}$} (nxi);
					\path (xi) edge node[above right] {$\stack{x_i}{2m}~\stack{C}{m}~\stack{\AA}{2 - \frac{m}{2^{i+1}}}$} (?xi+1);
					\path (nxi) edge node[below right] {$\stack{\neg x_i}{2m}~\stack{C}{m}~\stack{\AA}{2}$} (?xi+1);
					\path (xn) edge node[above right] {$\stack{x_n}{2m}~\stack{C}{m}~\stack{\AA}{2 - \frac{m}{2^{n+1}}}$} (C1);
					\path (nxn) edge node[below right] {$\stack{\neg x_n}{2m}~\stack{C}{m}~\stack{\AA}{2}~\stack{\WW}{1}$} (C1);
					\path (C1) edge node[above] {$\stack{\AA}{2}$} (C2);
					\path (C2) edge node[above] {$\stack{\AA}{2}$} (Cp);
					\path (Cp) edge[bend right=40] node[above] {$\stack{\AA}{2}$} (?x1);
					\path (nxi) edge (m1);
					\path (nxn) edge[bend right=15] (m1);
					\path (xi) edge (m2);
					\path (xn) edge[bend left=15] (m2);
					\path (C1) edge[bend left=15] (m3);
					\path (Cp) edge[bend right=15] (m3);
					\path (m1) edge[loop below] node[below] {$\stack{\AA}{1}~\stack{x_1, \neg x_1, \dots, x_n, \neg x_n}{4}~ \stack{C_1, \dots, C_p}{2}~\stack{\WW}{1}$} (m1);
				\end{tikzpicture}
				\caption{The game $\Game_{\|a}$}
				\label{fig_G_pnp_bis}
			\end{figure*}

\begin{customthm}{\ref{thm_mp_ach_verif}}
		In the class of mean-payoff games, the achaotic subgame-perfect rational verification problem, deterministic or not, is $\Poly^\NP$-complete.
	\end{customthm}

\begin{proof}
		Using Lemma~\ref{lm_epsilon} and the same arguments as in the proof of Theorem~\ref{thm_product_game}, those two problems reduce to the following problem, and conversely: given a game $\Game_{\|v_0}$ and a threshold $t \in \QQ$, does every $\epsilon_{\min}$-SPE $\bsigma$ in $\Game_{\|v_0}$ satisfy $\mu_\LL(\< \bsigma \>) > t$?
		Let us prove that that problem is $\Poly^\NP$-complete.
		
		\begin{itemize}
			\item \emph{Easiness.}
			
			By~\cite{Icalp}, there is an $\NP$ algorithm that decides, given $\epsilon$ and $\Game_{\|v_0}$, whether there exists an $\epsilon$-SPE in $\Game_{\|v_0}$, i.e. whether $\epsilon \geq \epsilon_{\min}$.
			Using Lemma~\ref{lm_epsilon}, a dichotomous search can therefore compute $\epsilon_{\min}$ by a polynomial number of calls to that algorithm.
			Then, one last call to that same algorithm can decide whether there exists an $\epsilon_{\min}$-SPE $\bsigma$ such that $\mu_i(\< \bsigma \>) \leq t$.

			\item \emph{Hardness.}
			
			We proceed by reduction from the following $\Poly^\NP$-complete problem: given a Boolean formula $\phi$ in conjunctive normal form over the ordered variables $x_1, \dots, x_n$, is the lexicographically first valuation $\nu_{\min}$ satisfying $\phi$ such that $\nu_{\min}(x_n) = 1$? (and in particular, does such a valuation exist?)
			Let us write $\phi = \bigwedge_{j=1}^p C_j$.
			We construct a game $\Game_{\|a}$, with a player called \emph{Witness} and written $\WW$, in which there exists an $\epsilon_{\min}$-SPE $\bsigma$ such that $\mu_\WW(\< \bsigma\>) \leq 0$ if and only if $\phi$ is satisfiable and $\nu_{\min}(x_n) = 1$.
			That game, depicted in Figure~\ref{fig_G_pnp_bis}, has $2n + p + 4$ players: the literal players $x_1, \neg x_1, \dots, x_n, \neg x_n$; the clause players $C_1, \dots, C_p$; the player \emph{Solver}, written $\SS$; the player \emph{Witness}, written $\WW$; the player \emph{Alice}, written $\AA$; and the player \emph{Bob}, written $\BB$.
			It contains $3n + p + 4$ vertices:
			\begin{itemize}
				\item the initial vertex $v_0 = a$, controlled by Alice;
				
				\item two vertices $b$ and $c$, controlled by Bob;
				
				\item for each variable $x_i$, a vertex $?x_i \in V_\SS$ , a vertex $x_i \in V_{x_i}$, and a vertex $\neg x_i \in V_{\neg x_i}$;
				
				\item for each clause $C_j$, a vertex $C_j \in V_{C_j}$;
				
				\item a sink vertex $\blacktriangledown$ (drawn three times in Figure~\ref{fig_G_pnp} for convenience).
			\end{itemize}
			
			which are connected by the following edges (unmentioned rewards are equal to $0$, and we write $m = 2n+p$):
			\begin{itemize}
				\item from the vertex $a$ to the vertex $b$ and from the vertex $b$ to the vertex $a$, two edges that give Alice the reward $0$, Bob the reward $3$, and Witness the reward $1$;
				
				\item from $a$ to $?x_1$ and from $b$ to $c$, an edge;
				
				\item from $c$ to itself, an edge giving both Alice and Bob the reward $2$, and giving Witness the reward $1$;
				
				\item from each $?x_i$ to $\neg x_i$ and from $\neg x_i$ to $?x_{i+1}$ (or to $C_1$ if $i = n$), an edge giving:
				\begin{itemize}
					\item the reward $2m$ to $\neg x_i$,
					\item the reward $m$ to every player $C_j$ such that the clause $C_j$ contains the literal $\neg x_i$,
					\item the reward $2$ to Alice;
					\item and if $i = n$, the reward $1$ to Witness;
				\end{itemize}
				
				\item from each $?x_i$ to $x_i$ and from $x_i$ to $?x_{i+1}$ (or to $C_1$ if $i = n$), an edge giving:
				\begin{itemize}
					\item the reward $2m$ to $x_i$,
					\item the reward $m$ to every player $C_j$ such that the clause $C_j$ contains the literal $x_i$,
					\item and the reward $2 - \frac{m}{2^{i+1}}$ to Alice;
				\end{itemize}
				
				\item from each $C_j$ to $C_{j+1}$ (or $?x_1$ if $j = p$), an edge giving the reward $2$ to Alice;
				
				\item from the sink vertex $\blacktriangledown$ to itself, an edge giving the reward $1$ to Alice, the reward $2$ to each clause player, the reward $4$ to each literal player, and $1$ to Witness.
			\end{itemize}

Let us now present the correspondence between $\epsilon$-SPEs and valuations satisfying $\phi$.

\begin{itemize}
    \item \emph{The strategy profile $\bsigma_\nu$.}

    Let $\nu$ be a valuation satisfying $\phi$.
    We define the memoryless strategy profile $\bsigma_\nu$ as follows: from the vertex $a$, Alice always goes to ${?x_1}$; from the vertex $b$, Bob always goes to $c$; from each vertex ${?x_i}$, Solver always goes to $x_i$ if $\nu(x_i) = 1$, and to $\neg x_i$ otherwise; from their vertices, the clause players and the literal players that are satisfied by $\nu$ do never go to the vertex $\blacktriangledown$; and the literal players that are not satisfied by $\nu$ do whenever they have the opportunity.
					
	Such a strategy profile is an $\epsilon$-SPE, where $\epsilon = \sum_{i=1}^n \frac{\nu(x_i)}{2^i}$.
	Indeed, Solver's payoff is constant, and Witness does not control any vertex, hence they have no profitable deviation.
	In every subgame where they have actions to choose, all the clause players get at least the payoff $2$, since at least one of their literals is satisfied (and therefore the corresponding vertex is visited at each turn).
	Similarly, from the vertex he controls, each literal player gets the payoff $4$, either because he is satisfied by $\nu$ or by going to the vertex $\blacktriangledown$; they have therefore no profitable deviation.
	As for Bob, in every subgame where he has an action to choose, he gets the payoff $2$, and cannot get a better one, since Alice always plans to go to $?x_1$.
	Let us now focus on Alice.
	From the vertex $a$, the only one that she controls, she could get the payoff $2$ by going to $b$.
	By going to the vertex $?x_1$, she gets the payoff:
	\begin{gather*}
	\frac{1}{2n+p} \left(\sum_{i=1}^n 2\left( 2 - (2n+p) \frac{\nu(x_i)}{2^{i+1}} \right) + 2p \right) \\
	= 2 - \sum_{i=1}^n \frac{\nu(x_i)}{2^i} = 2 - \epsilon,
	\end{gather*}
	hence $\bsigma_\nu$ is an $\epsilon$-SPE, and is not a $\delta$-SPE for any $\delta < \epsilon$.

    \item \emph{The valuation $\nu_{\bsigma}$.}

    Let $\bsigma$ be a $1$-SPE such that the play $\pi = \< \bsigma \>$ traverses the vertex $?x_1$, but does never reach the vertex $\blacktriangledown$.
    We define the valuation $\nu_{\bsigma}$ by $\nu_{\bsigma}(x_i) = 1$ if and only if $\pi$ traverses the vertex $x_i$.
    
	Let us note that for every vertex $x_i$ that is traversed by $\pi$, the player $x_i$ can deviate and go to the vertex $\blacktriangledown$, where he can get the payoff $4$.
	Since $\bsigma$ is a $1$-SPE, we have therefore $\mu_{x_i}(\pi) \geq 3$, and therefore $\mu_{\neg x_i}(\pi) \leq 1$.
    As a consequence, the vertices $x_i$ and $\neg x_i$ cannot both be traversed; and if the vertex $\neg x_i$ is traversed, then we have $\nu_{\bsigma}(x_i) = 0$.
				
	Let us prove that the valuation $\nu_{\bsigma}$ satisfies $\phi$.
	Let us consider a clause $C_j$ of $\phi$.
	Since the vertex $C_j$ is visited (infinitely often), there is a deviation available for player $C_j$ that grants her the payoff $2$.
	Since $\bsigma$ is a $1$-SPE, we have $\mu_{C_j}(\pi) \geq 1$.
	Therefore, there is at least one literal $\l$ of $C_j$ such that the vertex $\l$ is visited by $\pi$, i.e. such that $\nu_{\bsigma}$ satisfies $\l$.
    Consequently, the valuation $\nu_{\bsigma}$ satisfies $C_j$, and therefore $\phi$.

    Moreover, since the play $\pi$ does never traverse both $x_i$ and $\neg x_i$ for any $i$, it grants Alice the payoff:
    \begin{gather*}
	\frac{1}{2n+p} \left(\sum_{i=1}^n 2\left( 2 - (2n+p) \frac{\nu_{\bsigma}(x_i)}{2^{i+1}} \right) + 2p \right) \\
	= 2 - \sum_{i=1}^n \frac{\nu_{\bsigma}(x_i)}{2^i}.
	\end{gather*}
\end{itemize}

Therefore, if $\nu_{\min}$ is the least valuation satisfying $\phi$ (and in particular if such a valuation exists), then we have $\epsilon_{\min} = \sum_{i=1}^n \frac{\nu_{\min}(x_i)}{2^i}$.
Indeed, the strategy profile $\bsigma_{\nu_{\min}}$ is a $\sum_{i=1}^n \frac{\nu_{\min}(x_i)}{2^i}$-SPE, and we can prove that there is no $\epsilon$-SPE with $\epsilon < \sum_{i=1}^n \frac{\nu_{\min}(x_i)}{2^i}$: if $\bsigma$ is an $\epsilon$-SPE with $\epsilon < 1$, then necessarily, there exist infinitely many integers $k$ such that $\sigma_\AA((ab)^ka) = {?x_1}$, because otherwise each play $\< \bsigma_{\|(ab)^ka} \>$ would either end in the vertex $c$ (and then Bob would have a deviation profitable by more than $1$ by refusing to go to $c$), or would be equal to $(ab)^\omega$ (and then Alice would have a deviation profitable by more than $1$ by going to $?x_1$).
Then, for each such $k \neq 0$, we have $\sigma_\BB((ab)^k) = c$ (if Bob does not go to $c$, then he gets only the payoff $1$, while by going to $c$, he gets the payoff $2$).
Therefore, if we choose some $k$ such that $\sigma_\AA((ab)^ka) = {?x_1}$, we also have $\mu_\AA(\< \bsigma_{\|(ab)^kab} \>) = 2$.
Since $\bsigma$ is an $\epsilon$-SPE with $\epsilon < 1$, we have therefore $\mu_\AA(\< \bsigma_{\|(ab)^ka} \>) > 1$, which means that the play $\< \bsigma_{\|(ab)^ka} \>$ does never reach the vertex $\blacktriangledown$.
Then, the strategy profile $\bsigma_{\|(ab)^ka}$ is a $1$-SPE such that the play $\< \bsigma_{\|(ab)^ka} \>$ traverses the vertex $?x_1$, but does never reach the vertex $\blacktriangledown$: the valuation $\nu_{\bsigma_{\|(ab)^ka}}$ is defined, satisfies the formula $\phi$, and is such that:
$$\mu_\AA(\< \bsigma_{\|(ab)^ka} \>) = 2 - \sum_{i=1}^n \frac{\nu_{\bsigma_{\|(ab)^ka}}(x_i)}{2^i}.$$
By going to the vertex $b$ instead of $?x_1$, Alice has therefore a deviation that is profitable by $\sum_{i=1}^n \frac{\nu_{\bsigma_{\|(ab)^ka}}(x_i)}{2^i}$.
Moreover, since the valuation $\nu_{\bsigma_{\|(ab)^ka}}$ satisfies $\phi$, it is lexicographically greater than or equal to $\nu_{\min}$, hence the inequality:
$$\sum_{i=1}^n \frac{\nu_{\bsigma_{\|(ab)^ka}}(x_i)}{2^i} \geq \sum_{i=1}^n \frac{\nu_{\min}(x_i)}{2^i},$$
and therefore $\epsilon \geq \sum_{i=1}^n \frac{\nu_{\min}(x_i)}{2^i}$, as desired.
			
We can now conclude that there exists an $\epsilon_{\min}$-SPE $\bsigma$ in this game such that $\mu_\WW(\< \bsigma \>) \leq 0$ if and only if $\nu_{\min}(x_n) = 1$:

\begin{itemize}
    \item if there exists such an $\epsilon_{\min}$-SPE, then necessarily the play $\< \bsigma \>$ visits infinitely often the vertex $x_n$ (every other play gives Witness a positive payoff).
    Therefore, the valuation $\nu_{\min} = \nu_{\bsigma}$ is such that $\nu_{\min}(x_n) = 1$.

    \item Conversely, if $\nu_{\min}(x_n) = 1$, then the $\epsilon_{\min}$-SPE $\bsigma_{\nu_{\min}}$ traverses the vertex $?x_1$ and does never traverse $\neg x_n$, nor reach the vertex $\blacktriangledown$; it grants therefore Witness the payoff $0$. \hspace{1em plus 1fill}\qedhere
\end{itemize}
		\end{itemize}
	\end{proof}

\end{document}